\documentclass[11pt,a4paper]{article}
\usepackage[T1]{fontenc}
\usepackage[utf8]{inputenc}
\usepackage{lmodern}
\usepackage[margin=27mm,headheight=14pt]{geometry}
\usepackage{amsmath,amssymb,amsthm,mathtools}
\usepackage{microtype}
\usepackage{enumitem,booktabs,array}
\usepackage{xcolor}
\usepackage{tikz}
\usetikzlibrary{arrows.meta,decorations.pathreplacing}
\usepackage{fancyhdr}
\usepackage[unicode,colorlinks=true,linkcolor=blue!45!black,citecolor=blue!45!black,urlcolor=blue!45!black]{hyperref}
\hypersetup{pdftitle={Bounded Integer Programming in Almost N\textasciicircum{}(2d) Time},pdfauthor={}}
\setlist{itemsep=2pt,topsep=4pt}
\allowdisplaybreaks[2]
\numberwithin{equation}{section}
\newtheorem{theorem}{Theorem}[section]
\newtheorem{lemma}[theorem]{Lemma}
\newtheorem{proposition}[theorem]{Proposition}

\theoremstyle{definition}
\newtheorem{definition}[theorem]{Definition}
\theoremstyle{remark}
\newtheorem{remark}[theorem]{Remark}
\newcommand{\R}{\mathbb R}
\newcommand{\Z}{\mathbb Z}
\newcommand{\eps}{\varepsilon}

\newcommand{\cW}{\mathcal W}
\newcommand{\cD}{\mathcal D}
\newcommand{\cT}{\mathcal T}
\newcommand{\cR}{\mathcal R}
\newcommand{\Athin}{\mathcal A_{\rm geo}}
\newcommand{\Alift}{\mathcal A_{\rm lift}}
\newcommand{\Bx}{\mathsf B}
\newcommand{\Pref}{P_{\rm ref}}
\newcommand{\T}{\mathsf T}
\newcommand{\va}{\mathbf a}\newcommand{\vb}{\mathbf b}\newcommand{\vc}{\mathbf c}
\newcommand{\ve}{\mathbf e}\newcommand{\vo}{\mathbf o}\newcommand{\vp}{\mathbf p}
\newcommand{\vq}{\mathbf q}\newcommand{\vt}{\mathbf t}\newcommand{\vu}{\mathbf u}
\newcommand{\vv}{\mathbf v}\newcommand{\vw}{\mathbf w}\newcommand{\vx}{\mathbf x}
\newcommand{\vy}{\mathbf y}\newcommand{\vz}{\mathbf z}\newcommand{\vY}{\mathbf Y}
\newcommand{\vpi}{\boldsymbol\pi}\newcommand{\vbeta}{\boldsymbol\beta}
\newcommand{\vlambda}{\boldsymbol\lambda}\newcommand{\vpsi}{\boldsymbol\psi}
\newcommand{\vdelta}{\boldsymbol\delta}\newcommand{\vell}{\boldsymbol\ell}
\newcommand{\vzero}{\mathbf 0}
\DeclareMathOperator{\lat}{lat}
\DeclareMathOperator{\conv}{conv}
\DeclareMathOperator{\covol}{covol}
\DeclareMathOperator{\epi}{epi}
\DeclareMathOperator{\rank}{rank}
\DeclareMathOperator{\vol}{vol}
\DeclareMathOperator{\relint}{relint}
\DeclareMathOperator{\proj}{proj}
\title{Bounded Integer Programming in Almost \texorpdfstring{$N^{2d}$}{N\^{}(2d)} Time}
\author{Lin Chen, Yuchen Mao, Guochuan Zhang\footnote{\{chenlin198662,maoyc,zgc\}@zju.edu.cn. Zhejiang University.}}
\date{\today}
\begin{document}
\maketitle
\begin{abstract}
We study integer programs in standard form
$\min\{\vc^{\T}\vx:A\vx=\vb,\ \vell\le\vx\le\vu,\ \vx\in\Z^n\}$, where
$A\in\Z^{d\times n}$, all entries of $\vb,\vc,\vell,\vu$ are integers, and
every variable bound is finite.
Let $N$ be an upper bound on the absolute values of the entries of $A$,
and $c_{\max}$ one on the absolute values of the costs.
For every fixed $d$ and every fixed $\eps>0$ we give a deterministic
algorithm running in time
\[
 T_{\rm LP}+O_d(n)+O_{d,\eps}\bigl(N^{2d+\eps}\log^2(c_{\max}+2)\bigr),
\]
where $T_{\rm LP}$ is the time to solve one linear programming relaxation
exactly. Prior to our work, the best known running time for bounded
variables was $N^{O(d^2)}$ [Eisenbrand \& Weismantel'20], and it has
been a fundamental open problem whether the exponent $O(d^2)$ can be
improved to one linear in $d$. We answer this question affirmatively,
with the exponent $2d+\eps$.
\end{abstract}

\section*{AI Statement} The main proof route was designed by the authors, building upon their recent work on subset sum and integer programming feasibility \cite{CMZ26}. The authors came up with the basic idea of adapting the algorithmic technique for integer programming feasibility test to optimization, and
GPT Astra Ultra was used to turn the insights of the authors into complete mathematical proofs.  GPT and Claude were  used for drafting the paper. GPT also verified this paper’s main results using the Lean 4 proof. The authors take full responsibility for
the entire content, technical claims, and integrity of the paper. 

\newpage 

\tableofcontents

\clearpage

%======================================================================
\section{Introduction}\label{sec:intro}

We study the integer program
\begin{equation}\label{eq:main}
 \min\{\vc^{\T}\vx:\ A\vx=\vb,\ \vell\le \vx\le \vu,\ \vx\in\Z^n\},
\end{equation}
in standard form with finite variable bounds, where $A\in\Z^{d\times n}$
and all entries of $\vb,\vc,\vell,\vu$ are
integers. The number $d$ counts only the equations $A\vx=\vb$, not the
variable bounds, which are finite and binary encoded. Put
\[
 N=\max\{2,\max_{i,j}|A_{ij}|\},\qquad c_{\max}=\max_i|c_i|
\]
($c_{\max}=0$ for empty input). An algorithm must report infeasibility or
return an optimal integer vector. Integer programming is NP-hard in
general~\cite{Karp72}, already for $d=1$ when $N$ is part of the input, so
the question is how the running time depends on the structural parameters
$d$ and $N$, while the targets, the bounds, and the costs may be
arbitrarily large. (In the literature the number of constraints is usually
denoted by $m$ and the largest entry by $\Delta$; we use $d$ and $N$.)

\subsection{Previous work}

\paragraph{Few variables.}
When the number $n$ of variables is fixed, integer programming is
solvable in polynomial time by Lenstra's algorithm~\cite{Lenstra83}.
Kannan improved the dependence on the dimension to $n^{O(n)}$
times a polynomial in the input length~\cite{Kannan87}, and Reis and
Rothvoss reduced it to $(\log n)^{O(n)}$~\cite{ReisRothvoss23}. These
algorithms underlie the point-or-thin primitive used in this
paper (Lemma~\ref{lem:point-thin})~\cite{Lenstra83}.

\paragraph{Few constraints and small coefficients.}
For programs in standard form with $d$ equations and nonnegative
variables, Papadimitriou gave a pseudopolynomial algorithm for fixed
$d$~\cite{Papadimitriou81}. Eisenbrand and Weismantel used the Steinitz
lemma to obtain running time $(dN)^{O(d)}\|\vb\|_\infty^2$ for
$\min\{\vc^{\T}\vx:A\vx=\vb,\ \vx\ge\vzero,\ \vx\in\Z^n\}$, and proved that an
optimal LP vertex has an optimal integer solution within $\ell_1$-distance
$d(2dN+1)^d$, independently of $n$, $\vb$, and the objective~\cite{EW19};
proximity results of this kind go back to Cook, Gerards, Schrijver, and
Tardos~\cite{CGST86}. Jansen and Rohwedder sharpened the dependence on
$d$ and $N$ for nonnegative variables to $O(\sqrt d\,N)^{2d}$ times
lower-order factors, and showed that improving the exponent $2d$
substantially would contradict the $(\min,+)$-convolution
hypothesis~\cite{JansenRohwedder19}. Under the Exponential Time Hypothesis,
Fomin, Panolan, Ramanujan, and Saurabh~\cite{FPRS23} and Knop, Pilipczuk,
and Wrochna~\cite{KPW20} showed that an exponent of order $d$ is necessary,
up to lower-order factors. Thus running times of the form
$(dN)^{O(d)}$ are the natural target.

\paragraph{Upper bounds.}
Finite upper bounds on the variables make the problem harder to handle by
these methods. For the bounded version, Eisenbrand and Weismantel obtained
running time $n\cdot(dN)^{O(d^2)}$~\cite{EW19}, by combining proximity with
a dynamic program over a bounded region; to our knowledge this quadratic
dependence on $d$ in the exponent had remained the best general bound for
bounded optimization. The case $d=1$ contains subset sum and bounded
knapsack; see Section~\ref{sec:further-related}.

\paragraph{Feasibility.}
Chen, Mao, and Zhang recently obtained near-optimal algorithms for subset
sum and for testing feasibility of integer programs~\cite{CMZ26}. A key
ingredient is a constructive structure theorem: after removing few
exceptional items, the remaining items admit common coordinates in a
generalized arithmetic progression, together with a small reserve whose
subset sums fill a box (Theorem~\ref{thm:gap}). Feasibility tests are
insensitive to costs, and a reserve that repairs a weight error may be
far too expensive for an optimization query. Extending the approach to
optimization is the subject of this paper.

\subsection{Our results}

Our main result is an algorithm for \eqref{eq:main} whose running time
has the form $N^{2d+\eps}$ for every fixed $\eps>0$, up to the cost of
reading the input and solving one linear program.

\begin{theorem}[Bounded integer programming]\label{thm:main}
For each fixed $d$ and each fixed $\eps>0$, problem \eqref{eq:main} can
be solved deterministically in time
\begin{equation}\label{eq:main-time}
 T_{\rm LP}+O_d(n)+O_{d,\eps}\bigl(N^{2d+\eps}\log^2(c_{\max}+2)\bigr).
\end{equation}
The algorithm returns an optimal integer vector or correctly reports
infeasibility. Repeated columns, with possibly different costs and
bounds, are allowed. The same statement holds when the $d$ equations are
replaced by $d$ inequalities or by $d$ two-sided constraints
$\vb^-\le A\vx\le\vb^+$, in addition to the finite variable bounds.
\end{theorem}

Here $T_{\rm LP}$ is the time to decide feasibility of the linear
programming relaxation of \eqref{eq:main} and, when it is feasible, to
compute an exact optimal vertex and an optimal dual price; it is
polynomial in the input length. We use the word RAM with words of at
least $\log_2(nN(c_{\max}+2))$ bits, so that every column entry and every
cost fits in a word, and unit-cost arithmetic on words. The term
$O_d(n)$ counts arithmetic operations on the input numbers and on
numbers of $O_d(B_{\rm in})$ bits derived from them, where $B_{\rm in}$
is the bit length of one input number (Section~\ref{sec:bits}); all
later operations are on numbers of a fixed number of words. In the bit model, reading the input,
arithmetic on the original data, and writing the answer keep their
encoding-length costs; Section~\ref{sec:bits} makes this precise. The
numerical sizes of $\vb$, $\vell$, and $\vu$ enter only the first two
terms. The implicit constant in the third term depends on $d$ and
$\eps$, but the exponent $2$ of $\log(c_{\max}+2)$ does not: one factor
comes from the binary search over the optimal value, the other from the
iteration counts of the rounding steps applied to cost-dependent bodies
(Remark~\ref{rem:log}). All powers of $\log n$ and $\log N$ are absorbed,
the former because $n$ enters only through a kernel with $O_d(N^{2d})$
items, the latter into $N^{\eps}$.

For $d=1,2,3$ the exponents of $N$ are $2+\eps$, $4+\eps$, and $6+\eps$.
Compared with the $(dN)^{O(d^2)}$ bound for bounded
variables~\cite{EW19}, the quadratic dependence on $d$ in the exponent
is removed. The exponent $2d+\eps$ also matches, up to $\eps$, the
exponent $2d$ of the algorithm of Jansen and Rohwedder for nonnegative
variables~\cite{JansenRohwedder19}, which cannot be improved
substantially under the $(\min,+)$-convolution hypothesis. We do not
remove the factor $N^\eps$; Remark~\ref{rem:eps} explains where it is
used.

The proof has three parts.
\begin{enumerate}
\item \emph{A compact signed kernel} (Section~\ref{sec:kernel}). One
exact LP and the Steinitz lemma reduce \eqref{eq:main} to a binary
problem with $n_0=O_d(N^{2d})$ labelled items, columns bounded by $N$,
costs bounded by $c_{\max}$, and a target of norm at most $dN$. An
optimal vertex of the original LP is inherited as an optimal vertex of
the kernel LP.
\item \emph{Exact strict decisions}
(Sections~\ref{sec:prelim}--\ref{sec:lift}). The binary problem is
solved through strict queries by a recursion that measures interior in
the epigraph of the fractional cost function. Each unresolved query
yields a smaller item hull, a smaller rank, a proper sublattice, or one
additional thin integer coordinate.
\item \emph{Accounting} (Sections~\ref{sec:implementation}
and~\ref{sec:complexity}). A node with buffered lattice mass $w$ and
normalized volume $\Psi$ is processed in time
$\widetilde O_{d,\eps}(N^{2\delta}w\Psi)\cdot\log(c_{\max}+2)$ for a
small $\delta>0$, where $\widetilde O$ hides powers of $\log N$ only, by
grouping items of equal column into types, by a strongly polynomial
fixed-dimensional LP oracle, and by output-sensitive geometric oracles. A potential that is linear in $w\Psi$ and weighted by
the remaining arithmetic depth bounds the total work by the root value
$w_0\Psi_0=\widetilde O_{d,\eps}(N^{2d+d\delta})$.
\end{enumerate}

The algorithm works with strict decision queries
\begin{equation}\label{eq:strict}
 A\vx=\vt,\qquad \vc^{\T}\vx<s,\qquad \vx\in\{0,1\}^n;
\end{equation}
the bound $\vc^{\T}\vx\le q$ is the query with $s=q+1$, and binary
search over $q$ gives optimization. Costs are never used as an
additional geometric coordinate: they are compared only through an LP
price, and exact costs are stored as values attached to weights.

\paragraph{Our approach.}
The algorithm follows the principle of the feasibility framework: solve a
target that lies deep inside the region of fractional solutions by
rounding and repairing with the reserve, and reduce the problem when a
boundary obstructs rounding. For optimization, \emph{depth} is measured in
the epigraph of the fractional minimum-cost function rather than in the
weight space alone, so that it accounts for the cost of moving the target.
A target with little affordable room yields subproblems whose items span
a thinner region; a deep target either forces a proper canonical
sublattice, when no residual query keeps enough of its cost allowance, or
leads to an affordable fibre of the common progression, which yields a
solution or a new thin integer coordinate. A single potential, combining
a lattice-point mass, a normalized volume, the rank, the number of
added coordinates, and the arithmetic depth, bounds the whole recursion.
Finite bounds are handled by one LP and a Steinitz argument, which leave
a binary kernel with few items; items with equal columns are kept as
labelled copies and are processed as one \emph{type}, using a multiset
version of the common-progression theorem.
Section~\ref{sec:technical-overview} explains these ideas informally,
with worked examples.

\subsection{Further related work}\label{sec:further-related}

Subset sum is the special case of \eqref{eq:main} with $d=1$, binary
variables, and zero costs: given positive integers $a_1,\dots,a_n$ and a
target $t$, decide whether some subset sums to $t$. It is NP-hard, and
its pseudopolynomial algorithms have been studied intensively, with
running times measured either in the target $t$ or in the largest item
$w_{\max}=\max_ia_i$, which plays the role of $N$.

\paragraph{Running time in terms of the target.}
Bellman's dynamic program solves subset sum in time
$O(nt)$~\cite{Bellman57}. Koiliaris and Xu gave a deterministic
$\widetilde O(n+\sqrt n\,t)$ algorithm using sumset
computations~\cite{KoiliarisXu19}, and Bringmann gave a randomized
algorithm running in $\widetilde O(n+t)$~\cite{Bringmann17}, which was
recently derandomized by Chan~\cite{Chan26}; Jin and Wu found a simpler
randomized $\widetilde O(n+t)$ algorithm based on formal power
series~\cite{JinWu19}. This is essentially optimal: Abboud,
Bringmann, Hermelin, and Shabtay showed that an algorithm running in
time $t^{1-\varepsilon}2^{o(n)}$ for some $\varepsilon>0$ would refute the
Strong Exponential Time Hypothesis~\cite{ABHS22}; in particular, an
$\widetilde O(n+w_{\max}^{1-\varepsilon})$ algorithm is ruled out under
this hypothesis.

\paragraph{Running time in terms of the largest item.}
When the items are small compared with the target, the relevant
parameter is $w_{\max}$. Pisinger showed that subset sum can be solved in
time $O(nw_{\max})$, linear in the number of items~\cite{Pisinger99}.
Polak, Rohwedder, and W\k{e}grzycki combined proximity arguments with
additive combinatorics to obtain $\widetilde O(n+w_{\max}^{5/3})$, and
gave an $\widetilde O(n+w_{\max}^{3})$ algorithm for $0$-$1$
knapsack~\cite{PolakRW21}. Chen, Lian, Mao, and Zhang improved subset sum
to $\widetilde O(n+w_{\max}^{3/2})$ by fine-grained proximity results,
which also apply to bounded multiplicities~\cite{CLMZ24}, and, for
distinct items, to $\widetilde O(n+w_{\max}^{5/4})$; without the
distinctness assumption they obtained
$\widetilde O(n+\sqrt{w_{\max}t})$~\cite{CLMZ24c}. Very recently, Jin and
Mao gave a deterministic algorithm running in
$O(n+w_{\max}^{1+o(1)})$ time, which is optimal up to the
$w_{\max}^{o(1)}$ factor under the Strong Exponential Time
Hypothesis~\cite{JinMao26}. Their algorithm also solves the
single-constraint bounded feasibility problem $\sum_ia_ix_i=b$,
$x_i\in\{0,\dots,u_i\}$, with $a_i\in[-w_{\max},w_{\max}]$. It models
almost all items as elements of a generalized arithmetic progression
using the inverse theorem of Conlon, Fox, and Pham~\cite{CFP23}, so that
every lattice point deep inside the zonotope of the progression
coordinates is a subset sum, and computes the remaining sums near the
boundary by a quotient version of economic cap covering with Macbeath
regions. Independently and concurrently, Chen, Mao, and Zhang obtained
the same running time and extended their method to multidimensional
subset sum and integer programming feasibility~\cite{CMZ26}; their
structure theorem is the starting point of this paper. Macbeath regions
also play a central role here, but in the epigraph of the fractional
cost function rather than in the zonotope
(Section~\ref{sec:buffers}).

\paragraph{Dense instances and additive structure.}
A parallel line of work exploits the additive structure of dense
instances, where the number of items is large compared with $w_{\max}$.
Galil and Margalit showed, using results on arithmetic progressions in
subset sums, that dense subset sum can be solved in almost linear
time~\cite{GalilMargalit91}, and Bringmann and Wellnitz characterized
the targets reachable in dense instances and obtained near-linear
algorithms in this regime~\cite{BringmannWellnitz21}. The common
progressions with filling reserves used in this paper
(Theorem~\ref{thm:gap}) are a constructive, multidimensional analogue of
this structure.

\paragraph{Optimization versus feasibility.}
For optimization the picture is different. Under the
$(\min,+)$-convolution hypothesis, $0$-$1$ knapsack has no algorithm with
running time $O((n+t)^{2-\varepsilon})$~\cite{CMWW19}, whereas subset sum
is solvable in $\widetilde O(n+t)$. In terms of the largest weight, $0$-$1$
knapsack can be solved in $\widetilde O(n+w_{\max}^2)$ time by the
algorithms of Jin and of Bringmann~\cite{Jin24,Bringmann24}. Costs thus
add genuine difficulty beyond feasibility; in our setting this
difficulty appears as the need to measure depth in the epigraph of the
fractional cost function rather than in the weight space.

\section{Technical overview}\label{sec:technical-overview}

\emph{Optimization follows the same guiding principle as feasibility:
solve an interior target by rounding, and reduce the problem when a
boundary obstructs rounding.} The extra issue is that a feasible repair
can be too expensive. We therefore measure interior using both weight and
cost. We use \emph{interior case} and \emph{boundary case} for the two
complementary outcomes: enough room for affordable movement in every
item direction, or insufficient room in at least one direction.

There is one further distinction to keep in mind. A target can have ample
room in the current weight space while the additive structure has a thin direction
that those coordinates do not show. In that event we expose the missing
direction and recurse. Thus the interior branch either finds an affordable
solution, restricts the problem to a proper lattice, or discovers a hidden
boundary. The last two outcomes reduce the problem rather than ending the
algorithm.

As in \eqref{eq:main}, we minimize cost. Section~\ref{sec:kernel} reduces
finite bounds to the binary case: given columns $\mathbf{a}_i\in\Z^d$
(possibly repeated) with coordinates bounded by $N$ and integer costs
$c_i$, we seek
\[
 \min\left\{\sum_i c_ix_i:\ \sum_i \mathbf{a}_ix_i=\mathbf{t},\quad
                              x_i\in\{0,1\}\right\}.
\]
For a decision query, $s\in\Z$ is the given cost threshold: we ask
whether there exists a \emph{binary} selection with weight $\mathbf{t}$ and
cost strictly below $s$. (A non-strict integer bound
$\sum_ic_ix_i\le s'$ is the strict query with $s=s'+1$.) A logarithmic
number of threshold queries recovers an optimum. The goal is a running
time of the form
\[
 N^{2d+\eps}\cdot\log^2(c_{\max}+2),
 \qquad c_{\max}=\max_i|c_i|,
\]
for every fixed $\eps>0$, together with input-reading costs. Large
costs should enter through their bit length, not the size of a dynamic
programming array, and only through a fixed power of it.

\paragraph{Notation.}
As throughout the paper, vectors are bold and their scalar components are
in ordinary type (Table~\ref{tab:notation}). For equally sized vectors,
$\langle\mathbf{v},\mathbf{w}\rangle=\sum_jv_jw_j$ denotes their inner
product. We write price
vectors as rows, so $\boldsymbol{\pi}\mathbf{a}_i
=\langle\boldsymbol{\pi},\mathbf{a}_i\rangle
=\sum_{j=1}^d\pi_j(\mathbf{a}_i)_j$ is the total priced weight of item $i$.
In a one-dimensional example, the price $\pi$ is a scalar.

\paragraph{Roadmap.}
Section~\ref{ov:feasibility} first recalls the feasibility framework:
structured interior targets can be rounded and repaired, while boundary
targets lead to smaller instances. This is background for the optimization
algorithm. The remaining sections explain how cost changes that framework.
\begin{itemize}
\item \textbf{Price the repair}
(Section~\ref{ov:pricing}). An LP price separates the cost already needed to
reach the target from the extra cost of rounding. This tells us which
items can safely form the reserve.
\item \textbf{Apply the interior/boundary test}
(Section~\ref{ov:affordability}). In the boundary case, process the
few exceptional changes and shrink the remaining weight problem
(Section~\ref{ov:boundary}). In the interior case, construct cheap structure and
attempt the feasibility repair (Section~\ref{ov:interior}). Arithmetic
obstructions lead to a proper lattice; a thin direction hidden by the
current weights leads to a new coordinate. These are outcomes
\emph{within the interior branch}.
\item \textbf{Bound the whole recursion}
(Section~\ref{ov:accounting}). Both smaller weight regions and proper
lattices decrease one geometric measure. Shared accounting controls all
branches, including those created when a hidden coordinate is exposed.
\end{itemize}

Our structural starting point is the constructive common-progression and
filling-reserve framework of Chen, Mao, and Zhang~\cite{CMZ26}, in the
form stated as Theorem~\ref{thm:gap}. The overview explains how the
optimization argument uses it. We suppress scale choices and proof bookkeeping
to emphasize why each ingredient is needed.

\subsection{Feasibility background: interior is easy, boundary reduces}
\label{ov:feasibility}

The feasibility algorithm has two complementary responses to a target.
In the interior, enough surrounding subset sums allow us to round a
fractional solution and repair its error. Near a boundary, movement is
restricted in some direction. Relative to a suitable fractional reference,
only a controlled number of items large in that direction can change;
processing those choices exactly leaves a smaller weight problem.
The comparison is always at a scale set by the items. We explain the
interior repair below; Section~\ref{ov:affordability} will adapt both
responses to a cost constraint.

Allowing $0\le x_i\le1$ gives the fractional weight region
\[
 Z=\left\{\sum_i \mathbf{a}_ix_i:0\le x_i\le1\right\}.
\]
This is a sum of line segments, usually called a \emph{zonotope}.
An integer point in $Z$ need not be a subset sum. For example, even item
weights can produce a fractional sum of odd weight. Being far from the
boundary does not remove such arithmetic obstructions. We first need
coordinates that describe the subset sums, rather than only their convex
hull.

\paragraph{Common progression coordinates.}
Consider the four weights $11,12,22,24$. Their subset sums look like
separated clusters on the line. Represent the items instead by
$(1,0),(0,1),(2,0),(0,2)$, and recover the original weights by
\[
 \phi(x,y)=11x+12y.
\]
In these coordinates, the subset sums fill every integer point of
$[0,3]^2$. The two descriptions in Figure~\ref{fig:progression} are
equivalent, but the second makes the regularity visible.

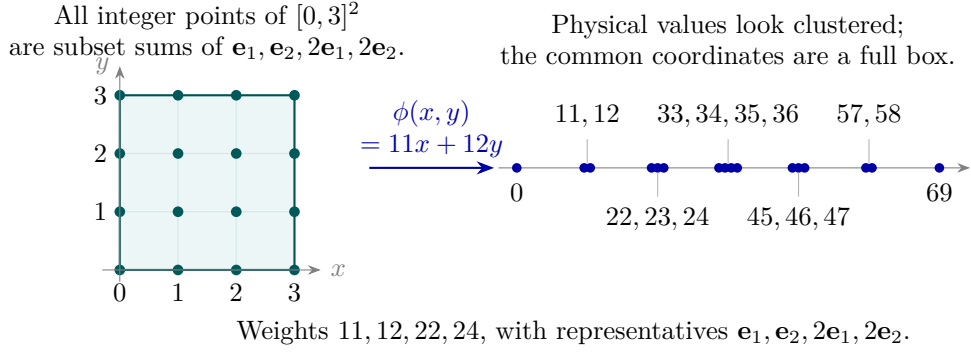
\begin{figure}[htbp]
\centering
\begin{tikzpicture}[x=1cm,y=1cm,>=Stealth,
  every node/.style={font=\small},
  ovblue/.style={blue!65!black},ovteal/.style={teal!70!black}]
  \begin{scope}[shift={(0.3,0.1)},x=.77cm,y=.77cm]
    \fill[teal!7] (0,0) rectangle (3,3);
    \draw[step=1,teal!18] (0,0) grid (3,3);
    \draw[teal!65!black,thick] (0,0) rectangle (3,3);
    \foreach \xx in {0,1,2,3} {
      \foreach \yy in {0,1,2,3} {
        \fill[teal!70!black] (\xx,\yy) circle (2pt);
      }
      \node[below] at (\xx,-.07) {$\xx$};
    }
    \foreach \yy in {1,2,3} \node[left] at (-.04,\yy) {$\yy$};
    \draw[->,gray] (-.3,0)--(3.45,0) node[right] {$x$};
    \draw[->,gray] (0,-.2)--(0,3.5) node[left] {$y$};
    \node[align=center] at (1.5,4.1)
      {All integer points of $[0,3]^2$\\are subset sums of $\mathbf{e}_1,\mathbf{e}_2,2\mathbf{e}_1,2\mathbf{e}_2$.};
  \end{scope}
  \draw[->,blue!65!black,thick] (3.6,1.45)--(5.25,1.45)
    node[midway,above,align=center] {$\phi(x,y)$\\$=11x+12y$};
  \begin{scope}[shift={(5.55,1.45)},x=.081cm,y=1cm]
    \draw[->,gray] (-3,0)--(74,0);
    \foreach \xx in {0,11,12,22,23,24,33,34,35,36,45,46,47,57,58,69}
      \fill[blue!65!black] (\xx,0) circle (1.8pt);
    \foreach \xx/\txt in {0/0,69/69}
      \node[below] at (\xx,-.08) {$\txt$};
    \draw[gray!60] (11.5,.06)--(11.5,.38);
    \node[above] at (11.5,.4) {$11,12$};
    \draw[gray!60] (23,-.06)--(23,-.38);
    \node[below] at (23,-.4) {$22,23,24$};
    \draw[gray!60] (34.5,.06)--(34.5,.38);
    \node[above] at (34.5,.4) {$33,34,35,36$};
    \draw[gray!60] (46,-.06)--(46,-.38);
    \node[below] at (46,-.4) {$45,46,47$};
    \draw[gray!60] (57.5,.06)--(57.5,.38);
    \node[above] at (57.5,.4) {$57,58$};
    \node[align=center] at (34.5,1.7)
      {Physical values look clustered;\\the common coordinates are a full box.};
  \end{scope}
  \node[align=center] at (6.3,-.72)
    {Weights $11,12,22,24$, with representatives
      $\mathbf{e}_1,\mathbf{e}_2,2\mathbf{e}_1,2\mathbf{e}_2$.};
\end{tikzpicture}
\caption{An exact common-progression example. The sixteen coordinate sums
fill the integer box, although their images in one dimension have gaps.
The map $\phi$ recovers every physical item and every subset sum.}
\label{fig:progression}
\end{figure}

A generalized arithmetic progression is the image of an integer box
under a linear map such as $\phi$. The structural construction finds one
common system of this kind for most of the items: each surviving item
$\mathbf{a}_i$ has an integer representative $\mathbf{u}_i$ in a bounded box, with
$\phi\mathbf{u}_i=\mathbf{a}_i$. The same map therefore recovers every subset sum. The
coordinate dimension can exceed the original weight dimension; its job
is to record additive information that the original coordinates hide.

The construction separates a small set of exceptional items, a small
\emph{reserve}, and a remaining \emph{bulk}. Here small means a small
fixed power of $N$, up to logarithmic and dimension-dependent factors.
Exceptional items are handled
exactly. The bulk and reserve have common progression coordinates, and
the reserve's subset sums fill a box large enough to correct rounding
errors. Here \emph{constructive} means that we compute these coordinates
and recover the actual selected reserve items whenever a correction is
used. A mere statement that many subset sums exist would not suffice.

\paragraph{What the reserve repairs.}
Suppose there are $k$ progression coordinates. Fix one reference subset
of the reserve. If an integer coordinate target can be reached by this
reference subset plus a fractional bulk selection, choose an LP vertex
for that representation. There are only $k$ coordinate equations, so at
most $k$ bulk variables are fractional.

Round those variables down. The lost vector is an integer vector, since
both the target and the rounded sum are integer. It is also small: it is
a sum of at most $k$ fractional item vectors. The reserve is designed
precisely for this scale. Its subset sums contain a whole integer box
around the reference sum, so we replace the reference subset by one that
supplies the lost vector. The rounded bulk and this replacement reserve
subset give the exact target. They use disjoint actual items.

For example, the two bulk vectors $(1,1)$ and $(-1,1)$ reach $(0,1)$ by
taking one half of each. Rounding down loses $(0,1)$. A disjoint reserve
containing $(0,1)$ repairs the error. Figure~\ref{fig:reserve} shows a
reserve with a stronger guarantee: its subset sums fill the entire box
$[-6,6]^2\cap\Z^2$, so it handles all errors of the required size, not
only this particular one.

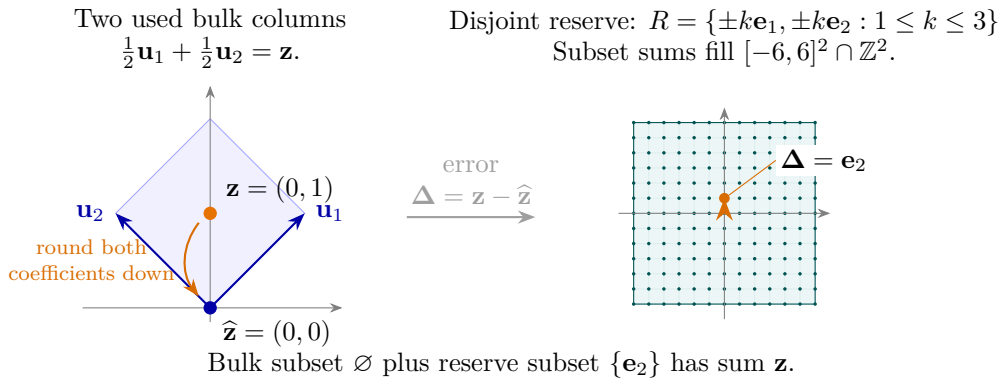
\begin{figure}[htbp]
\centering
\begin{tikzpicture}[x=1cm,y=1cm,>=Stealth,
  every node/.style={font=\small}]
  \begin{scope}[shift={(2.4,.2)},x=1.25cm,y=1.25cm]
    \fill[blue!6] (0,0)--(1,1)--(0,2)--(-1,1)--cycle;
    \draw[blue!35] (0,0)--(1,1)--(0,2)--(-1,1)--cycle;
    \draw[->,gray] (-1.35,0)--(1.4,0);
    \draw[->,gray] (0,-.15)--(0,2.35);
    \draw[->,blue!65!black,thick] (0,0)--(1,1)
      node[right] {$\mathbf{u}_1$};
    \draw[->,blue!65!black,thick] (0,0)--(-1,1)
      node[left] {$\mathbf{u}_2$};
    \fill[orange!90!black] (0,1) circle (2.5pt);
    \node[above right] at (.06,1) {$\mathbf{z}=(0,1)$};
    \fill[blue!65!black] (0,0) circle (2.5pt);
    \node[below right] at (.03,-.02) {$\widehat{\mathbf{z}}=(0,0)$};
    \draw[->,orange!85!black,thick,bend right=38]
      (-.12,.9) to node[left,align=center,font=\scriptsize]
      {round both\\coefficients down} (-.12,.12);
    \node[align=center] at (0,2.85)
      {Two used bulk columns\\
       $\frac12\mathbf{u}_1+\frac12\mathbf{u}_2=\mathbf{z}$.};
  \end{scope}
  \draw[->,gray!75,thick] (5,1.4)--(6.7,1.4)
    node[midway,above,align=center] {error\\$\boldsymbol{\Delta}=\mathbf{z}-\widehat{\mathbf{z}}$};
  \begin{scope}[shift={(9.2,1.45)},x=.2cm,y=.2cm]
    \fill[teal!7] (-6,-6) rectangle (6,6);
    \draw[step=1,teal!13] (-6,-6) grid (6,6);
    \draw[teal!65!black] (-6,-6) rectangle (6,6);
    \foreach \xx in {-6,...,6} \foreach \yy in {-6,...,6}
      \fill[teal!65!black] (\xx,\yy) circle (.65pt);
    \draw[->,gray] (-7,0)--(7,0);
    \draw[->,gray] (0,-7)--(0,7);
    \draw[->,orange!85!black,very thick] (0,0)--(0,1);
    \fill[orange!85!black] (0,1) circle (2pt);
    \draw[orange!85!black] (.3,1.2)--(3.4,3.5);
    \node[right,fill=white,inner sep=1.5pt] at (3.4,3.5) {$\boldsymbol{\Delta}=\mathbf{e}_2$};
    \node[align=center] at (0,11.6)
      {Disjoint reserve: $R=\{\pm k\mathbf{e}_1,\pm k\mathbf{e}_2:1\le k\le3\}$\\
       Subset sums fill $[-6,6]^2\cap\mathbb Z^2$.};
  \end{scope}
  \node[align=center] at (6.25,-.55)
    {Bulk subset $\varnothing$ plus reserve subset $\{\mathbf{e}_2\}$ has sum $\mathbf{z}$.};
\end{tikzpicture}
\caption{The reserve repairs a small rounding error. On the left,
$\mathbf{u}_1=(1,1)$ and $\mathbf{u}_2=(-1,1)$ reach $\mathbf{z}=(0,1)$ with coefficients
$1/2,1/2$. Rounding down gives $\widehat{\mathbf{z}}=0$; the disjoint reserve
supplies the missing vector. On the right, the twelve reserve items
$\{\pm k\mathbf{e}_1,\pm k\mathbf{e}_2:1\le k\le3\}$ fill every integer point of
$[-6,6]^2$ by subset sums, providing a whole region of possible repairs.
Here $\mathbf{e}_1=(1,0)$ and $\mathbf{e}_2=(0,1)$.}
\label{fig:reserve}
\end{figure}

This is the easy interior mechanism. The bulk supplies the large-scale
weight, while the reserve supplies a bounded correction. The reserve is
not supposed to realize the entire target by itself. For optimization,
the same repair must also fit below the cost threshold. Understanding
its \emph{extra} cost is the purpose of Section~\ref{ov:pricing}.

\subsection{LP pricing: separate the unavoidable cost from the repair}
\label{ov:pricing}

If the cost coefficients were already small, we could include the cost
constraint in the feasibility problem. Indeed, for integer costs,
$\sum_i c_ix_i<s$ is equivalent to
$\sum_i c_ix_i+y=s-1$ with a bounded nonnegative integer slack $y$.
This adds one equation whose coefficients are the costs, so a
feasibility method whose running time depends on the largest coefficient
would remain efficient only while the costs are small compared with
$N$. (We do not use this route.) The challenge is that the costs can be
arbitrarily larger than $N$: using them as another coordinate would make
the geometric problem arbitrarily large.

We therefore seek an equivalent way to measure cost that removes a large
common contribution. Assign a price to each unit of weight and subtract
the priced weight from every item cost. At a fixed target, the total
subtracted is the same for every feasible solution, so this changes
neither their ordering nor the optimal subsets. A suitable LP price will measure
only the extra cost of changing an individually preferred choice.

This does not promise that \emph{all} new costs are small. The algorithm
uses the changes that are cheap for the reserve and processes the
expensive changes separately. The following example first illustrates
the favorable situation: an arbitrarily large common cost disappears,
and the remaining rounding cost is just $1$.

\begin{samepage}
\paragraph{An equivalent objective exposes a cheap repair.}
Take three items:
\begin{center}
\begin{tabular}{lrrr}
\toprule
 & item 1 & item 2 & item 3\\
\midrule
Weight & 2 & 3 & 5\\
Cost & 199 & 301 & 500\\
\bottomrule
\end{tabular}
\end{center}
\end{samepage}
At target weight $5$, ask for a subset costing strictly less than $501$.
The fractional optimum takes item 1 and $3/5$ of
item 3, for cost $199+300=499$. Rounding down leaves only item 1.
Using item 2 as the repair supplies the missing weight $3$ and gives
total cost $199+301=500$. The repair item costs $301$, but the increase
over the LP optimum is only $1$: the LP had already paid $300$ for the
same missing weight. Figure~\ref{fig:lp-price} makes this comparison
directly.

% Embeddable TikZ figure. Required packages: amsmath, tikz, xcolor.
% Required TikZ library: arrows.meta.
\begin{figure}[htbp]
\centering
\begin{tikzpicture}[x=1cm,y=1cm,>=Stealth,
  every node/.style={font=\small,align=center}]
  \node[font=\small\bfseries] at (6.3,4.05)
    {Both rows have target weight $5$};

  % Widths encode the weight contributions: 2+3 in both rows.
  \node[anchor=east,align=right] at (2.55,3.13)
    {Fractional LP\\$\mathbf{x}=(1,0,3/5)$};
  \fill[blue!11] (2.85,2.7) rectangle (5.25,3.55);
  \draw[blue!65!black,thick] (2.85,2.7) rectangle (5.25,3.55);
  \node at (4.05,3.13) {weight $2$\\cost $199$};
  \fill[teal!9] (5.25,2.7) rectangle (8.85,3.55);
  \draw[teal!65!black,thick] (5.25,2.7) rectangle (8.85,3.55);
  \node at (7.05,3.13)
    {$3/5$ of weight-$5$ item\\weight $3$, cost $300$};
  \node[anchor=west,align=left] at (9.2,3.13)
    {Total cost $499$\\LP optimum};

  \node[anchor=east,align=right] at (2.55,1.73)
    {Binary repair\\$\mathbf{x}=(1,1,0)$};
  \fill[blue!11] (2.85,1.3) rectangle (5.25,2.15);
  \draw[blue!65!black,thick] (2.85,1.3) rectangle (5.25,2.15);
  \node at (4.05,1.73) {weight $2$\\cost $199$};
  \fill[orange!12] (5.25,1.3) rectangle (8.85,2.15);
  \draw[orange!80!black,thick] (5.25,1.3) rectangle (8.85,2.15);
  \node at (7.05,1.73)
    {reserve item\\weight $3$, cost $301$};
  \node[anchor=west,align=left] at (9.2,1.73)
    {Total cost $500<501$\\integral and affordable};

  \draw[->,orange!85!black,thick] (10.15,2.75)--(10.15,2.13)
    node[midway,right=5pt,font=\small\bfseries] {$+1$};

  \node[align=center] at (6.3,.5)
    {At the LP price $\pi^*=100$, the missing weight $3$ is priced at $300$.\\
     The reserve cost is $301=\pi^*\cdot3+1$: the repair adds only $1$.};
\end{tikzpicture}
\caption{An exact example of LP pricing with weights $(2,3,5)$,
costs $(199,301,500)$, target $5$, and strict allowance $501$.
The equal bar widths represent equal weight contributions. The price
$\pi^*=100$ per unit weight supports the LP optimum. The reserve pays
for weight that the fractional solution already budgeted, leaving a
surcharge of only $1$. This toy shows one repair, rather than the full
filled-reserve guarantee used by the algorithm.}
\label{fig:lp-price}
\end{figure}

An optimal LP price here is $\pi^*=100$ per unit of weight. Every selection
of weight $5$ pays the same priced amount $100\cdot5=500$. Subtracting
priced weight from each item cost leaves
\[
 199-100\cdot2=-1,\qquad
 301-100\cdot3=1,\qquad
 500-100\cdot5=0.
\]
The first item is better than the price predicts, the second is worse by
$1$, and the third agrees exactly. Thus we prefer to include item 1 and
omit item 2. Either choice for item 3 is equally good at this price.
For \emph{every} fractional or binary selection of weight $5$,
\[
 \text{cost}=500-x_1+x_2
            =499+(1-x_1)+x_2.
\]
The last two terms now tell us exactly how much more than the LP optimum
we pay: $1$ for omitting item 1 and $1$ for adding item 2. The repair above
adds item 2 and pays just $1$ extra. If the threshold is $501$, we may
spend strictly less than $501-499=2$ extra, so this repair succeeds.

The number $100$ is only for readability. Replace it by any large integer
$H$, and give the same three weights costs $2H-1,3H+1,5H$.
At target $5$, the common cost is $5H$, the fractional optimum is
$5H-1$, and the repaired subset costs $5H$. Subtracting the price $H$
again leaves $-1,1,0$ as the item costs. Thus the example addresses huge
costs, not just a small numerical instance. The reserve repairs weight
that the LP already supplied; only the difference between the actual
repair cost and that fractional payment consumes the room above the LP
optimum.

\paragraph{What an LP price does in general.}
Fix the decision query: find a binary selection of weight $\mathbf{t}$ and
original cost $\sum_i c_ix_i<s$, where $s$ is the given integer threshold.
To obtain a lower bound, first allow fractional choices and write
\[
 f(\mathbf{t})=\min\left\{\sum_i c_ix_i:\ \sum_i \mathbf{a}_ix_i=\mathbf{t},
                                      \ 0\le x_i\le1\right\}
\]
for their minimum cost, with value $+\infty$ if the fractional system
is infeasible. If $f(\mathbf{t})\ge s$, the binary query has no solution.
If $f(\mathbf{t})<s$, the query remains unresolved: a fractional solution can be
cheaper than every binary solution. The positive difference $s-f(\mathbf{t})$ is
the extra cost we may spend, strictly, above the LP optimum.

Start with any price vector $\boldsymbol{\pi}\in\R^d$.
The adjusted cost of item $i$ is $c_i-\boldsymbol{\pi}\mathbf{a}_i$, where the
product is the inner product defined above. Subtracting priced weight
preserves all cost comparisons at target $\mathbf{t}$: every feasible selection
has the same amount $\boldsymbol{\pi}\mathbf{t}$ subtracted from its total cost.
If we now temporarily drop the weight equation, the best adjusted choice
for each item contributes $\min\{0,c_i-\boldsymbol{\pi}\mathbf{a}_i\}$. Hence
\[
 \boldsymbol{\pi}\mathbf{t}+\sum_i\min\{0,c_i-\boldsymbol{\pi}\mathbf{a}_i\}
 \ \le\ f(\mathbf{t}).
\]
This is a lower bound for an arbitrary price. LP duality guarantees an
\emph{optimal price} $\boldsymbol{\pi}^{*}$ at the current target for
which the bound is attained:
\begin{equation}\label{ov:price}
 f(\mathbf{t})=\boldsymbol{\pi}^{*}\mathbf{t}
       +\sum_i\min\{0,c_i-\boldsymbol{\pi}^{*}\mathbf{a}_i\}.
\end{equation}
We use the star whenever this equality is needed. The price need not be
unique. Geometrically it is a supporting slope: for every feasible
displacement $\mathbf{v}$,
\[
 f(\mathbf{t}+\mathbf{v})\ge f(\mathbf{t})+\boldsymbol{\pi}^{*}\mathbf{v}.
\]
Thus $\boldsymbol{\pi}^{*}\mathbf{v}$ predicts the unavoidable linear part of
the cost of moving the target; the actual increase may be larger.

\paragraph{Preferred choices give a reference selection.}
Fix $\boldsymbol{\pi}^{*}$ and again set aside the weight equation for
a moment. Item $i$ has two binary choices. Omitting it contributes $0$
to the adjusted objective; including it contributes
$c_i-\boldsymbol{\pi}^{*}\mathbf{a}_i$. Choose the cheaper of these two options:
include the item when the adjusted cost is negative, and omit it when
the adjusted cost is positive. At a tie, use omission as a convention.
These are the \emph{preferred choices}; the word \emph{endpoint} simply
refers to choosing $0$ or $1$ in the interval $[0,1]$.

Doing this independently for all items produces a reference subset,
which need not have weight $\mathbf{t}$. In the three-item example,
$\pi^*=100$ gives adjusted costs $-1,1,0$, so the reference subset
contains only item 1 and has weight $2$, rather than $5$.
Because the price is optimal, an optimal fractional solution can keep
every non-tied preferred choice and adjust only tied items to reach
the target. Here it uses $3/5$ of tied item 3. This fractional completion
does not yet provide the binary solution we seek.

To obtain a binary solution, we may have to change some preferred
choices. Removing an item preferred for inclusion loses its adjusted
saving; adding an item preferred for omission incurs its adjusted cost.
Either change costs exactly $|c_i-\boldsymbol{\pi}^{*}\mathbf{a}_i|\ge0$.
For any binary subset of the required weight, equation~\eqref{ov:price}
therefore gives
\begin{equation}\label{ov:extra}
 \text{original cost}
 =f(\mathbf{t})+\sum_{\text{changed items }i}
                         |c_i-\boldsymbol{\pi}^{*}\mathbf{a}_i|.
\end{equation}
For fractional choices, the same formula counts the corresponding
fraction of each change. In our example, removing item 1 costs $1$ extra,
adding item 2 costs $1$ extra, and using any amount of item 3 costs $0$
extra. This recovers the earlier identity
$\text{cost}=499+(1-x_1)+x_2$ at weight $5$.

This identity is exact, but the numbers
$|c_i-\boldsymbol{\pi}^{*}\mathbf{a}_i|$ can still be huge. What matters is how
they compare with the available extra cost $s-f(\mathbf{t})$. The following two
uses connect prices to the recursion and the reserve.

\paragraph{First use: limit the expensive changes.}
As a separate example, take weights $1,2,3,4,5$. Each item costs $100$
times its weight, except that the weight-$1$ and weight-$5$ items each
cost an additional $3$. At target $8$, the fractional optimum is $800$:
take the weight-$2$ and weight-$3$ items and $3/4$ of the weight-$4$ item.
The optimal price $\pi^*=100$ leaves adjusted costs $3,0,0,0,3$.
With threshold $s=805$, we may spend strictly less than $5$ extra.
We can therefore select at most one of the two items carrying the
extra cost $3$. The subset with weights $3,5$ costs $803$ and passes;
the subset with weights $1,2,5$ has the same weight $8$ but costs $806$
and fails. Thus the price identity limits changes even among selections
that satisfy the weight equation.

In general, the number of expensive changes is bounded by the available
extra cost divided by the smallest cost of such a change. If one change
costs at least the whole available amount, that choice is fixed in every
successful binary solution. The algorithm processes the possible
expensive changes with exact weight tables, storing their costs as
values rather than enumerating cost as another coordinate.

\paragraph{Second use: make every reserve repair cheap.}
Return to the original three-item example, where $f(5)=499$ and $s=501$.
Rounding the LP leaves the weight-$2$ item. Adding the reserve item of
weight $3$ costs $301$ in the original objective, but changes a preferred
omission at an extra cost of only $1$. The repaired subset costs
$499+1=500<501$. This is the favorable repair illustrated in
Figure~\ref{fig:lp-price}.

For the algorithm we need a guarantee for the \emph{whole reserve},
because the required correction is not known in advance. For example,
suppose the available extra cost is $10$ and a reserve has three items,
each costing at most $1$ to change from its preferred choice. Every
selection of reserve changes then costs at most $3$. If the rounded bulk
uses less than $7$ extra, any correction supplied by this reserve fits
below the threshold. The filling property from
Section~\ref{ov:feasibility} supplies the needed weight correction;
the bound $3$ makes it affordable. This is why the algorithm chooses the
cheapness threshold with the total reserve size in mind.

The convenient binary variable now records whether a preferred choice
is changed. For a preferred inclusion, a change removes the item, so
its weight is negated; for a preferred omission, a change adds the item.
All change costs are nonnegative. Rounding down these fractional
\emph{change variables} cannot increase their cost, and the reserve
supplies the missing weight.

Thus prices identify where the loss from integrality is paid. They do
not alter the optimization problem or erase the loss. The exact dynamic
programs continue to store original costs and actual item selections.
Formally, the price identity is Lemma~\ref{lem:deviation}, the
normalization is \eqref{eq:normalize}, and the bound on expensive changes
is Lemma~\ref{lem:drift}.

\subsection{The affordable interior/boundary test}
\label{ov:affordability}

We now adapt the interior-easy, boundary-reduce framework recalled in
Section~\ref{ov:feasibility}. Feasibility measures the room around the
target inside the fractional weight region $Z$. Section~\ref{ov:pricing}
explained why optimization must also account for the cost of moving
within that region. A target can be far from the
boundary of $Z$ but have almost no affordable freedom: moving away from
it may force costly changes.

\paragraph{What should ``affordable movement'' measure?}
The LP minimum at the current target is $f(\mathbf{t})$, and the query leaves
$s-f(\mathbf{t})$ extra to spend. For a displacement $\mathbf{u}$, compare the best
fractional selections at the two opposite targets $\mathbf{t}+\mathbf{u}$ and $\mathbf{t}-\mathbf{u}$.
We include the displacement when their \emph{average} cost is at most
$s$. This gives a symmetric neighborhood of zero,
\begin{equation}\label{ov:Q}
 Q=\{\mathbf{u}:f(\mathbf{t}+\mathbf{u})+f(\mathbf{t}-\mathbf{u})\le2s\}.
\end{equation}
If either target is fractionally infeasible, its cost is $+\infty$ and
the displacement is excluded. Equivalently, the test asks whether
\[
 \underbrace{f(\mathbf{t}+\mathbf{u})+f(\mathbf{t}-\mathbf{u})-2f(\mathbf{t})}_{\text{extra cost of the two opposite movements}}
 \ \le\ 2\bigl(s-f(\mathbf{t})\bigr).
\]
Thus $Q$ measures two-sided flexibility beyond the fractional optimum.
It is a comparison region for the proof; it does not assert that each
neighboring target individually has a solution below the original
threshold.

Why take an average instead of requiring both costs to be at most $s$?
Moving in one direction changes the priced weight by
$\boldsymbol{\pi}^{*}\mathbf{u}$, while the opposite movement changes it by
$-\boldsymbol{\pi}^{*}\mathbf{u}$. These terms cancel in the sum. Consequently,
subtracting an LP price leaves the above extra-cost test unchanged.
The test measures the freedom left after the unavoidable linear cost has
been removed. It also treats opposite directions symmetrically, which
is what the boundary comparison with a reference LP needs.

\paragraph{Reading the affordable neighborhood.}
Take two items of weights $2,-2$ and costs $8,4$, respectively. Fractional
weights range over $[-2,2]$. To reach a nonnegative weight $w$, we take
$w/2$ of the positive item and pay $4w$; to reach a nonpositive weight,
we take $-w/2$ of the negative item and pay $-2w$. Hence
\[
 f(w)=\begin{cases}4w,&0\le w\le2,\\-2w,&-2\le w\le0.\end{cases}
\]
At target $0$ the LP costs $0$. With allowance $s=3$, a displacement of
$1$ gives opposite costs $f(1)=4$ and $f(-1)=2$. One exceeds $3$, but
their average equals $3$, so the displacement belongs to $Q$. More
generally, the average cost at $w,-w$ is $3|w|$. It follows that
$Q=[-1,1]$. With allowance $s=1$, the same calculation gives
$Q=[-1/3,1/3]$. The weight region is unchanged; the room left by the cost
constraint has shrunk.

\begin{samepage}
\paragraph{Two complementary cases.}
Interior must be measured relative to the items: a neighborhood that
looks large in absolute units may still be too small to absorb their
rounding errors. Fix a sufficiently large number, depending on $N$ and
$d$ but not on the costs. The proof chooses it large enough to leave
room for the reserve, the exceptional items, and the rounding repairs.
We use this number in the following complementary test.
\begin{itemize}
\item In the \emph{interior case}, $Q$ contains that number times the
convex hull of the item vectors and their negatives. There is enough room
in every item direction to construct cheap structure and try reserve
repair (Section~\ref{ov:interior}).
\item In the \emph{boundary case}, this containment fails. A separating
direction then certifies insufficient room relative to the items, and
the algorithm uses it to shrink the remaining weight problem
(Section~\ref{ov:boundary}).
\end{itemize}
\end{samepage}
Here ``boundary'' includes a neighborhood of the boundary at the
chosen scale, not only points exactly on its surface.
At an actual face of the weight zonotope, some choices are forced;
handling them first reduces the weight dimension. The discussion below
concerns the remaining targets.

\paragraph{Correspondence with the formal proof.}
In Sections~\ref{sec:buffers}--\ref{sec:lift} the test is made not at
every target separately but at \emph{anchors}: nearby queries of no
larger gap are grouped around one anchor $\vv$ (Lemma~\ref{lem:packing}), and all of
them share its neighbourhood $Q_I(\vv)$, the body $Q$ above. The interior
case is a \emph{deep} anchor, $Q_I(\vv)\supseteq KP_I$ (tested on an
inscribed parallelepiped), treated in
Sections~\ref{sec:depth}--\ref{sec:lift}; the boundary case is a
\emph{thin} anchor (Lemma~\ref{lem:thin}); and targets on faces of the
zonotope are reduced first by Lemma~\ref{lem:boundary}.

\subsubsection{Boundary branch: few large changes, then a smaller problem}
\label{ov:boundary}

This branch follows the same geometric strategy as the feasibility
test: choose a fractional reference, process the items that are large
in a restricted direction, and recurse on the smaller hull of the
surviving items. Optimization changes the reference and the reason that
only a few large choices can change. We use a minimum-cost LP solution
as the reference and measure available room in $Q$, so that the
deviation bound applies to solutions satisfying the cost threshold.

The basic intuition is again that a tight region leaves little freedom.
A binary solution
below the threshold cannot differ from it too much in a direction in
which affordable movement is scarce. The proof makes this stronger than
a bound on the final displacement: it bounds the \emph{sum of the
absolute changes}, measuring each item against the available room in
$Q$.

Consequently, among items large relative to that room, only a controlled
number can change their reference choices. The few fractional variables
of the reference LP are handled separately. An exact dynamic program
processes the large items and records the possible changed weights with
their least original costs. After removing these items, every survivor
is small relative to $Q$. Since $Q$ was thin in some direction, the
survivor hull is much thinner in that direction than the original item
hull; see Figure~\ref{fig:boundary}.

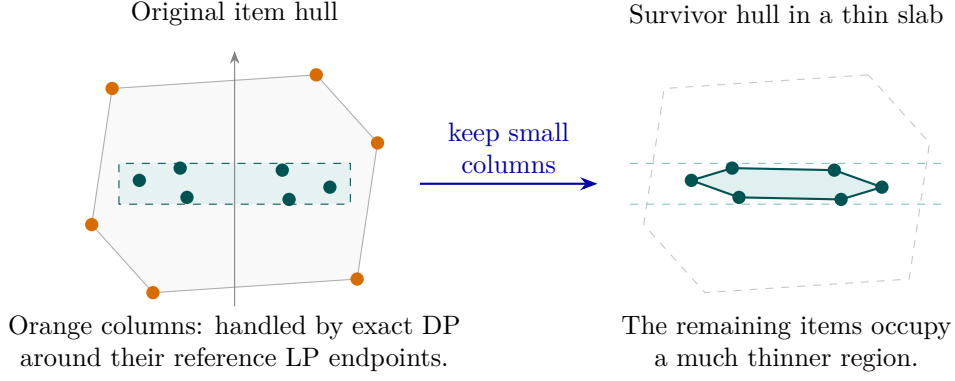
\begin{figure}[htbp]
\centering
\begin{tikzpicture}[x=1cm,y=1cm,>=Stealth,
  every node/.style={font=\small}]
  \begin{scope}[shift={(2.3,0)},x=.9cm,y=.9cm]
    \fill[gray!5] (1.2,1.6)--(-1.8,1.4)--(-2.1,-.6)--(-1.2,-1.6)
      --(1.8,-1.4)--(2.1,.6)--cycle;
    \draw[gray!65] (1.2,1.6)--(-1.8,1.4)--(-2.1,-.6)--(-1.2,-1.6)
      --(1.8,-1.4)--(2.1,.6)--cycle;
    \fill[teal!10] (-1.7,-.3) rectangle (1.7,.3);
    \draw[teal!65!black,dashed] (-1.7,-.3) rectangle (1.7,.3);
    \foreach \xx/\yy in {1.2/1.6,-1.8/1.4,-2.1/-.6,-1.2/-1.6,1.8/-1.4,2.1/.6}
      \fill[orange!85!black] (\xx,\yy) circle (2.5pt);
    \foreach \xx/\yy in {-1.4/.05,-.8/.23,.7/.2,1.4/-.05,.8/-.23,-.7/-.2}
      \fill[teal!65!black] (\xx,\yy) circle (2.5pt);
    \node at (0,2.5) {Original item hull};
    \draw[->,gray] (0,-1.8)--(0,1.95) ;
    \node[align=center] at (0,-2.35)
      {Orange columns: handled by exact DP\\around their reference LP endpoints.};
  \end{scope}
  \draw[->,blue!65!black,thick] (4.75,0)--(7.1,0)
    node[midway,above,align=center] {keep small\\columns};
  \begin{scope}[shift={(9.6,0)},x=.9cm,y=.9cm]
    \draw[gray!45,dashed] (1.2,1.6)--(-1.8,1.4)--(-2.1,-.6)--(-1.2,-1.6)
      --(1.8,-1.4)--(2.1,.6)--cycle;
    \fill[teal!12] (-1.4,.05)--(-.8,.23)--(.7,.2)--(1.4,-.05)
      --(.8,-.23)--(-.7,-.2)--cycle;
    \draw[teal!65!black,thick] (-1.4,.05)--(-.8,.23)--(.7,.2)--(1.4,-.05)
      --(.8,-.23)--(-.7,-.2)--cycle;
    \draw[teal!50,dashed] (-2.3,.3)--(2.3,.3);
    \draw[teal!50,dashed] (-2.3,-.3)--(2.3,-.3);
    \foreach \xx/\yy in {-1.4/.05,-.8/.23,.7/.2,1.4/-.05,.8/-.23,-.7/-.2}
      \fill[teal!65!black] (\xx,\yy) circle (2.5pt);
    \node at (0,2.5) {Survivor hull in a thin slab};
    \node[align=center] at (0,-2.35)
      {The remaining items occupy\\a much thinner region.};
  \end{scope}
\end{tikzpicture}
\caption{The boundary branch of Section~\ref{ov:affordability}. The target has little affordable
room in the vertical direction. Items large relative to that room
(orange) are processed by an exact table of the possible changes.
The survivors (teal) lie in a much thinner hull. Shapes and scales are
schematic; the progress comes from shrinking the item hull.}
\label{fig:boundary}
\end{figure}

There are two reasons to insist on the stronger absolute-change bound.
Signed item weights can cancel, so a small final displacement alone does
not bound the intermediate states of a dynamic program. Bounding the
sum of absolute contributions controls those states as well. Also,
progress must concern the remaining weights, not just a smaller cost
allowance; the latter could still leave an arbitrarily long sequence of
cost thresholds to process.

\paragraph{How cost changes the exact weight table.}
In feasibility, a table entry answers whether the items processed so far
can realize a given weight. In optimization, the weight states are the
same, but a yes/no answer loses information: two selections of the same
weight can leave different amounts to spend on the remaining items.
We therefore store the \emph{least original cost} at each weight,
together with an actual selection attaining it.

The item-by-item update still has the same two choices: omit the next
item or select it. If $D(\mathbf{u})$ is the stored cost at weight $\mathbf{u}$, the update
for item $i$ is
\[
 D_{\mathrm{new}}(\mathbf{u})=\min\{D_{\mathrm{old}}(\mathbf{u}),\,
                         D_{\mathrm{old}}(\mathbf{u}-\mathbf{a}_i)+c_i\},
\]
with $+\infty$ for an unreachable weight. The two alternatives use only
the old table, so the item can be selected at most once. Feasibility
would combine the two reachability answers by ``or''; optimization
keeps the cheaper of the two actual selections. The geometric deviation
bound determines the weight region that must be stored. Costs change
the value at each state, rather than adding a new array coordinate.

For example, take item weights $2,3,5,7$, with costs $3,4,11,6$,
target $12$, and threshold $15$. After processing the first three
items, weight $5$ has two witnesses: the items of weights $2,3$ cost
$7$, whereas the single weight-$5$ item costs $11$. Retain the former.
Both leave the same untouched item and residual target $7$, but their
remaining allowances are $15-7=8$ and $15-11=4$. The remaining item
costs $6$: the retained witness gives total cost $13<15$, while the
other gives $17$ and fails.

This explains why keeping only the cheapest witness is exact. At the
same processed weight and with the same remaining item pool, any
completion that works with a more expensive witness also works with
the cheapest one. After the selected items are removed, subtract their
weight from the target and their original cost from the threshold.
We then test the resulting residual query, including whether its
fractional minimum lies below its remaining threshold. For a batch of
queries, equal residual targets on the same pool are merged by keeping
the largest remaining allowance and its attaining witness. The recursive
question is whether \emph{any} retained residual query has a completion.
Thus the added operations are cost addition, minimum comparisons, and
threshold updates; even huge costs are stored as values, not enumerated.

\subsubsection{Interior branch: try a cheap repair, then expose what blocks it}
\label{ov:interior}

Return now to the interior branch: the affordable neighborhood has ample
room in all current weight directions. Our aim is to turn this room
into an affordable binary solution, or to obtain a smaller problem for
recursion. The three stages below have different roles. Stage 1 extracts
the structure needed for repair. Stage 2 either passes a suitable residual
query to stage 3, rejects the current branch, or recurses on a proper
lattice. Stage 3 either constructs a solution or adds a thin coordinate
and recurses. In particular, stage 2 need not reach stage 3.

\paragraph{1. Prepare a cheap core and reserve.}
In feasibility, the common-progression construction from
Section~\ref{ov:feasibility} separates a small exceptional set, a reserve,
and a bulk. It gives the bulk and reserve common integer coordinates,
and the reserve fills a box large enough to correct the error from
rounding a fractional bulk solution. Thus, once a suitable integer
coordinate target has been found, the construction turns it into an
actual subset of items.

For optimization, we want to use exactly this repair while controlling
its cost. We therefore apply the same construction \emph{only to cheap
items}, where cheap means a small cost of changing from the preferred
choice at the fixed LP price, not a small original item cost. We choose
the per-item cutoff so that the maximum reserve size times this cutoff
is a small part of the extra-cost allowance. Then \emph{every} reserve
subset the repair might request is affordable.

We cannot simply discard all other items and assume that enough
structure remains. We first extract a cheap \emph{core} with two further
properties. It contains all but a small number of the cheap items and
retains enough affordable interior room for the later fractional
completion. Also, deleting any sufficiently small set leaves its
generated lattice unchanged. In particular, setting aside the exceptional
items and reserve leaves the bulk generating the same lattice.
A lattice here consists of all integer combinations of the item vectors.
The core has the same real span as the current pool, but its lattice may
be smaller; stage 2 will deal with that arithmetic restriction.

Why does enough structure remain after imposing cheapness? A selection
with bounded change cost can use only a small total amount of expensive
changes: at most its allowance divided by the expensive-item cutoff.
Removing these changes therefore moves its weight by only a controlled
multiple of the item hull. The affordable neighborhood in the interior
case is wide enough to absorb this loss. Consequently, the cheap items
still leave a large inner region of weights with low-cost fractional
representations. This is where the stronger, cost-aware notion of
interior is used.

The deletion property can also be enforced while losing few cheap items.
Whenever deleting a permitted small set changes the generated lattice,
make that deletion and restore every original cheap item whose weight
lies in the new lattice. Each restriction excludes few cheap items.
The available interior room ensures that the resulting lattices still
have the same real span.

To see why this cannot continue for long, the \emph{index} of a smaller
lattice in a larger one is the number of residue classes of the larger
lattice modulo the smaller one. Two points belong to the same class
when their difference lies in the smaller lattice.
For example, $2\Z$ has index $2$ in $\Z$, corresponding to even and odd
integers; $4\Z$ has index $4$ in $\Z$. Passing from $2\Z$ to $4\Z$
again makes the lattice twice as sparse. In general, each proper
restriction of the same rank multiplies the index \emph{relative to the
fixed starting lattice} by an integer at least two. The volume bound in
Section~\ref{ov:accounting} therefore allows only logarithmically many
restrictions, and the total number of excluded cheap items stays small.

The proof implements these restrictions efficiently and produces the
unique largest core with the deletion property. This canonical choice
matters later: enlarging the cheap pool can only enlarge its core, which
lets different recursive branches share their accounting. Applying the
feasibility construction to this core now gives the desired bulk and
cheap reserve. Stage 1 prepares the repair; it does not yet solve the
target query.

\paragraph{2. Pay for the outside choices and check what remains.}
In feasibility, once a bulk and reserve have been set aside, we process
the other items exactly. For each possible choice among those items,
we subtract the selected weights from the target. The question left
for the bulk and reserve is then simply whether they can supply this
\emph{remaining weight}. This is the residual target.

Here the items processed are every item outside the cheap core, together
with the exceptional items returned by the progression construction.
Fixing one such choice means deciding, for each processed item, whether
it is selected or omitted; all processed items are then removed from
the remaining pool. Only selected items contribute to the subtraction.
Optimization also subtracts their \emph{original costs} from the original
threshold. The residual question is therefore concrete: can the untouched
bulk and reserve supply the remaining weight at a cost strictly below
the remaining threshold? A solution of this question, joined to the
stored outside selection, is a solution of the original query.

These outside choices matter even in feasibility, because the core's
lattice can be smaller than the lattice of all items. For example,
an even core cannot reach an odd target by itself. Suppose the cheap
pool is $\{3,4,6,8,10\}$ and its core must survive one deletion.
The canonical core is $\{4,6,8,10\}$: it still generates $2\Z$ after
any one item is removed, whereas the lone odd item $3$ cannot survive
that requirement. The algorithm therefore processes item $3$ separately.
For target $17$, omitting it leaves the odd target $17$, so that branch
is discarded. Selecting it leaves target $14$ and subtracts its cost
from the threshold. The survivors can supply $14$ as $4+10$ or $6+8$;
optimization must also check whether the chosen completion fits the
remaining threshold.

We represent all these outside choices by the exact weight table from
the boundary branch. Feasibility records reachable weights; optimization
keeps the least original cost and an actual selection at each weight.
As explained there, a more expensive selection at the same weight can
never help the untouched items complete the solution. The bound on
expensive changes and the small number of cheap exceptions keep the
necessary weight states in a bounded region. The table includes all
potentially successful choices, including choices of expensive items
whose weights happen to lie in the core lattice.

Discard residuals outside the core lattice, without fractional
completions, or with fractional minimum at least their remaining
threshold. The first two checks are the familiar arithmetic and
fractional-feasibility checks; the last also considers cost.

Passing these checks still does not guarantee that rounding and repair
will be affordable. For that, return to the LP price interpretation from
Section~\ref{ov:pricing}. The original query allows total change cost
less than $s-f(\mathbf{t})$. Each fixed outside choice uses some of this allowance
if it differs from its preferred choice. Subtract those nonnegative
charges; what is left must pay for \emph{both} the bulk changes and the
reserve repair. For example, an initial allowance of $10$ leaves $7$
after outside changes charged $3$, but only $1$ after changes charged
$9$. The first leaves much more room for bulk changes and repair.
The second might still have an affordable completion, but need not
leave the margin required by our repair argument.
Throughout this comparison the LP price is fixed:
we are accounting for how the original allowance is spent, not taking
the gap above a newly optimized residual LP.

This trial has three possible outcomes, which are the three cases of
Lemma~\ref{lem:high-low}. Only the second passes to stage 3 in the current
call.

\emph{No candidate survives: reject this branch.}
Every affordable binary completion would give a table entry passing
the necessary filters, because the trial table is complete. If none survives, this interior branch has no
solution. Other branches, if present, continue independently.

\emph{Some candidate keeps enough allowance: proceed to stage 3.}
If a compatible candidate retains a fixed fraction of the reference
allowance, such as one quarter, select one for the geometric test.
The whole reserve was made much cheaper than this amount, leaving room
for both bulk choices and repair. Compatibility and a cheap fractional
completion do not yet give a binary solution; that is the task of stage 3.

\emph{Candidates survive, but all have little allowance: recurse on
a proper lattice.}
The trial gives no suitable repair target. We restore every current
item whose weight belongs to the core lattice, including expensive
items and former exceptional items. Recompute the exact table using
only items outside this lattice, so its witnesses use no restored items.
The remaining pool now contains every item permitted by the lattice
restriction, with its original cost.

Why must this be a proper restriction? If the core lattice were the
whole current lattice, all items would be restored and no outside
choices would be needed. In the trial, all additionally processed items
could be left at their preferred choices, using zero extra cost.
The retained interior room guarantees a cheap fractional completion
on the bulk and reserve. We would therefore have found a compatible
trial entry with substantial allowance, contrary to this case.
We recurse on a strictly smaller lattice and skip stage 3 in the
current call. Restricting to a proper lattice is a familiar kind of
feasibility reduction. The new argument is that \emph{insufficient
remaining cost allowance} forces this arithmetic progress.

\paragraph{3. Return a solution or add a thin coordinate.}
In feasibility, the common progression lets us look for an integer
coordinate point representing the residual target. If we find one
reachable by the reference reserve subset plus a fractional bulk
selection, rounding and the reserve give an actual solution.
A thin direction in the coordinate region provides
the alternative: it reveals a restriction hidden by the old weights.
Retaining that direction as an additional integer equation gives a
smaller problem for recursion. Inside the progression coordinates,
we again solve interior targets and reduce boundary targets.

Optimization uses the same two outcomes. Stage 2 has supplied a target
with substantial extra-cost allowance, and stage 1 has supplied the
cheap reserve. We now look for a fractional bulk representation using
only a small part of that allowance, leaving enough for repair.
If the resulting low-cost region instead yields a thin direction, we
must show that it reduces \emph{every} relevant affordable completion,
including those spending more than the small amount used in the test.
Thus a successful test returns an affordable solution; a thin outcome
must preserve the whole trial table for recursion.

Let $\phi$ map common-progression coordinates back to the old weights,
and let $\mathbf{t}$ now denote the selected residual target. The equation
$\phi\mathbf{z}=\mathbf{t}$ defines its \emph{fiber}: different coordinate points can
represent the same old weight. Fix the certified reference subset of
the reserve and consider the points on this fiber that it can reach
together with a fractional bulk selection using a small fraction of
the remaining allowance. Affordable interior and the properties from
stage 1 ensure that this test region is nonempty. A fixed-dimensional
lattice procedure either finds an integer point there or finds an
integer linear coordinate varying over a short interval there, whose
value is not already fixed by the old weight equations.

\emph{Integer-point outcome: return an affordable solution.}
At the returned coordinate point, minimize the bulk change cost, take
an LP vertex, round down its few fractional variables, and use the reserve
to supply the missing coordinate vector. This is the feasibility repair.
The added cost guarantee is that the change costs are nonnegative, so
rounding down cannot increase the bulk cost, and every reserve subset
is cheap. The completed selection therefore fits the residual threshold.
Combine it with the stored outside-item witness and return a solution
to the original query.

\emph{Thin-direction outcome: add a coordinate and recurse.}
A thin test region need not be free of integer points, and it does not
rule out a solution spending more. Its value is that it reveals a
direction hidden by the old weight map. For example, $\phi(x,y)=x+y$
identifies the entire line $x+y=1$ with the same old weight $1$.
A short segment of that line can lie between consecutive integer
points, a restriction that the old weight alone cannot see
(Figure~\ref{fig:fiber}).

\begin{figure}[htbp]
\centering
\begin{tikzpicture}[x=1cm,y=1cm,>=Stealth,
  every node/.style={font=\small}]
  \begin{scope}[shift={(1.1,.2)},x=2.6cm,y=2.6cm]
    \draw[step=1,gray!18] (-.1,-.1) grid (1.1,1.1);
    \draw[->,gray] (-.15,0)--(1.35,0) node[right] {$x$};
    \draw[->,gray] (0,-.15)--(0,1.35) node[above] {$y$};
    \draw[gray!55,dashed] (-.15,1.15)--(1.15,-.15);
    \draw[blue!65!black,line width=4pt] (.25,.75)--(.75,.25);
    \draw[blue!65!black,fill=white] (.25,.75) circle (1.7pt);
    \draw[blue!65!black,fill=white] (.75,.25) circle (1.7pt);
    \fill[black!70] (0,1) circle (2.4pt);
    \fill[black!70] (1,0) circle (2.4pt);
    \node[above left] at (0,1) {$(0,1)$};
    \node[below] at (1,0) {$(1,0)$};
    \node[above right] at (.25,.75) {$K$};
    \node[align=center] at (.5,1.65)
      {$\phi(x,y)=x+y=1$\\
       $K=\{(x,1-x):\frac14\le x\le\frac34\}$};
  \end{scope}
  \draw[->,blue!65!black,thick] (5.05,1.65)--(7.25,1.65)
    node[midway,above,align=center] {new coordinate\\$x$};
  \begin{scope}[shift={(7.9,1.65)},x=3cm,y=1cm]
    \draw[->,gray] (-.15,0)--(1.35,0) node[right] {$x$};
    \draw[blue!65!black,line width=4pt] (.25,0)--(.75,0);
    \foreach \xx/\lab in {0/0,1/1} {
      \fill[black!70] (\xx,0) circle (2.5pt);
      \node[below] at (\xx,-.12) {$\lab$};
    }
    \draw[decorate,decoration={brace,amplitude=4pt,mirror}]
      (.25,-.45)--(.75,-.45)
      node[midway,below=5pt] {width $\frac12$};
    \node[align=center] at (.5,1.48)
      {A thin interval in an integer coordinate\\
       $\frac14\le x\le\frac34$.};
  \end{scope}
\end{tikzpicture}
\caption{A boundary hidden inside the interior branch. All points
on the line have the same old weight $x+y=1$. The short blue segment
between integer points is visible only if we retain another coordinate,
such as $x$. In the algorithm, a thin direction found at a small budget
is extended to cover every relevant completion before the new equation
is imposed. A failed small-budget repair alone does not reject the query.}
\label{fig:fiber}
\end{figure}
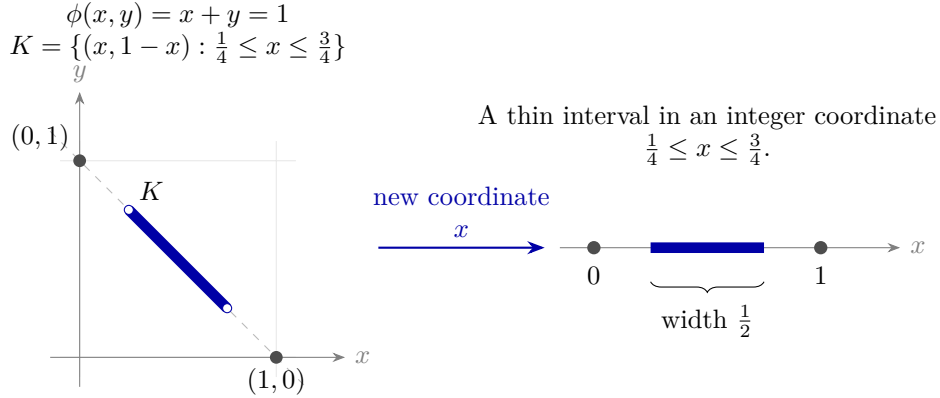

The additional optimization argument extends this thinness from the
small test region to all relevant affordable completions. Interior
room lets a small set of bulk choices be prescribed while the other
items still complete fractionally at low cost. If many bulk items
had large effects in the discovered direction, prescribing those
choices would force a large variation inside the test fiber,
contradicting thinness. LP duality gives a quantitative bound that
accounts for both old weights and cost.

As in feasibility, process the resulting small exceptional set and the
reserve, whose columns need not be thin, by an exact table. Here the
table retains minimum costs as before. The surviving columns are then
thin in the discovered coordinate, relative to a linear function of
the old coordinates. Express the new coordinate in terms of the
original item choices, append it to the surviving weights, and retain
their original costs. Above every old residual target, only
a dimension-dependent number of new integer values remain possible;
we keep all of them. This common reduction covers the \emph{whole trial
table}. The candidate selected in stage 2 is used to discover the
coordinate, not to discard other possible solutions. We now recurse
with one additional thin coordinate; the volume argument below explains
the progress.

Thus the interior branch either returns a solution, ends an infeasible
branch, recurses on a proper lattice in stage 2, or recurses with a
thin new coordinate in stage 3. Its geometric construction and repair
follow the feasibility framework. Cheap-core extraction, the test of
remaining allowance, and extension to all affordable completions are
what make that framework work for optimization.

\subsection{Why these reductions control the worst case}
\label{ov:accounting}

The local picture is now complete. An unresolved query leaves a smaller
item hull, a lower-dimensional problem, a proper lattice, or a problem
with one additional but very thin coordinate. We still have to explain
why these outcomes do not generate too much total work.

\paragraph{Measure size in the items' own lattice.}
Take the symmetric convex hull of the remaining item vectors and measure
its volume in units of a fundamental cell of their generated lattice.
With a dimension-dependent normalization, this quantity is at least $1$
for every nontrivial item pool. Initially it is at most $C_dN^d$.
Shrinking the hull decreases it. Passing to a proper sublattice also
decreases it, because a lattice cell becomes larger. Rank reductions
are compatible with the same measure.

Adding a coordinate might seem to undo this progress. Here it does not:
after subtracting a real linear function of the old coordinates for the
analysis, the new hull lies in a very thin prism over the old one. The
new coordinate is integer, so the new lattice cannot have an arbitrarily
small vertical spacing. Its volume in lattice units decreases by a
small power of $N$. Starting from $N^{O(d)}$ and staying at least $1$
allows only a bounded number of such additions, depending on $d$ and on
that power. The stored dimension remains bounded in the same way;
the full temporary progression dimension is never appended at once.

\paragraph{Charge all children together.}
A short path does not suffice: many paths could repeat the same work.
The proof therefore groups nearby residual queries and gives each group
a small affordable neighborhood to pay for its computation. When actual
items are selected and removed, every fractional completion of the child
extends to one of the parent with exactly those weight and cost offsets.
Child affordable neighborhoods consequently fit inside the corresponding
parent geometry. The absolute-change bounds also control the intermediate
states of the removed-item tables.

Arithmetic branches require an additional sharing argument. At a parent
state charged by several cores, LP duality and local price stability let
those cores be compared
using one common supporting price. The canonical core lattices are then
nested; identical remaining pools are merged. Their distinct indices
relative to the same parent lattice grow at least as $2,4,8,\ldots$,
so the normalized volumes of all cores sharing a parent state sum to at
most the parent's volume. This sharing bound is not a strict
contraction; the potential pays for it with the remaining arithmetic
depth, which drops by one in every arithmetic child and is at most
logarithmic. This is why canonical cores,
rather than arbitrary cheap subsets, matter to the worst-case bound.
The formal potential (Proposition~\ref{prop:contraction}) combines these
savings to control the entire recursion tree.

\begin{samepage}
\paragraph{Keep geometric size separate from cost magnitude.}
Thin integer coordinates can have large coefficients. After a reduction,
we therefore change to a basis of the full generated lattice and put the
items in a coordinate box whose \emph{product} of side lengths is bounded
by the volume measure. This keeps the geometric encodings within $N^{O(d)}$;
bounding each side separately by the full volume could instead cost
$N^{O(d^2)}$.

Throughout, exact original costs are values attached to these weight
states. They are not enumerated. Anchors are chosen by a single packing
ordered by decreasing gap. Its locality is only one-sided, since a query
may have a much smaller gap than its anchor, but this suffices for every
reduction, and all scales are independent of the costs.
\par
\end{samepage}

\paragraph{Pay for each node by mass times volume.}
The finite bounds are removed once, before the recursion: rounding an
optimal LP vertex down leaves a residual of norm at most $dN$, and a
Steinitz argument shows that an optimal correction uses $O_d(N^d)$ unit
changes. Keeping that many cheapest copies of each of the
$O_d(N^d)$ possible signed columns leaves $O_d(N^{2d})$ binary items.
Many of them share a column. The recursion therefore works with
\emph{types}, the distinct columns with their cost-sorted copies; a pool
of normalized volume $\Psi$ has $O_d(\Psi)$ types. Removed-item tables
process a whole type at once, because the cheapest $k$ copies have a
convex cost in $k$, and all queries assigned to one anchor share one
table. With these implementations a node of buffered mass $w$ costs
about $N^{2\delta}w\Psi$ for a small $\delta>0$, and the potential
$(J-j)\cdot w\Psi$, with polylogarithmic weights for the rank and the
number of lifts, pays for all nodes. At the root
$\Psi\approx N^d$ and $w\approx N^{d\delta}\Psi$, which gives
$N^{2d+O(\delta)}$ for one decision.

%======================================================================
\section{Finite bounds and a compact signed kernel}\label{sec:kernel}

This section reduces \eqref{eq:main} to a binary problem with
$O_d(N^{2d})$ labelled items, a target of norm at most $dN$, and an
inherited optimal LP vertex. Only one linear program on the original
instance is solved, and the
numerical sizes of $\vb$, $\vell$, $\vu$ do not enter the reduced
instance.

\subsection{Preprocessing and ranged rows}

An empty interval $\ell_i>u_i$ proves infeasibility. A variable with
$\ell_i=u_i$ is substituted. A variable with zero column is set to a
cost-minimizing endpoint (either endpoint if $c_i=0$). These operations
change only the target and a fixed objective offset. If no nonzero
column remains, the residual equations and the fixed cost settle the
instance directly. A negative lower bound or a negative target
coordinate needs no special treatment below.

\begin{proposition}[Ranged rows without new rows]\label{prop:ranged}
Consider $\vb^-\le A\vx\le\vb^+$, $\vell\le\vx\le\vu$, with finite variable
bounds; a missing row endpoint may be infinite. In $O(dn)$ arithmetic
operations one can either prove infeasibility or produce an equivalent
system of $d$ equations with finite variable bounds on at most $n+d$
variables. The maximum absolute matrix entry remains at most $N$, and
the added variables have cost zero.
\end{proposition}

\begin{proof}
Over the variable box the activity of row $j$ lies in
$[L^{\rm box}_j,U^{\rm box}_j]$ with
$L^{\rm box}_j=\sum_i\min\{A_{ji}\ell_i,A_{ji}u_i\}$ and
$U^{\rm box}_j=\sum_i\max\{A_{ji}\ell_i,A_{ji}u_i\}$. Put
$L_j=\max\{b^-_j,L^{\rm box}_j\}$ and $U_j=\min\{b^+_j,U^{\rm box}_j\}$. If
$L_j>U_j$, reject. Otherwise row $j$ is equivalent to
$A_j\vx+s_j=U_j$ with $0\le s_j\le U_j-L_j$; the slack has coefficient
one and cost zero, and is omitted for an equality row. Its bound is
stored as one integer.
\end{proof}

We henceforth use the equality form \eqref{eq:main}, after this
preprocessing and relabelling. The at most $d$ slack variables do not
change the $O_d(n)$ term, and fixed contributions are added to the final
objective offset. Thus all columns are nonzero, $\|\va_i\|_\infty\le N$,
and $|c_i|\le c_{\max}$.

\subsection{Independent signed moves}

Reject if the LP relaxation of \eqref{eq:main} is infeasible. Otherwise
compute an exact optimal vertex $\vx^*$ and an optimal dual price
$\vpi$ (a row vector); this is the time $T_{\rm LP}$. Put
\[
 \bar x_i=\lfloor x^*_i\rfloor,\qquad \theta_i=x^*_i-\bar x_i,\qquad
 F=\{i:\theta_i>0\}.
\]
At a vertex of a box-constrained equality LP at most $d$ variables lie
strictly between their bounds, since otherwise their columns are
dependent and give a two-sided feasible direction; so $|F|\le d$. As the
bounds are integral, $\vell\le\bar\vx\le\vu$, and the integral residual
satisfies
\begin{equation}\label{eq:kernel-residual}
 \vb_0=\vb-A\bar\vx=\sum_{i\in F}\theta_i\va_i,\qquad
 \|\vb_0\|_\infty\le dN .
\end{equation}

Variable $i$ has a \emph{forward pool} of $u_i-\bar x_i$ labelled moves
$(\va_i,c_i)$ and a \emph{backward pool} of $\bar x_i-\ell_i$ labelled
moves $(-\va_i,-c_i)$. A move $h$ is a pair $(\va_h,c_h)$ of signed column
and objective change. The pools are conceptual: each is stored by its
binary count, variable index, orientation, and common cost. For an
independent assignment $0\le y_h\le1$ to the available moves, let
$y^+_i$ and $y^-_i$ be the total forward and backward assignment of
variable $i$. Then
\begin{equation}\label{eq:moves-box}
 x_i=\bar x_i+y^+_i-y^-_i,\qquad
 \ell_i\le\bar x_i-y^-_i\le x_i\le\bar x_i+y^+_i\le u_i,
\end{equation}
so every independent assignment respects the original bounds, even if
both orientations of a variable are used, and
\begin{equation}\label{eq:moves-identity}
 A\vx=\vb\iff\sum_h\va_hy_h=\vb_0,\qquad
 \vc^{\T}\vx=\vc^{\T}\bar\vx+\sum_hc_hy_h .
\end{equation}
Conversely every integer feasible $\vx$ is represented, using only
forward moves where $x_i\ge\bar x_i$ and only backward moves otherwise.

Complementary slackness for $\vpi$ gives $c_i-\vpi\va_i\ge0$ when
$x^*_i=\ell_i<u_i$, $\le0$ when $x^*_i=u_i>\ell_i$, and $=0$ when
$\ell_i<x^*_i<u_i$. A nonempty forward pool means $x^*_i<u_i$, and a
nonempty backward pool means $x^*_i>\ell_i$. Hence every available move
has nonnegative reduced cost
\begin{equation}\label{eq:move-reduced}
 c_h-\vpi\va_h\ge0 .
\end{equation}
For each $i\in F$ \emph{mark} one forward move $h_i$; it exists because
$\bar x_i+1\le u_i$, and it has reduced cost zero. The fractional move
assignment
\begin{equation}\label{eq:marked}
 y^*_{h_i}=\theta_i\quad(i\in F),\qquad y^*_h=0\ \text{otherwise}
\end{equation}
represents $\vx^*$, and its move cost is
$\sum_{i\in F}c_i\theta_i=\vpi\vb_0$. Within one signed vector, ordering by
$c_h$ and by $c_h-\vpi\va_h$ is the same, so prices are not needed to
select cheap moves.

\subsection{A short optimal correction}

We include the classical rearrangement argument; Eisenbrand and
Weismantel~\cite{EW19} give the broader proximity context. The proofs
below are existence arguments, and their conceptual expansions are not
performed by the algorithm.

\begin{lemma}[Steinitz rearrangement]\label{lem:steinitz}
Let $\vz_1,\dots,\vz_L\in\R^d$ satisfy $\sum_i\vz_i=\vzero$ and
$\|\vz_i\|_\infty\le N$. There is an ordering in which every partial sum
lies in $[-dN,dN]^d$.
\end{lemma}

\begin{proof}
If $L\le d$ any ordering works. Otherwise construct nested sets
$S_L\supset S_{L-1}\supset\dots\supset S_d$ with $|S_k|=k$ for which
\[
 \sum_{i\in S_k}\alpha_i\vz_i=\vzero,\qquad
 \sum_{i\in S_k}\alpha_i=k-d,\qquad 0\le\alpha_i\le1
\]
is feasible. For $S_L=\{1,\dots,L\}$ take $\alpha_i=(L-d)/L$. Given a
solution for $S_k$ with $k>d$, scaling by $(k-d-1)/(k-d)$ gives a
solution with second right-hand side $k-d-1$. A vertex of this new
polytope has at most $d+1$ coordinates strictly between $0$ and $1$. It
has a zero coordinate: otherwise its deficits $1-\alpha_i$ would sum to
less than $d+1$, whereas the second equation makes them sum to exactly
$d+1$. Deleting such an index gives $S_{k-1}$.

Order the indices of $S_d$ arbitrarily and then append the index of
$S_k\setminus S_{k-1}$ for $k=d+1,\dots,L$. For $k\ge d$ the prefix
$S_k$ satisfies
$\|\sum_{S_k}\vz_i\|_\infty=\|\sum_{S_k}(1-\alpha_i)\vz_i\|_\infty\le
N\sum_{S_k}(1-\alpha_i)=dN$; shorter prefixes have at most $d$ terms.
\end{proof}

\begin{lemma}[Short optimal correction]\label{lem:short-correction}
If \eqref{eq:main} is integer feasible, it has an optimal solution whose
move representation uses fewer than
\begin{equation}\label{eq:Rpx}
 R_{\rm px}=(d+1)(2dN+1)^d
\end{equation}
labelled moves.
\end{lemma}

\begin{proof}
Finite bounds give an integer optimum. By
\eqref{eq:moves-box}--\eqref{eq:moves-identity}, minimizing move cost is
equivalent to \eqref{eq:main}. Among optimal move assignments choose one
of minimum cardinality, with multiset of signed columns $V$. It has no
nonempty zero-sum submultiset: by \eqref{eq:move-reduced} the cost of
such a submultiset is the sum of its nonnegative reduced costs, and
deleting it keeps all capacities and the target; positive cost
contradicts optimality and zero cost contradicts minimality.

Append the $|F|\le d$ marked vectors $-\theta_i\va_i$ ($i\in F$). They
have norm at most $N$, and the augmented sequence sums to zero by
\eqref{eq:kernel-residual}. Apply Lemma~\ref{lem:steinitz}. The marked
vectors cut the moves of $V$ into at most $|F|+1$ consecutive runs. A run
with $L_j$ moves has $L_j+1$ boundary partial sums, which lie in
$[-dN,dN]^d$ and in one translate of $\Z^d$, since the steps of the run
are integral. Two equal ones would give a zero-sum submultiset of $V$.
Hence $L_j\le(2dN+1)^d-1$, and $|V|\le(|F|+1)((2dN+1)^d-1)<R_{\rm px}$. The
appended vectors are treated as marked even if one of them is integral.
\end{proof}

\subsection{Retention, infeasibility, and the inherited LP vertex}

\begin{lemma}[Compact kernel]\label{lem:kernel}
For each signed vector $\va\in\Z^d\setminus\{\vzero\}$ retain the
$R_{\rm px}$ cheapest available moves with column $\va$, or all of them if
there are fewer; among equal costs give marked moves priority. The
resulting binary instance has
\[
 n_0\le R_{\rm px}\bigl((2N+1)^d-1\bigr)=O_d(N^{2d})
\]
labelled items with $\|\va_h\|_\infty\le N$ and $|c_h|\le c_{\max}$, and
target $\vb_0$ with $\|\vb_0\|_\infty\le dN$. It is feasible if and only
if \eqref{eq:main} is feasible; when feasible, its optimum plus
$\vc^{\T}\bar\vx$ is the optimum of \eqref{eq:main}; and every binary
solution maps to a feasible point of \eqref{eq:main} of the same cost
by \eqref{eq:moves-box}--\eqref{eq:moves-identity}.
\end{lemma}

\begin{proof}
There are at most $(2N+1)^d-1$ nonzero signed vectors. Moves of different
variables or orientations may share a vector; for example a forward
move of $\va_i$ and a backward move of $\va_j=-\va_i$.

If \eqref{eq:main} is feasible, take the optimum of
Lemma~\ref{lem:short-correction}. In each vector group it uses at most
$R_{\rm px}$ moves; replace them by the same number of cheapest moves of
that vector. The column sum is unchanged, the cost does not increase,
and the replacement consists of distinct available moves, so by
\eqref{eq:moves-box} it is feasible. Thus an original optimum survives.
Conversely every retained binary assignment maps to a feasible original
point with the same cost by \eqref{eq:moves-identity}; in particular an
infeasible original problem gives an infeasible kernel. Finally a
marked move has reduced cost zero and all moves of its vector have
nonnegative reduced cost, so it has minimum cost in its group; at most
$d\le R_{\rm px}$ moves are marked, so the tie rule retains all of them.
\end{proof}

\begin{lemma}[Inherited optimal LP vertex]\label{lem:kernel-vertex}
The assignment \eqref{eq:marked}, restricted to the retained items, is an
optimal vertex of the kernel LP, with value
$\vc^{\T}\vx^*-\vc^{\T}\bar\vx=\vpi\vb_0$. No LP is solved on the kernel.
\end{lemma}

\begin{proof}
All marked moves are retained, so the assignment is feasible and maps to
$\vx^*$. Every fractional kernel assignment maps to a feasible fractional
original point of the same cost, which proves optimality. A two-sided
feasible direction at this point vanishes on the unmarked coordinates,
which are zero. On the marked coordinates aggregation into original
variables is injective, since there is at most one marked move per
variable and all are forward. Such a direction would therefore give a
two-sided feasible direction at $\vx^*$.
\end{proof}

\paragraph{Deterministic selection.}
Split each marked move into its own run of length one; this adds at
most $d$ records. Group all runs by signed vector using direct
addressing in $[-N,N]^d$, in $O_d(n+N^d)$ operations. The first
$R_{\rm px}$ moves of a group are found by weighted selection rather
than sorting: take a median of the record keys (cost, marked priority,
and a fixed tie order) by median-of-medians, sum the run multiplicities
on each side, and recurse on the side containing the requested weighted
rank, after capping each run at $R_{\rm px}$. Each round discards half
of the records, so all cutoffs take $O_d(n)$ operations. Only the $n_0$
retained moves are expanded; sorting them into cost profiles adds
$O(n_0\log(n_0+2))$ operations. Each retained item stores its original
variable and orientation.

\paragraph{Compact data and infeasibility.}
The large numbers $\vb,\vell,\vu$ survive only in the initial LP and
in $\bar\vx$. The fractional parts $\theta_i$ have a common denominator
dividing a nonsingular minor of order at most $d$ of the columns of $F$,
of absolute value at most $d!N^d$, so they have $O_d(\log N)$ bits.
Every binary kernel solution has cost in $[-n_0c_{\max},n_0c_{\max}]$, so
the strict query with allowance $n_0c_{\max}+1$ decides integer
feasibility; it has positive fractional gap whenever the kernel LP is
feasible, since its LP value lies in the same interval. When the kernel
is feasible, $O(1+\log(n_0(c_{\max}+1)))$ strict queries find the optimal
cost, followed by witness reconstruction. If $n_0=0$ the equation
$\vb_0=\vzero$ is checked directly. No integer incumbent is needed.

\subsection{A small signed example}

Let $d=2$,
\[
 A=\begin{pmatrix}2&-1&1\\-1&2&0\end{pmatrix},\quad
 \vb=\begin{pmatrix}2\\1\end{pmatrix},\quad
 \vell=(-1,0,2),\quad \vu=(1,1,4),\quad \vc=(-1,1,1).
\]
Every fractional feasible point is $\vx=(2t-1,t,4-3t)$ with
$0\le t\le\frac23$, of cost $5-4t$. Thus
$\vx^*=(\frac13,\frac23,2)$, $\bar\vx=(0,0,2)$, $\vb_0=(0,1)$, and
$\vc^{\T}\bar\vx=2$. The moves are: increase $x_1$, column $(2,-1)$, cost
$-1$, one copy; decrease $x_1$, column $(-2,1)$, cost $1$, one copy;
increase $x_2$, column $(-1,2)$, cost $1$, one copy; increase $x_3$,
column $(1,0)$, cost $1$, two copies. The fractional anchor uses one third
of the first move and two thirds of the third, with column $(0,1)$ and
cost $\frac13$; with $\vpi=(-\frac13,\frac13)$ the first three moves have
reduced cost $0$ and the fourth $\frac43$. Integrality forces $t=0$, so
the unique integer point is $(-1,0,4)$ of cost $5$; its correction uses
the backward $x_1$ move and both forward $x_3$ moves:
$(-2,1)+2(1,0)=(0,1)$ and $2+(1+2)=5$. Using both orientations of $x_1$
would cancel and still respect its bound; a minimum-cardinality optimum
omits such a pair.

\paragraph{Convention for the rest of the paper.}
From now on the retained moves are the items of a binary instance: each
item $i$ has an integer column $\va_i$ with $\|\va_i\|_\infty\le N$, an
integer cost $|c_i|\le c_{\max}$, and its provenance. Items with equal
columns are distinct labelled items, and are independently selectable.
The constant $\vc^{\T}\bar\vx$ is kept separately. All later pools are
subsets of this universe of $n_0=O_d(N^{2d})$ items.

%======================================================================
\section{Preliminaries}\label{sec:prelim}

From now on we solve strict queries on the binary kernel of
Section~\ref{sec:kernel}. Its universe has $n_0=O_d(N^{2d})$ labelled
items with columns in $[-N,N]^d$ and costs in $[-c_{\max},c_{\max}]$.

\subsection{Notation}

All item collections are \emph{labelled}: equal vectors still represent
different, independently selectable items. Vectors (including points,
targets, prices, and linear functionals) are written in bold, and their
scalar components in ordinary type; matrices and linear maps are not
bold. For a pool $I$ with columns $\va_i$ and costs $c_i$ write
$A_X \vx$ and $\vc_X^{\T}\vx$ for sums over $X\subseteq I$, and
\[
 Z(I)=\sum_{i\in I}[\vzero,\va_i],\qquad
 P_I=\conv\{\pm \va_i:i\in I\},\qquad
 \Lambda_I=\lat\{\va_i:i\in I\},\qquad r=\rank I .
\]
A \emph{type} of $I$ is a nonzero vector $\va$ occurring as a column of
$I$, together with the list of items of $I$ with column $\va$, ordered by
cost (ties by a fixed label order); $T_I$ denotes the number of types.
Volumes and covolumes are measured in the real span of the pool. For a
compact convex $K$, $h_K(\vlambda)=\max_{\vz\in K}\vlambda\vz$ is its
support function; for a symmetric body $K$, $\|\cdot\|_K$ is its gauge. A
\emph{supporting price} at $\vt$ is a subgradient (a row vector) of a
convex function at $\vt$.

Besides $d$, the algorithm has one further fixed parameter, the exponent
$\delta>0$ of Section~\ref{sec:params}, chosen in terms of $d$ and
$\eps$. Symbols $C_d$ denote quantities depending only on $d$, and
$C_{\rm dim}$ quantities depending only on $d$ and $\delta$ (through the
fixed bounds on stored and latent ranks below); both may change between
occurrences. We write $\widetilde O_{d,\delta}(X)$ for $X$ times a
factor $C_{\rm dim}\bigl(\log(N+2)\bigr)^{C_{\rm dim}}$; since
$n_0\le C_dN^{2d}$, this also absorbs fixed powers of $\log(n_0+2)$. The
cost bound is \emph{not} hidden in this notation: it enters only through
\[
 L_c=\log_2(c_{\max}+2)\ \ge1,
\]
which is always written explicitly. Every number that arises has
$O_{\rm dim}(\log(n_0N(c_{\max}+2)))$ bits (Section~\ref{sec:implementation}),
that is, $O_{\rm dim}(1)$ words, so arithmetic on it takes $O_{\rm dim}(1)$
time; note $\log_2(n_0N(c_{\max}+2))\le C_d\log(N+2)\,L_c$.
Table~\ref{tab:notation} lists the principal symbols.

\begin{table}[ht]
\centering\small
\begin{tabular}{@{}ll@{}}
\toprule
Symbol & Meaning\\
\midrule
$I$, $X$, $R$ & item pool; removed (transferred) items; remaining items or core\\
$T_I$ & number of types (distinct nonzero columns) of $I$\\
$\vq=(\vt,s)$, $\cW$ & strict query (target, allowance); disjunctive batch\\
$f_I$, $g(\vq)$ & fractional value function; gap $s-f_I(\vt)$\\
$M_I(\vq)$, $Q_I(\vq)$ & Macbeath body in the epigraph; its weight projection\\
$\Psi(I)$ & normalized volume \eqref{eq:psi}\\
$\rho_j,H_j$; $w$ & buffer scales at arithmetic depth $j$; buffered lattice mass\\
$\alpha,\gamma=3\alpha$, $C_\gamma$ & packing scale; capture radius; comparison factor $\frac{1+\gamma}{1-\gamma}$\\
$K$, $D_v$ & thin/deep threshold; rounded body of an anchor\\
$\vv=(\vt_v,s_v)$, $g_v$, $\vpi$ & anchor, its gap, a supporting price at the anchor\\
$e_i(\vpi),\sigma_i,\bar\va_i,\bar c_i$ & preferred endpoint, sign, working column, normalized cost \eqref{eq:normalize}\\
$\cT_Y(b)$ & affordable body of pool $Y$ at normalized budget $b$\\
$\cR_\kappa$, $\kappa$, $k_0$, $J$ & canonical resilient core; resilience; outlier bound; depth bound\\
$\varrho=g_v/32$ & cheap budget\\
$E$ & drift bound (Lemma~\ref{lem:drift})\\
$\nu$, $\tau$, $h$ & reserve scale $\lceil N^\delta\rceil$; multiset scale; filling scale\\
$k,\vu_i,\phi,\Bx$ & GAP rank, coordinates, decoding map, box\\
$I_{\rm out},I_{\rm res},I_{\rm rem}$ & GAP outliers, reserve, remainder\\
$\vpsi$, $W$ & thin integral functional; its width bound\\
$\T_{\rm lift}$, $D$ & lift budget; stored-rank bound\\
$\Phi,B,\chi$ & potential, node weight, and its rank/lift factor\\
\bottomrule
\end{tabular}
\caption{Principal notation.}\label{tab:notation}
\end{table}

\subsection{Strict queries, batches, and exact transfer}\label{sec:queries}

A strict query on pool $I$ asks for
\[
 A_I\vx=\vt,\qquad \vc_I^{\T}\vx<s,\qquad \vx\in\{0,1\}^I,
\]
with integer $\vt$, $s$, and costs. A recursive \emph{state} stores: the
labelled pool $I$, current integer columns in the original orientation of
the items, original costs, the lattice $\Lambda_I$, a finite
\emph{disjunctive batch} $\cW$ of queries (it suffices to solve one of
them), the arithmetic depth $j$, the lift count $\ell$, and witness
pointers recording the assignment of all removed items. Among queries with
the same target only the largest allowance is kept, together with one
source attaining it; this is valid for a disjunctive batch. A query is
discarded if $\vt\notin\Lambda_I$, if $\vt\notin Z(I)$, or if $s\le f_I(\vt)$.

\begin{lemma}[Exact transfer and dominance]\label{lem:transfer}
Let $X\subseteq I$, $R=I\setminus X$, and $\vY\in\{0,1\}^X$. Binary or
fractional solutions of the residual query
$(\vt-A_X\vY,\ s-\vc_X^{\T}\vY)$ on $R$ extend by $\vY$ to solutions of $(\vt,s)$ on
$I$; in particular
$(A_X\vY,\vc_X^{\T}\vY)+\epi f_R\subseteq\epi f_I$. If $\vY'$ is another
assignment of $X$ with $A_X\vY'=A_X\vY$ and $\vc_X^{\T}\vY'\le \vc_X^{\T}\vY$, every
completion of $\vY$ is a completion of $\vY'$.
\end{lemma}

\begin{proof}
Weights and costs of disjoint labelled sets add.
\end{proof}

Removed items are therefore processed by exact \emph{removed-item tables}:
items are scanned in a fixed order, and for each attained weight the table
keeps the minimum original cost with a witness. Under any fixed price $\vpi$
original and normalized costs at a common weight differ by a fixed affine
function of that weight, so the same table is optimal for both. The size of
a table is controlled by restricting prefix weights to a window, which is
justified in each use.

\subsection{The structural theorem}

The following constructive structure theorem for multisets is the only
result imported from~\cite{CMZ26}. Its input is given by the support and
the multiplicities of a multiset, so repeated columns cost nothing extra.

\begin{theorem}[Multiset common progression and filling reserve; cited]\label{thm:gap}
Fix $r_0\ge1$ and a sufficiently small $\eta>0$. Let a multiset of nonzero
integer vectors in $\Z^r$, $r\le r_0$, be given by its support and
multiplicities, and let $\mathsf Q$ be an integer box containing the
support. Let $\tau$ be a sufficiently large integer,
$W_{\rm loc}=|\tau\mathsf Q\cap\Z^r|$, and
$h=\lfloor\tau/\log^3W_{\rm loc}\rfloor$. Assume $\tau\ge W_{\rm loc}^{2\eta}$
and that the multiset has more than $c'_{\eta,r_0}\tau$ copies. A
deterministic algorithm partitions the copies into outliers
$I_{\rm out}$, reserve $I_{\rm res}$, and remainder $I_{\rm rem}$, with
$|I_{\rm out}|+|I_{\rm res}|=O_{\eta,r_0}(\tau)$, and returns:
\begin{enumerate}[label=\textup{(\roman*)}]
\item a box $\Bx=\prod_{j=1}^k[-b_j,b_j]$ with integers $b_j\ge1$ and
$k=O_{\eta,r_0}(1)$;
\item vectors $\vu_i\in\Bx\cap\Z^k$ and an integer linear map $\phi$ with
$\phi\vu_i=\va_i$ for $i\in I_{\rm res}\cup I_{\rm rem}$;
\item injectivity of $\phi$ on $c_\eta h\Bx\cap\Z^k$, which contains all
$\vu_i$, so $\vu_i$ is determined by $\va_i$;
\item actual reserve witnesses for
\[
 \vz_0+(c_\eta h\Bx\cap\Z^k)\subseteq
 \Bigl\{\textstyle\sum_{i\in Y}\vu_i:Y\subseteq I_{\rm res}\Bigr\};
\]
\item real rank $k$ of $\{\vu_i:i\in I_{\rm rem}\}$, preserved after
deleting any $h$ copies.
\end{enumerate}
The construction uses compact multiplicities and runs in time
$\widetilde O_{r_0,\eta}(hW_{\rm loc})$. This is \cite[Theorem~9.5 and
Corollaries~9.6 and~10.3]{CMZ26}.
\end{theorem}

\begin{remark}[Using the theorem]\label{rem:gap-use}
Scales are chosen so that $c_\eta h\ge k$ (Lemma~\ref{lem:scales}). If
the construction internally flips signs of items to make encodings
positive, we undo this on the output: with $\vu_i=\sigma_i\vu_i^+$, the
reserve subset sums are translated by $-\sum_{\sigma_i=-1}\vu_i^+$,
$\vz_0$ is translated by the same vector, the symmetric box is
unchanged, and witnesses are complemented on flipped labels. The theorem
receives the support and multiplicities of the eligible types; actual
labels are selected only for the outliers and the reserve, and the
remainder is represented implicitly. (Kernel columns are nonzero and
pools only shrink or gain coordinates, so zero columns never occur; were
they present, they would be removed at once by selecting exactly those
with negative cost.)
\end{remark}

\paragraph{Classical primitives.}
We use Smith and Hermite normal forms of fixed-rank integer lattices
whose generators have $O_{\rm dim}(\log N)$ bits, and LLL
reduction~\cite{LLL82}; the point-or-thin primitive of
Lemma~\ref{lem:point-thin}, which is based on Lenstra's
method~\cite{Lenstra83}; and Megiddo's multidimensional
search~\cite{Megiddo84,Zemel84} for linear programs with a fixed number
of coupling rows (Lemma~\ref{lem:typed-lp}). The latter replaces the
ellipsoid method~\cite{GLS88}: its number of operations does not depend
on the bit length of the costs. No optimization result is imported
from~\cite{CMZ26}.

\subsection{Global parameters}\label{sec:params}

All scales are fixed before the recursion and depend only on $d$, $N$,
and the fixed parameter $\delta$; none depends on the costs. Let
$\Psi_0=C_dN^d$ be the bound of Lemma~\ref{lem:volume} on the root
normalized volume. Fix
\begin{equation}\label{eq:delta}
 0<\delta\le\min\Bigl\{\frac1{100},\ \frac{\eps}{2(d+2)}\Bigr\},
 \qquad \nu=\lceil N^\delta\rceil,
\end{equation}
and set
\begin{equation}\label{eq:params-a}
 \T_{\rm lift}=\lceil 4d/\delta\rceil+2,\qquad
 D=d+\T_{\rm lift},\qquad
 J=\lceil\log_2\Psi_0\rceil+D+1 .
\end{equation}
Fix the structural exponent $\eta$ of Theorem~\ref{thm:gap} so small that
\begin{equation}\label{eq:eta}
 2\eta(\delta D+d)<\delta,
\end{equation}
and apply the theorem with $r_0=D$. Let $W=W(d,\delta)$ be the maximum
of the width constants $W_m$ of Lemma~\ref{lem:point-thin} over all $m$
up to the largest GAP rank permitted by Theorem~\ref{thm:gap}.
Set
\begin{equation}\label{eq:tau}
 \tau=\bigl\lceil C_{\rm dim}\,\nu\log^3(N+2)\bigr\rceil .
\end{equation}
Next let
\begin{equation}\label{eq:params-k}
 k_0=\bigl\lceil C_{\rm dim}\,\nu\log^{c_0}(N+2)\bigr\rceil,\qquad
 \kappa=3k_0,\qquad E=(260+6J)k_0,
\end{equation}
where $c_0\ge3$ and the constant are chosen so that $k_0\ge\tau\ge\nu$, $k_0$
bounds $|I_{\rm out}|+|I_{\rm res}|$ in every application of
Theorem~\ref{thm:gap}, and $\kappa$ exceeds the size threshold
$c'_{\eta,D}\tau$ of that theorem. The packing and buffer scales are
\begin{equation}\label{eq:params-b}
 \rho_*=2^{-10},\qquad
 \alpha=\frac{\rho_*}{1000(J+1)},\qquad
 K=\frac{E}{\alpha},\qquad
 \gamma=3\alpha,\qquad
 C_\gamma=\frac{1+\gamma}{1-\gamma},
\end{equation}
\begin{equation}\label{eq:rho-schedule}
 \Delta_\rho=\gamma+(C_\gamma-1)\rho_*+\frac EK,\qquad
 \rho_j=\frac{\rho_*}4+(J-j)\Delta_\rho\qquad(0\le j\le J).
\end{equation}
Let $\Athin,\Alift\ge1$ be the mass factors of Lemma~\ref{lem:thin} and
Proposition~\ref{prop:lift}; they satisfy
$\Athin,\Alift\le C_{\rm dim}\alpha^{-C_{\rm dim}}$, depend only on $d$,
$\delta$, and $\alpha$, and do not depend on the clipping scales below,
so the following definition is not circular. Using $c_{\rm thin}=1$
from Lemma~\ref{lem:thin}, set
\begin{equation}\label{eq:schedule}
\begin{gathered}
 H_{\min}=12J(\Athin+1)DK,\\
 \Delta_H=2E+2\gamma H_{\min},\qquad
 H_j=H_{\min}+(J-j)\Delta_H\qquad(0\le j\le J).
\end{gathered}
\end{equation}
Finally $\zeta=\gamma+C_\gamma\rho_*$.

\begin{lemma}[Scale inequalities]\label{lem:scales}
There is $N_{\rm thr}=N_{\rm thr}(d,\delta)$ such that for $N\ge N_{\rm thr}$:
\begin{enumerate}[label=\textup{(S\arabic*)}]
\item\label{S:large} $\Psi_0\le N^{2d}$, $2(D+1)W\le N^{\delta/2}$,
$\tau\ge W_{\rm loc}^{2\eta}$ and $h\ge\nu$ in every application of
Theorem~\ref{thm:gap} to a pool of rank $r\le D$ with $\Psi\le\Psi_0$, $c_\eta h\ge k$ for all permitted GAP ranks $k$,
and $\nu\le k_0$;
\item\label{S:sched} $\rho_*/4\le\rho_j<\rho_*/2$, $E/K=\alpha$,
$\alpha/2\le\rho_j$, $E/(2K)\le\rho_j$, and $E\le H_{\min}$;
\item\label{S:own} $\rho_j\ge\gamma+C_\gamma\rho_{j+1}+E/K$ and
$H_j-H_{j+1}\ge\gamma H_j+E$ for $0\le j<J$;
\item\label{S:core} $2\bigl(K/128-2k_0\bigr)>\kappa(J+D)$;
\item\label{S:collar} $(1/64-\zeta)K-4k_0\ge E+16(J+1)k_0$;
\item\label{S:lift} with $\mu=(1/128-\gamma)K-2k_0-2E$ and
$R_{\rm p}=(1+2\gamma)K+E+2k_0$: $\mu\ge2k_0$ and
$\mu-k_0\ge(R_{\rm p}+k_0)/254$;
\item\label{S:drift} $128(1+2\gamma)k_0+\kappa J+k_0\le E/2$.
\end{enumerate}
Moreover $k_0=O_{\rm dim}(N^\delta\log^{c_0}N)$,
$J,\alpha^{-1}=O_{d,\delta}(\log N)$,
$K=O_{\rm dim}(N^\delta\log^{c_0+2}N)$, and $H_0=\widetilde O_{d,\delta}(\nu)$.
\end{lemma}

\begin{proof}
The verification is given in Appendix~\ref{app:scales}.
\end{proof}

(The hypotheses $r\le D$ and $\Psi\le\Psi_0$ are supplied along every
path by the prospective argument of Lemma~\ref{lem:depth-bounds}.)
For $N<N_{\rm thr}(d,\delta)$ the kernel has $n_0\le C_{\rm dim}$ items, and
exhaustive enumeration of its $2^{n_0}$ assignments, with exact costs,
finds an optimum directly in $C_{\rm dim}$ operations on words, without
any binary search. Henceforth $N\ge N_{\rm thr}$.

%======================================================================
\section{Value functions, normalization, and normalized volume}\label{sec:value}

For a pool $I$ define the convex piecewise-linear function
\[
 f_I(\vt)=\min\bigl\{\vc_I^{\T}\vx:\ A_I\vx=\vt,\ \vzero\le \vx\le\mathbf 1\bigr\}
\]
($+\infty$ outside $Z(I)$), and the gap $g(\vt,s)=s-f_I(\vt)$ of a query.
Since all columns are nonzero, a pool of rank zero is empty; its only
feasible target is $\vzero$, with value $0$, so this case is solved
directly.

\paragraph{Temporary normalization.}
For a row vector $\vpi$ define, in the original orientation of each item,
\begin{equation}\label{eq:normalize}
 e_i(\vpi)=[\,\vpi \va_i>c_i\,],\quad \sigma_i=1-2e_i(\vpi),\quad
 \bar\va_i=\sigma_i\va_i,\quad \bar c_i=|c_i-\vpi \va_i|.
\end{equation}
Here $[\cdot]$ is the Iverson bracket ($1$ if the condition holds, $0$ otherwise), so ties choose $e_i=0$; write $\ve=\ve(\vpi)=(e_i)_i$. Substituting $x_i=e_i+\sigma_iz_i$ turns a
query into
\[
 \bar A \vz=\vt-A_I\ve,\qquad \bar\vc^{\T}\vz<s-\vc_I^{\T}\ve-\vpi(\vt-A_I\ve)
\]
with nonnegative costs $\bar\vc$; $\vz=\vzero$ represents $\vx=\ve$. The quantity on
the right is the \emph{normalized allowance}. Normalizations are local
computations; stored states keep original orientations and costs.
Normalization changes neither $P_I$ nor any lattice. We write
$\cT_Y(b)=\{\sum_{i\in Y}\bar\va_iz_i:\vzero\le \vz\le\mathbf 1,\ \sum_{i\in Y}\bar c_iz_i\le b\}$
for the affordable body of a subpool $Y$ in normalized coordinates.

\paragraph{Normalized volume.}
For $r=\rank I>0$ put
\begin{equation}\label{eq:psi}
 \Psi(I)=\frac{r!}{2^r}\,\frac{\vol_r(P_I)}{\covol(\Lambda_I)},
\end{equation}
and $\Psi(I)=1$ at rank zero.

\begin{lemma}[Normalized volume]\label{lem:volume}
\leavevmode
\begin{enumerate}
\item $\Psi(I)\ge1$, and every determinant of $r$ independent columns, in
coordinates of $\Lambda_I$, is at most $\Psi(I)$ in absolute value.
\item If $R\subseteq I$ has rank $r$ and $q=[\Lambda_I:\Lambda_R]$, then
$\Psi(R)\le\Psi(I)/q$.
\item If $R\subseteq I$ has smaller rank, then $\Psi(R)\le\Psi(I)$.
\item For $\vt\in\Lambda_I$, every optimal vertex for $f_I(\vt)$ has a common
denominator at most $\Psi(I)$. Hence every positive gap is at least
$1/\Psi(I)$.
\item Initially $\Psi\le\Psi_0:=C_dN^d$.
\end{enumerate}
\end{lemma}

\begin{proof}
(1) For independent columns $\vv_1,\dots,\vv_r$, the crosspolytope
$\conv\{\pm \vv_j\}\subseteq P_I$ has volume $2^r|\det(\vv_j)|/r!$, and
$|\det(\vv_j)|$ is a positive integer multiple of $\covol(\Lambda_I)$.
(2) $P_R\subseteq P_I$ and $\covol\Lambda_R=q\covol\Lambda_I$.
(3) Let $V=\operatorname{span}R$, $k'=\dim V$, $\Gamma=\Lambda_I\cap V$.
Choose $r-k'$ columns of $I$ independent modulo $V$; after a
volume-preserving shear they lie in $V^\perp$, and their projected
determinant $\Delta_\perp$ is at least the covolume $\covol_\perp$ of the
projection of $\Lambda_I$ to $V^\perp$. The hull
of $P_R$ and these columns with signs lies in $P_I$ and has volume
$\tfrac{2^{r-k'}k'!}{r!}\vol_{k'}(P_R)\,\Delta_\perp$. Since
$\covol\Lambda_I=\covol\Gamma\cdot\covol_\perp$ ($\Gamma$ is primitive) and
$\covol\Lambda_R\ge\covol\Gamma$, this gives $\Psi(I)\ge\Psi(R)$ (for
$k'=0$ use (1)).
(4) At a vertex the fractional variables have independent columns; extend
them to a basis of $r$ columns. The other variables are $0$ or $1$, so the
residual right-hand side lies in $\Lambda_I$, and Cramer's rule gives a
denominator dividing the basis determinant, which is at most $\Psi(I)$ by
(1). Costs and allowances are integers.
(5) $P_I$ lies in the intersection of $[-N,N]^d$ with the $r$-dimensional
span of $I$, whose $r$-volume is at most $(2\sqrt d\,N)^r\le C_dN^d$, and
integral lattices have covolume at least one.
\end{proof}

\begin{lemma}[Types and lattice counts]\label{lem:types}
Let $I$ have rank $r>0$.
\begin{enumerate}
\item $T_I\le(2^r+r)\Psi(I)$.
\item For $s\ge1$, every translate of $sP_I$ contains at most
$C_rs^r\Psi(I)$ points of $\Lambda_I$.
\item For $s>0$, the body $sP_I$ contains at least $(s^r/r!)\Psi(I)$ points
of $\Lambda_I$.
\end{enumerate}
\end{lemma}

\begin{proof}
(1) Triangulate $P_I$ using all its lattice points as vertices. If there
are $t$ full-dimensional simplices and $M$ vertices, a spanning tree of
the facet-adjacency graph of the simplices shows $M\le t+r$. Each simplex
has volume at least $\covol(\Lambda_I)/r!$, so
$M\le r!\vol_r(P_I)/\covol(\Lambda_I)+r=2^r\Psi(I)+r$. Every type is one
of these lattice points.
(2) If a translate contains a lattice point, subtracting it injects its
lattice points into $2sP_I$. The integer hull of $2sP_I\cap\Lambda_I$ is
full-dimensional, as it contains the columns, and has volume at most
$\vol_r(2sP_I)$; apply the triangulation bound of (1) to it.
(3) The translates of $(s/2)P_I$ by a fundamental cell of $\Lambda_I$
contain on average $\vol_r((s/2)P_I)/\covol(\Lambda_I)=(s^r/r!)\Psi(I)$
lattice points, so some translate contains at least that many.
Subtracting one of them injects the others into $sP_I\cap\Lambda_I$,
since $(s/2)P_I-(s/2)P_I=sP_I$. (This is a counting argument only; the
algorithm enumerates deterministically.)
\end{proof}

%======================================================================
\section{Buffers, packing, and deviations}\label{sec:buffers}

\subsection{Buffers and mass}

For a query $\vq=(\vt,s)$ with $g(\vq)>0$ let
\begin{equation}\label{eq:macbeath}
 M_I(\vq)=(\epi f_I-\vq)\cap(\vq-\epi f_I),\qquad
 Q_I(\vq)=\proj_{\rm weight}M_I(\vq).
\end{equation}
Equivalently, $\vu\in Q_I(\vq)$ iff there are fractional $\vx^\pm\in[0,1]^I$ with
$A_I\vx^\pm=\vt\pm \vu$ and $\vc_I^{\T}(\vx^++\vx^-)\le2s$, that is,
$f_I(\vt+\vu)+f_I(\vt-\vu)\le2s$. For a batch $\cW$ at
arithmetic depth $j$ the \emph{buffered set} and \emph{mass} are
\begin{equation}\label{eq:mass}
 \cD_j(\cW)=\Lambda_I\cap\bigcup_{(\vt,s)\in\cW}
 \Bigl[\vt+\bigl(\rho_jQ_I(\vt,s)\cap H_jP_I\bigr)\Bigr],\qquad
 w=|\cD_j(\cW)|.
\end{equation}
Since targets are distinct after merging, $|\cW|\le w$. The mass is an
analysis quantity; the algorithm does not construct the union. Only
nonempty batches of fractionally feasible positive-gap queries are
emitted as nodes, so every node has $w\ge1$; filtering candidate child
queries is charged to the parent that produced them.
Lemma~\ref{lem:transfer} gives the \emph{epigraph embedding}
\begin{equation}\label{eq:embedding}
 Q_R(\vt-A_X\vY,\ s-\vc_X^{\T}\vY)\subseteq Q_I(\vt,s),
\end{equation}
in old coordinates, also after new coordinates have been adjoined and after
maximum-allowance merging (use the attaining source): extend both
witnesses of a residual pair by the same $\vY$.

Targets on the relative boundary of $Z(I)$ are handled first
(Section~\ref{sec:boundary-thin}); for the remaining targets $Q_I(\vq)$ is
full-dimensional in the span of $I$ and $M_I(\vq)$ has dimension $r+1$.

\subsection{Body comparisons}

\begin{lemma}[Macbeath comparison]\label{lem:macbeath}
Let $\mathcal E$ be convex and $M_q=(\mathcal E-\vq)\cap(\vq-\mathcal E)$. If
$\vz-\vq\in\theta M_q$ with $0\le\theta<1$ then
$(1-\theta)M_q\subseteq M_z\subseteq(1+\theta)M_q$.
For $\mathcal E=\epi f_I$ this gives
\begin{equation}\label{eq:local-comparison}
 (1-\theta)Q_q\subseteq Q_z\subseteq(1+\theta)Q_q,\qquad
 (1-\theta)g(\vq)\le g(\vz)\le(1+\theta)g(\vq).
\end{equation}
\end{lemma}

\begin{proof}
For $\theta=0$ there is nothing to prove. For the first inclusion mix
$\vq\pm \vu\in\mathcal E$ with $\vq+(\vz-\vq)/\theta\in\mathcal E$; for the
second mix $\vz\pm \vu$ with $\vq-(\vz-\vq)/\theta$, in proportions
$1:\theta$. (The first mixture uses weights $1-\theta$ on
$\vq\pm\vu$ and $\theta$ on the other point, giving
$\vz\pm(1-\theta)\vu$; the second uses weights $\frac1{1+\theta}$ and
$\frac\theta{1+\theta}$, giving $\vq\pm\vu/(1+\theta)$.) Projection gives the $Q$ inclusions. Intersecting the body
inclusions with the vertical axis gives the gap bounds, since
$M_q\cap(\{\vzero\}\times\R)=\{\vzero\}\times[-g(\vq),g(\vq)]$.
\end{proof}

At a fixed target $\vz$, the body $Q_I(\vz,f_I(\vz)+g)$ is the sublevel
set $\{\vu:f_I(\vz+\vu)+f_I(\vz-\vu)-2f_I(\vz)\le2g\}$ of a convex,
nonnegative function vanishing at the origin. Hence it is monotone in
$g$, and for $\lambda\ge1$
\begin{equation}\label{eq:gap-dilation}
 Q_I(\vz,f_I(\vz)+\lambda g)\subseteq\lambda Q_I(\vz,f_I(\vz)+g).
\end{equation}

\subsection{Center-capture packing}

Fix a node after the boundary reduction of Lemma~\ref{lem:boundary},
and write $f=f_I$, $P=P_I$, $H=H_j$, $Q_q=Q_I(\vq)$, $g_q=g(\vq)$, and
$S_q=Q_q\cap HP$. Order the queries by nonincreasing gap, breaking ties
deterministically. Repeatedly select the first unassigned query $\vv$ as
an \emph{anchor}, enumerate the lattice targets in $\vt_v+\gamma S_v$,
and assign to $\vv$ every still unassigned query with one of these
targets (including $\vv$ itself).

\begin{lemma}[Center-capture packing]\label{lem:packing}
Every query $\vq$ assigned to an anchor $\vv$ satisfies
\begin{equation}\label{eq:anchor}
 g_q\le g_v,\qquad \vt_q-\vt_v\in\gamma S_v,\qquad
 Q_q\subseteq C_\gamma Q_v,\qquad S_q\subseteq C_\gamma S_v .
\end{equation}
The inner bodies $\vt_v+\frac\alpha2S_v$ of distinct anchors are pairwise
disjoint, and
\begin{equation}\label{eq:packing-count}
 \sum_v\bigl|\bigl(\vt_v+\tfrac\alpha2S_v\bigr)\cap\Lambda_I\bigr|\le w .
\end{equation}
\end{lemma}

\begin{proof}
The first two relations hold by construction. Lift $\vt_q-\vt_v\in\gamma Q_v$
to a point $\vq'\in\vv+\gamma M_I(\vv)$ with target $\vt_q$. By
\eqref{eq:local-comparison}, $g_{q'}\ge(1-\gamma)g_v\ge(1-\gamma)g_q$, so
by monotonicity and \eqref{eq:gap-dilation} at the common target
$\vt_q$, $Q_q\subseteq(1-\gamma)^{-1}Q_{q'}\subseteq C_\gamma Q_v$.
Intersecting with $HP$, which is convex and contains the origin, gives
the last inclusion.

Suppose the bodies $\vt_v+\alpha S_v$ and $\vt_{v'}+\alpha S_{v'}$ of an
earlier anchor $\vv$ and a later anchor $\vv'$ (so $g_{v'}\le g_v$) meet
at $\vz$. Lift $\vz$ separately to points of $\vv+\alpha M_I(\vv)$ and of
$\vv'+\alpha M_I(\vv')$; by \eqref{eq:local-comparison}, the gap of the
second lift is at most $(1+\alpha)/(1-\alpha)$ times that of the first.
Thus \eqref{eq:local-comparison}
and \eqref{eq:gap-dilation} at their common target give
$S_{v'}\subseteq\bigl(\frac{1+\alpha}{1-\alpha}\bigr)^2S_v$. Hence
$\vt_{v'}-\vt_v\in\alpha\bigl(1+(\frac{1+\alpha}{1-\alpha})^2\bigr)S_v
\subseteq3\alpha S_v=\gamma S_v$, so $\vv'$ would have been assigned to $\vv$,
a contradiction. Thus the bodies $\vt_v+\alpha S_v$ have disjoint
relative interiors, and the closed half-sized bodies are disjoint. Each
lies in the buffer of its anchor since $\alpha/2\le\rho_j$
(\ref{S:sched}), and every anchor belongs to $\cW$; this
proves \eqref{eq:packing-count}.
\end{proof}

An assigned query may have a much smaller gap than its anchor: the
locality \eqref{eq:anchor} is one-sided. The following lemma supplies
everything the reductions need from the anchor.

\begin{lemma}[Allowance and deviations]\label{lem:deviation}
Let $\vv$ be an anchor whose target lies in the relative interior of
$Z(I)$, let $\vpi\in\partial f_I(\vt_v)$, $a=g_v$, and let $\vx_0$ be an
optimal fractional vertex at $\vt_v$. For every query $\vq=(\vt_q,s_q)$
assigned to $\vv$, its \emph{normalized allowance} satisfies
\begin{equation}\label{eq:allowance}
 0<g_q\le a_q:=s_q-f_I(\vt_v)-\vpi(\vt_q-\vt_v)\le(1+2\gamma)a .
\end{equation}
Moreover $f_I(\vt)=\vpi\vt+\sum_{i\in I}\min(0,c_i-\vpi\va_i)$ at
$\vt=\vt_v$, so $a_q=s_q-\vpi\vt_q-\sum_i\min(0,c_i-\vpi\va_i)$, and every
$(\vu,h)\in M_I(\vv)$ satisfies $|h-\vpi\vu|\le a$. For every binary
solution $\vx$ of $\vq$,
\begin{equation}\label{eq:deviation}
 \sum_{i\in I}|x_i-x_{0i}|\,[-\va_i,\va_i]\subseteq3Q_v,
 \qquad
 \sum_{i\in I}|x_i-x_{0i}|\,\|\va_i\|_{Q_v}\le3r.
\end{equation}
\end{lemma}

\begin{proof}
By LP duality,
$f_I(\vt)=\max_{\vpi'}\bigl[\vpi'\vt+\sum_i\min(0,c_i-\vpi'\va_i)\bigr]$, and
the subgradients at $\vt$ are exactly the maximizers; this gives the
formula at $\vt_v$, and the bound $|h-\vpi\vu|\le a$ follows from the
supporting inequality at $\vt_v\pm\vu$. Put
$V(\vt)=f_I(\vt)-f_I(\vt_v)-\vpi(\vt-\vt_v)\ge0$. For $\vu\in Q_v$ the
paired witnesses give $V(\vt_v+\vu)+V(\vt_v-\vu)\le2a$; as
$\vt_q-\vt_v\in\gamma Q_v$, convexity gives $V(\vt_q)\le2\gamma a$. Now
$a_q=g_q+V(\vt_q)$ and $g_q\le a$, proving \eqref{eq:allowance}.

Write $\mathbf h=\vt_q-\vt_v$ and choose a fractional $\vz^-$ with target
$\vt_v-\mathbf h/\gamma$ and normalized cost at most $2a$; it exists
because $\mathbf h/\gamma\in Q_v$ and the paired witnesses have
nonnegative normalized costs summing to at most $2a$. For
$|\delta_i|\le|x_i-x_{0i}|$ the vectors
$\vy^\pm=(\vx+\vx_0\pm\vdelta+\gamma\vz^-)/(2+\gamma)$ are box-feasible,
with targets $\vt_v\pm A_I\vdelta/(2+\gamma)$. The normalized cost of
$\vx_0$ is zero, that of $\vx$ is less than $a_q$, and normalized cost is
affine in the assignment, so the average normalized cost of $\vy^\pm$ is
less than $(a_q+2\gamma a)/(2+\gamma)\le(1+4\gamma)a/(2+\gamma)<a$. Hence
$A_I\vdelta\in(2+\gamma)Q_v\subseteq3Q_v$. For the second inequality take
$\vq_1,\dots,\vq_r\in Q_v$ of maximal determinant with dual functionals
$\vlambda_j$; then $|\vlambda_j|\le1$ on $Q_v$ and
$\|\va\|_{Q_v}\le\sum_j|\vlambda_j(\va)|$. Apply $\vlambda_j$ to the first
inclusion and sum over $j$.
\end{proof}

\begin{lemma}[Folding]\label{lem:folding}
Let $\Lambda$ be a lattice of rank $r$, $\vt\in\Lambda$, $S$ a symmetric
convex body, and $s>0$. Then
\[
 |(\vt+sS)\cap\Lambda|\le m^r\,|(\vt+\tfrac\alpha2S)\cap\Lambda|,\qquad
 m=\lceil4s/\alpha\rceil.
\]
\end{lemma}

\begin{proof}
In each class modulo $m\Lambda$ fix a point $\vz$ of $\vt+sS$. The map
$\vx\mapsto \vt+(\vx-\vz)/m$ is injective on the class, takes values in
$\Lambda$, and $(\vx-\vz)/m\in(2s/m)S\subseteq\frac\alpha2S$.
\end{proof}

With $s=\gamma=3\alpha$, so $m=12$, Lemma~\ref{lem:folding} shows that the
capture windows $\vt_v+\gamma S_v$ have at most $12^r$ times as many
lattice points as the inner bodies, so by \eqref{eq:packing-count} all capture
windows together have $O_r(w)$ lattice points.

%======================================================================
\section{Boundary and thin reductions}\label{sec:boundary-thin}

\begin{lemma}[Boundary targets]\label{lem:boundary}
Let $\vt$ lie in the relative interior of a proper face $\mathcal F$ of
$Z(I)$. A supporting functional of $\mathcal F$ forces every column on
which it is nonzero to a fixed endpoint in every fractional solution at
$\vt$; fixing these items gives an exact state of smaller rank, with
$Q$ unchanged up to translation; the child inherits the depth $j$ and lift
count $\ell$. Grouping queries by their minimal face, the face children
have total mass at most the parent mass.
\end{lemma}

\begin{proof}
Forcing is the usual complementary-slackness argument and is detected by
LP. Both opposite points defining $Q_I(\vq)$ lie in $\mathcal F$, so $Q$ is
unchanged after the fixed translation. Since $\rho_j<1$, buffered points of
the child, translated back, lie in $\relint\mathcal F$; distinct minimal
faces have disjoint relative interiors, and the child clipping uses
$P_{\rm child}\subseteq P_I$. Hence translated child buffers are disjoint
subsets of $\cD_j(\cW)$.
\end{proof}

\paragraph{The rounded body and the thin/deep test.}
For an interior anchor $\vv$ we construct an explicit symmetric
parallelepiped $D_v$ with
\begin{equation}\label{eq:rounded}
 D_v\subseteq Q_v\subseteq c_rD_v,\qquad c_r=2r .
\end{equation}
Find an approximate maximum-determinant basis
$\vv_1,\dots,\vv_r\in Q_v$ by support queries: whenever some point of
$Q_v$ has a coordinate of absolute value greater than $2$ in this basis,
replace the corresponding basis vector; each replacement at least
doubles the determinant. The basis vectors are optimizers returned by the
oracle of Lemma~\ref{lem:typed-lp}, with $O_{\rm dim}(\log(n_0N(c_{\max}+2)))$
bits, so as in the proof of Lemma~\ref{lem:point-thin} there are
$O_{\rm dim}(\log(n_0N(c_{\max}+2)))=O_{\rm dim}(\log(N+2)\,L_c)$
replacements, each found by $2r$ support queries. For $B=[\vv_1\cdots\vv_r]$ put
$D_v=\frac1rB[-1,1]^r\subseteq\conv\{\pm\vv_i\}\subseteq Q_v$; the stopping
rule gives $Q_v\subseteq2rD_v$. The gauge and support function of $D_v$
are explicit. The anchor is \emph{deep} if $KP_I\subseteq D_v$, and
\emph{thin} otherwise. A deep anchor satisfies $KP_I\subseteq D_v\subseteq
Q_v$, which is all that Sections~\ref{sec:depth}--\ref{sec:lift} use.

\begin{lemma}[Thin anchors]\label{lem:thin}
Let $\vv$ be a thin interior anchor and $S_v=Q_v\cap H_jP_I$. There is an
exact reduction of the queries assigned to $\vv$ into children whose pools
are subsets of $I$, inheriting the depth $j$. Summed over all thin
anchors,
\begin{equation}\label{eq:thin-mass}
 \sum w_{\rm child}\le\Athin w,\qquad \Athin\le C_{\rm dim}\alpha^{-C_{\rm dim}},
\end{equation}
the same bound covers all intermediate table windows, and every child of
the same rank satisfies
$\Psi_{\rm child}\le c_{\rm thin}(rK/H_j)\Psi(I)$, where one may take
$c_{\rm thin}=1$.
\end{lemma}

\begin{proof}
\emph{Reduction.} Let $I_{\rm fr}$ be the (at most $r$) fractional
coordinates of the anchor vertex $\vx_0$; branch on their $2^{|I_{\rm fr}|}$
binary values. Call $i\notin I_{\rm fr}$ \emph{light} if $H_j\va_i\in D_v$
and \emph{heavy} otherwise; this is a whole-type test. Heavy items are
transferred by a removed-item table that records deviations
$Y_i-x_{0i}$ from the anchor's LP endpoints; light items form the child
pool. By \eqref{eq:rounded}, $\|\cdot\|_{D_v}\le c_r\|\cdot\|_{Q_v}$, so
Lemma~\ref{lem:deviation} gives
$\sum_i|x_i-x_{0i}|\,\|\va_i\|_{D_v}\le3rc_r$ for every solution of an
assigned query. A heavy item has $\|\va_i\|_{D_v}>1/H_j$, so fewer than
$3rc_rH_j$ heavy endpoints change, and every partial sum of the heavy
deviations lies in $3Q_v\cap3rc_rH_jP_I\subseteq3rc_rS_v$. Restricting
all table prefixes to this window loses no solution, and dominance is
exact (Lemma~\ref{lem:transfer}).

The LP endpoints, rather than the price endpoints, are the right
reference here. For example, take $n$ unit-weight items of cost zero
and target $n-1$: at any positive allowance the Macbeath body
$Q_v$ is $[-1,1]$, but the price $\vpi=0$ has the all-zero preferred assignment,
at distance $n-1$ from every solution, whereas an LP vertex already
selects $n-1$ items.

\emph{Mass.} Each light column satisfies $H_j\va_i\in D_v\cap H_jP_I\subseteq S_v$,
so $H_jP_{\rm child}\subseteq S_v$. Translate child targets by the integral
vector $A_{I_{\rm fr}}\vz_{\rm fr}+A_{\rm heavy}\vx_{0,\rm heavy}$. A
translated buffered point is then $\vt_v$ plus: the query offset (in
$\gamma S_v$), a heavy deviation (in $3rc_rS_v$), and a child buffer
displacement (in $S_v$, by the embedding and the inclusion above). All of
them lie in $\vt_v+(3rc_r+2)S_v$. The child lattices are sublattices of
$\Lambda_I$ after this translation. By Lemma~\ref{lem:folding} with
$s=3rc_r+2$, every table window and every child buffer for $\vv$ has at
most $C_r\alpha^{-r}|(\vt_v+\frac\alpha2S_v)\cap\Lambda_I|$ points, and
there are $2^r$ branches. Summing over anchors with
\eqref{eq:packing-count} gives \eqref{eq:thin-mass}.

\emph{Volume.} Since $KP_I\not\subseteq D_v$ there are a column $\va$ of
$I$ and a facet functional $\vlambda$ of $D_v$ with
$h_{D_v}(\vlambda)<Kh_{P_I}(\vlambda)$. Every light column has
$|\vlambda(\va_i)|\le h_{D_v}(\vlambda)/H_j<(K/H_j)h_{P_I}(\vlambda)$. For a
symmetric convex $r$-body the central slab of relative half-width
$\delta'\le1$ has volume at most $r\delta'$ times the body: the slab is
at most the maximal (central) section times its width, and the double
cone over that section lies in the body and has volume $2/r$ times
section times half-width. A same-rank child has covolume at least
$\covol\Lambda_I$, which gives the volume bound with $c_{\rm thin}=1$.
\end{proof}

%======================================================================
\section{Affordable depth and drift}\label{sec:depth}

Let $\vv$ be a \emph{deep} anchor: $Q:=Q_I(\vv)\supseteq D_v\supseteq KP$, $P=P_I$. Fix
$\vpi\in\partial f_I(\vt_v)$ and normalize by \eqref{eq:normalize}. Let
$\bar \vt_v=\vt_v-A_I\ve(\vpi)$; by Lemma~\ref{lem:deviation} its
normalized fractional value is $0$ and its normalized allowance is $g_v$.

\begin{lemma}[Affordable depth]\label{lem:depth}
Let $Y\subseteq I$ contain every item of normalized cost at most
$\theta>0$. Let $H\ge K$, $S=Q\cap HP$, $\Pref=S/K$. For
$0\le\zeta'\le\xi\le1$ and $\mu=(\xi-\zeta')K-2\xi g_v/\theta\ge0$,
\begin{equation}\label{eq:depth}
 \bar \vt_v+\zeta'Q+\mu\Pref\subseteq\cT_Y(2\xi g_v),\qquad P\subseteq\Pref.
\end{equation}
If $Y'\subseteq Y$ has $|Y'|\le e\le\mu$, then \eqref{eq:depth} holds with
$\cT_{Y\setminus Y'}(2\xi g_v)$ in place of $\cT_Y(2\xi g_v)$ and $\mu-e$
in place of $\mu$.
\end{lemma}

\begin{proof}
For $\vu\in Q$ the paired witnesses have total normalized cost at most
$2g_v$, hence the one at $\bar \vt_v+\vu$ costs at most $2g_v$; mixing with
weight $\xi$ with a zero-cost witness at $\bar \vt_v$ gives
$\bar \vt_v+\xi Q\subseteq\cT_I(2\xi g_v)$. In a representation of cost at
most $2\xi g_v$, items outside $Y$ have total coefficient at most
$2\xi g_v/\theta$, so $\cT_I(2\xi g_v)\subseteq\cT_Y(2\xi g_v)+
(2\xi g_v/\theta)P$. In support functions,
\[
 h_{\cT_Y(2\xi g_v)}\ge\vlambda(\bar \vt_v)+\zeta'h_Q+(\xi-\zeta')h_Q
 -(2\xi g_v/\theta)h_P .
\]
Since $KP\subseteq S$ we have $P\subseteq\Pref$ and $Q\supseteq K\Pref$, so
$h_Q\ge Kh_{\Pref}$ and $h_P\le h_{\Pref}$; this is the support function
of the left side of \eqref{eq:depth}. Deleting a set $Y'$ of $e$ items with
nonnegative costs gives $\cT_Y(b)\subseteq\cT_{Y\setminus Y'}(b)+eP\subseteq
\cT_{Y\setminus Y'}(b)+e\Pref$; subtracting $e\,h_{\Pref}$ in the support
inequality gives the last assertion.
\end{proof}

We use the cheap threshold and budget
\begin{equation}\label{eq:theta0}
 \theta_0=\frac{g_v}{128k_0},\qquad \varrho=\frac{g_v}{32},
 \qquad\text{so that}\quad \frac{\varrho}{2\theta_0}=2k_0,\quad
 \frac{(1+2\gamma)g_v}{\theta_0}=128(1+2\gamma)k_0 .
\end{equation}

\begin{lemma}[Drift of removed assignments]\label{lem:drift}
Let $X\subseteq I$ contain at most $e'$ items of normalized cost at most
$\theta_0$, let $\vq$ be a query assigned to $\vv$, and let $\vY$ be an
assignment of $X$ whose normalized cost is less than $a_q$ (this holds, in
particular, whenever the residual of $\vq$ has positive normalized allowance,
for instance positive gap). Then for every $X'\subseteq X$
\[
 A_{X'}\bigl(\vY-\ve_{X'}(\vpi)\bigr)\in\bigl(128(1+2\gamma)k_0+e'\bigr)P.
\]
In particular, for $e'\le\kappa J+k_0$ every such partial displacement,
including every prefix of a removed-item table, lies in $\frac E2P$.
\end{lemma}

\begin{proof}
The normalized cost of $\vY$ is less than $a_q\le(1+2\gamma)g_v$
(Lemma~\ref{lem:deviation}), and normalized costs are nonnegative.
Each deviation $Y_i\ne e_i(\vpi)$ on an item of cost greater than $\theta_0$
costs more than $\theta_0$, so there are fewer than
$128(1+2\gamma)k_0$ of them, and at most $e'$ deviations on cheap items. A
partial sum over $X'$ consists of at most this many terms $\pm \va_i\in P$.
The last assertion is \ref{S:drift}.
\end{proof}

%======================================================================
\section{Deep anchors I: canonical cores and arithmetic children}
\label{sec:arithmetic}

Subproblems created at different supporting prices must be compared. We do
this with a uniquely determined subpool whose lattice survives every
deletion of $\kappa$ items.

\subsection{The canonical resilient core}

\begin{definition}
A subpool $Y'\subseteq Y$ is \emph{$\kappa$-resilient} if
$\lat(Y'\setminus F')=\lat(Y')$ for all $F'\subseteq Y'$ with
$|F'|\le\kappa$. The union of all $\kappa$-resilient subpools of $Y$ is
denoted $\cR_\kappa(Y)$.
\end{definition}

\begin{lemma}[Canonical core]\label{lem:core}
$\cR_\kappa(Y)$ is the largest $\kappa$-resilient subpool of $Y$, it is
monotone in $Y$, and
$\cR_\kappa(Y)=\{i\in Y:\va_i\in\lat\cR_\kappa(Y)\}$. It is obtained by
repeating: delete at most $\kappa$ items if this strictly shrinks the
generated lattice, then restore all items of $Y$ lying in the new lattice.
\end{lemma}

\begin{proof}
If $Y_\iota$ are resilient and $|F'|\le\kappa$, each $Y_\iota\setminus F'$
still generates $\lat Y_\iota$, so the union is resilient; this gives
existence, uniqueness, and monotonicity. Adding an item lying in the
lattice of a resilient pool keeps it resilient, which gives the closure
identity. A resilient $S\subseteq Y$ survives every step of the procedure:
after a permitted deletion its survivors still generate $\lat S$, so all of
$S$ is restored. At termination the current pool is resilient and contains
every resilient subpool.
\end{proof}

\begin{lemma}[Constructing a deep core]\label{lem:core-alg}
Let $Y\subseteq I$ with $\vz+\mu P\subseteq Z(Y)$ and
$2\mu>\kappa(J+r)$. Then the procedure of Lemma~\ref{lem:core} keeps rank
$r$, performs at most $J$ strict reductions, and excludes at most $\kappa J$
items of $Y$. Using the cost-ordered type lists, it runs in time
$\widetilde O_{d,\delta}\bigl((\kappa+1)(J+1)T_I\Psi(I)\bigr)$.
\end{lemma}

\begin{proof}
If at most $\kappa(J+r)$ items are absent and a linear functional
$\vlambda$, nonzero on the span of $I$ (so $h_P(\vlambda)>0$), vanishes on
all present columns, then the width of $Z(Y)$ in direction
$\vlambda$ is at most $\kappa(J+r)h_P(\vlambda)<2\mu h_P(\vlambda)$, a
contradiction. So rank is kept while at most $J$ reductions have occurred
(the term $\kappa J$ covers excluded items, the term $\kappa r$ the bases
extracted below).
Every intermediate lattice contains a basis of actual columns, whose index
in $\Lambda_I$ is at most $\Psi(I)\le\Psi_0$ (Lemmas~\ref{lem:volume}
and~\ref{lem:depth-bounds}); each
strict full-rank reduction at least doubles the index, so there are at most
$\log_2\Psi_0\le J$ of them.

\emph{Deletion oracle.} In the current pool $A'$ extract greedily $\kappa+1$
disjoint bases (possible by the width argument). For each basis $B'$
enumerate the characters $\chi$ of $\lat A'/\lat B'$, a group of order at
most $\Psi(I)$ (the cross-polytope on $B'$ lies in $P_I$), by Smith normal
form. For each character, sum the multiplicities of the current types on
which it is nonzero, stopping at $\kappa+1$; this is a scan of $O(T_I)$
types. If some nontrivial $\chi$ is nonzero on between $1$ and $\kappa$
current items, deleting them strictly shrinks the lattice; afterwards all
items of $Y$ whose vectors lie in the new lattice are restored, which is
a whole-type operation. Conversely, if deleting $F'$ ($|F'|\le\kappa$) leaves a proper
sublattice $\Gamma'$, some extracted basis misses $F'$, so
$\lat B'\subseteq\Gamma'$, and a nontrivial character of $\lat A'/\Gamma'$
vanishes on all survivors and is nonzero on some item of $F'$.

\emph{Time.} Each of the at most $J+1$ iterations extracts $\kappa+1$
bases and scans $O((\kappa+1)\Psi(I))$ characters over $O(T_I)$ types, with
fixed-rank arithmetic on numbers of polylogarithmic bit length.
\end{proof}

\subsection{Threshold plateau and the every-price collar}

Fix a deep anchor $\vv$, a price $\vpi\in\partial f_I(\vt_v)$, and the
normalization \eqref{eq:normalize}. Put
\[
 L_{\vpi}(\theta)=\{i\in I:\bar c_i(\vpi)\le\theta\},\qquad
 \Gamma_{\vpi}(\theta)=\lat\cR_\kappa(L_{\vpi}(\theta)).
\]
By Lemma~\ref{lem:depth} with $\xi=1/128$, $\zeta'=0$, $\theta=\theta_0$,
\[
 \bar \vt_v+(K/128-2k_0)P\subseteq\cT_{L_{\vpi}(\theta_0)}(\varrho/2).
\]
The affordable body lies in a translate of $Z(L_{\vpi}(\theta_0))$, so by \ref{S:core} and Lemma~\ref{lem:core-alg} every $\Gamma_{\vpi}(\theta)$
with $\theta\ge\theta_0$ has full rank (here $r\le D$ by
Lemma~\ref{lem:depth-bounds}). These lattices increase with $\theta$
(monotonicity of $\cR_\kappa$), and as functions of $\theta$ they are
constant except at the values $\bar c_i(\vpi)$, where they may jump; each
is right-continuous since $L_{\vpi}(\theta)$ is defined by $\bar c_i\le\theta$.
Each contains a basis of actual columns, so its index in $\Lambda_I$ is at
most $\Psi(I)\le\Psi_0$, and each strict increase at least halves this
index; hence there are at most $\log_2\Psi_0<J$ strict changes. Divide
$[\theta_0,2\theta_0]$ into $J+1$ intervals $[\theta_0+m\Delta,
\theta_0+(m+1)\Delta]$ with $\Delta=\theta_0/(J+1)$. The half-open intervals
$(\theta_0+m\Delta,\theta_0+(m+1)\Delta]$ are disjoint and fewer than $J$ of
them contain a change point, so some interval $[\theta_-,\theta_+]$
satisfies
\begin{equation}\label{eq:plateau}
 \Gamma_{\vpi}(\theta_-)=\Gamma_{\vpi}(\theta_+)=:\Gamma.
\end{equation}
Let $\theta_{\rm m}=(\theta_-+\theta_+)/2$, $R=\{i\in I:\va_i\in\Gamma\}$, and
$X=I\setminus R$. Then $\lat R=\Gamma$, and by Lemma~\ref{lem:core-alg} at
most $\kappa J$ items of $X$ lie in $L_{\vpi}(\theta_+)$. Put
$q_{\rm col}=4\varrho/\Delta=16(J+1)k_0$.

\begin{lemma}[Every-price collar]\label{lem:collar}
Let $(\vt',s')$ be a positive-gap $R$-query obtained from a query assigned to
$\vv$ by an assignment $\vY$ of $X$, and let $\vy$ be a point of its buffer at
depth $j+1$. Then, with
$V_R(\vw)=f_R(\vw)-\vpi \vw-\sum_{i\in R}\min(0,c_i-\vpi \va_i)$,
\begin{equation}\label{eq:collar}
 0\le V_R(\vy+\vz)\le\varrho\qquad(\vz\in q_{\rm col}P).
\end{equation}
Consequently every $\vpi'\in\partial f_R(\vy)$ satisfies
\begin{equation}\label{eq:collar-price}
 |(\vpi'-\vpi)\va_i|\le\Delta/4\quad(i\in I),\qquad
 \Gamma=\lat\cR_\kappa\bigl(L_{\vpi'}(\theta_{\rm m})\bigr).
\end{equation}
This also holds for merged queries, via their attaining source.
\end{lemma}

\begin{proof}
Write $\vy=\vt-A_X\vY+\vz_b$ with $\vz_b$ a buffer displacement. By
Lemma~\ref{lem:drift} with $e'=\kappa J$, $A_X(\vY-\ve_X(\vpi))\in\frac E2P$. By
the embedding \eqref{eq:embedding}, \eqref{eq:anchor}, and
$\rho_{j+1}\le\rho_*$ (by \ref{S:sched}, with $j+1\le J$ by
Lemma~\ref{lem:depth-bounds}), $\vz_b\in\rho_*C_\gamma Q$, and
$\vt-\vt_v\in\gamma Q$. Hence
\[
 \vy-A_R\ve_R(\vpi)\in\bar \vt_v+\zeta Q+\tfrac E2P .
\]
Apply Lemma~\ref{lem:depth} to $L_{\vpi}(\theta_+)\supseteq L_{\vpi}(\theta_0)$
with $\xi=1/64$ and $\zeta'=\zeta$: the radius is
$(1/64-\zeta)K-4k_0$. Deleting the at most $\kappa J$ items of
$L_{\vpi}(\theta_+)$ outside $R$ costs $\kappa J$, and adding items of $R$ only
enlarges the body. Since $\frac E2+\kappa J\le E$, condition \ref{S:collar}
places $\vy-A_R\ve_R(\vpi)+q_{\rm col}P$ in $\cT_R(\varrho)$, which is the upper
bound in \eqref{eq:collar}; the lower bound holds because normalized costs
are nonnegative. Applying the supporting inequality of $V_R$ (with price
$\vpi'-\vpi$) at $\vy\pm q_{\rm col}\va_i$ gives
$q_{\rm col}|(\vpi'-\vpi)\va_i|\le\varrho$, i.e.\ the first part of
\eqref{eq:collar-price}. Then
$L_{\vpi}(\theta_-)\subseteq L_{\vpi'}(\theta_{\rm m})\subseteq L_{\vpi}(\theta_+)$, and
monotonicity of $\cR_\kappa$ with \eqref{eq:plateau} gives the lattice
identity.
\end{proof}

\subsection{Ownership and the arithmetic contraction}

\begin{lemma}[Tie-safe ownership]\label{lem:ownership}
Fix $R\subseteq I$, $X=I\setminus R$. For each $\vy$ in the domain of $f_R$
choose $\vpi_y\in\partial f_R(\vy)$ and set $U_R(\vy)=\vy+A_X\ve_X(\vpi_y)$. Then
$f_I(U_R(\vy))=f_R(\vy)+\vc_X^{\T}\ve_X(\vpi_y)$, every element of
$\partial f_I(U_R(\vy))$ lies in $\partial f_R(\vy)$, and $U_R$ is injective.
\end{lemma}

\begin{proof}
At price $\vpi_y$ an optimum for $f_R(\vy)$ and the endpoints $\ve_X(\vpi_y)$
jointly minimize the reduced cost $\vc^{\T}\vx-\vpi_yA_I\vx$ over the box, which
gives the value identity. If $\vpi^*$ supports $f_I$ at $U_R(\vy)$, extending
competing $R$-solutions by $\ve_X(\vpi_y)$ shows that $\vpi^*$ supports $f_R$ at
$\vy$. If $U_R(\vy_1)=U_R(\vy_2)$, apply $\vpi_{y_1}-\vpi_{y_2}$:
\[
 0=\langle\vpi_{y_1}-\vpi_{y_2},\vy_1-\vy_2\rangle+
 \sum_{i\in X}\bigl[(\vpi_{y_1}-\vpi_{y_2})\va_i\bigr]\bigl(e_i(\vpi_{y_1})-e_i(\vpi_{y_2})\bigr).
\]
The first term is nonnegative by monotonicity of subgradients, each summand
is nonnegative, and a summand with different endpoints is positive (e.g.\
$\vpi_{y_1}\va_i>c_i\ge\vpi_{y_2}\va_i$). So all endpoints agree and $\vy_1=\vy_2$.
\end{proof}

\begin{proposition}[Linear arithmetic charge]\label{prop:arithmetic}
Merge arithmetic children with the same core $R$, over all deep anchors of
a node, and give them depth $j+1$. If every core lattice is a proper
full-rank sublattice of $\Lambda_I$, then $w_R\le w$ for every core and
\begin{equation}\label{eq:arith-contraction}
 \sum_Rw_R\Psi_R\le w\Psi(I).
\end{equation}
\end{proposition}

\begin{proof}
Let $\vy$ be a buffered point of core $R$, with attaining source
$\vq=(\vt_q,s_q)$ assigned to the anchor $\vv$, removed assignment $\vY$,
and $\vy=\vt_q-A_X\vY+\vz_b$. Lemma~\ref{lem:collar} gives a
full-dimensional neighbourhood of $\vy$ in the domain of $f_R$; choose
$\vpi_y\in\partial f_R(\vy)$. By \eqref{eq:collar-price}, items with
$\bar c_i(\vpi)>\theta_+$ keep their endpoints when $\vpi$ is replaced by
$\vpi_y$, so at most $\kappa J$ endpoints change and, with
Lemma~\ref{lem:drift}, $\mathbf d=A_X(\vY-\ve_X(\vpi_y))\in EP$. Embedding
and clipping give $\vz_b\in\rho_{j+1}Q_q\cap H_{j+1}P_I$, and
$Q_q\subseteq C_\gamma Q_v$ by \eqref{eq:anchor}. Deepness gives
$EP\subseteq(E/K)Q_v$. Hence
\begin{align*}
 U_R(\vy)-\vt_v&=(\vt_q-\vt_v)+\vz_b-\mathbf d
 \in\bigl(\gamma+C_\gamma\rho_{j+1}+E/K\bigr)Q_v,\\
 U_R(\vy)-\vt_v&\in(\gamma H_j+H_{j+1}+E)P_I .
\end{align*}
By \ref{S:own}, $U_R(\vy)\in\vt_v+(\rho_jQ_v\cap H_jP_I)$, so
$U_R(\vy)\in\cD_j(\cW)$: the image lies in the buffer of the anchor itself,
not only in that of the source. By Lemma~\ref{lem:ownership}, $U_R$ is
injective, so $w_R\le w$.

Fix $\vu\in\cD_j(\cW)$ and $\vpi^*\in\partial f_I(\vu)$. By
Lemma~\ref{lem:ownership}, $\vpi^*$ supports every core $R$ at its preimage
of $\vu$, so by Lemma~\ref{lem:collar} the lattice of $R$ is
$\lat\cR_\kappa(L_{\vpi^*}(\theta_{{\rm m},R}))$. For fixed $\vpi^*$ these
lattices form a chain; equal lattices give equal cores, which were
merged; and injectivity gives at most one preimage per core. The distinct
proper lattices in the chain have indices at least $2,4,8,\dots$ in
$\Lambda_I$, so by Lemma~\ref{lem:volume}
\[
 \sum_{R\ \text{owning}\ \vu}\Psi_R\le\Psi(I)\sum_{m\ge1}2^{-m}=\Psi(I).
\]
Summing over the $w$ points $\vu$ gives \eqref{eq:arith-contraction}.
\end{proof}

The contraction \eqref{eq:arith-contraction} is not strict. The
potential of Section~\ref{sec:complexity} pays for it with the remaining
arithmetic depth $J-j$, which decreases by one in every arithmetic child.

\subsection{The high-residual trial and exact restoration}

Let $A_c=\cR_\kappa(L_{\vpi}(\theta_-))$, so $\lat A_c=\Gamma$, and apply
Theorem~\ref{thm:gap} to the multiset of working columns of $A_c$, given
by the support and multiplicities of its types in the coordinates of
Lemma~\ref{lem:boxing}, obtaining
$I_{\rm out}$, $I_{\rm res}$, $I_{\rm rem}$, and $\vu_i,\phi,\Bx,\vz_0$. Put
$I_\circ=I_{\rm res}\cup I_{\rm rem}=A_c\setminus I_{\rm out}$.

A \emph{trial} transfers every item of $I\setminus A_c$ and then
$I_{\rm out}$ into a removed-item table. The removed set contains at most
$\kappa J+k_0$ items of normalized cost at most $\theta_0$, so by
Lemma~\ref{lem:drift} every needed prefix lies within $\frac E2P$ of its
integral endpoint reference; this is the table window.

A trial entry (respectively, an entry of a restored $X$-only table) is
\emph{compatible} if its residual target lies in $\Gamma$, the lattice of
the remaining pool $I_\circ$ (respectively $R$). Since remaining items have
columns in $\Gamma$, original and normalized residual targets differ by an
element of $\Gamma$, so the condition does not depend on the normalization;
for trial entries it says that the residual target lies in $\phi(\Z^k)$.
Incompatible entries have no completion and are discarded.
The \emph{filters} applied to trial and restored entries are exactly
the three tests of step~1 of Section~\ref{sec:merge}, applied to the
residual query in the remaining pool: compatibility, fractional
feasibility (the residual target lies in the zonotope of the remaining
pool), and positive gap. No other filter is used. The same three tests,
in the child's own pool and lattice, are applied to every emitted child
query of every kind (face, thin, arithmetic, and lift children,
including each new coordinate value produced by the strip enumeration
of Section~\ref{sec:lift}) before the child is emitted.

\begin{lemma}[High residual or proper restoration]\label{lem:high-low}
$\phi(\Z^k)=\Gamma$; every item of $I_{\rm rem}$ has
$\bar c_i\le\varrho/(2\nu)$; and every reserve subset has normalized cost
at most $g_v/64$. Exactly one of the following is selected.
\begin{enumerate}[label=\arabic*.,start=0]
\item No compatible trial entry survives the filters. Then no query
assigned to $\vv$ has a solution, and these queries are discarded; the
anchor creates no children.
\item Some compatible surviving trial entry has normalized allowance at
least $g_v/4$. One such entry is the \emph{centre} of the lift construction
of Section~\ref{sec:lift}.
\item Some compatible trial entry survives, but none has normalized
allowance at least $g_v/4$. Then $\Gamma\subsetneq\Lambda_I$; restore
$R=I\cap\Gamma$ and compute the removed-item table on $X=I\setminus R$
alone. These $R$-states are the arithmetic children.
\end{enumerate}
All outcomes preserve exact completeness and disjointness of witnesses.
\end{lemma}

\begin{proof}
The size hypothesis of Theorem~\ref{thm:gap} holds: a $\kappa$-resilient
pool of positive rank has more than $\kappa$ items, and $\kappa$ exceeds
the required multiple of $\tau$ (Section~\ref{sec:params}). The reserve
witnesses for $\vz_0$ and $\vz_0+\ve_m$ ($c_\eta h\ge1$) show that
differences of lifted reserve sums give every unit vector $\ve_m$, so
$\phi(\Z^k)$ is generated by differences of sums of actual reserve
columns, which lie in $\Gamma$; thus $\phi(\Z^k)\subseteq\Gamma$.
Conversely $A_c$ is $\kappa$-resilient and
$|I_{\rm out}|+|I_{\rm res}|\le k_0\le\kappa$, so
$\lat I_{\rm rem}=\Gamma$, and these columns are images $\phi\vu_i$; hence
$\phi(\Z^k)=\Gamma$. (Generation of $\Z^k$ by the lifted remainder is not
needed; only its real rank, robust under $h$ deletions, is used in
Section~\ref{sec:lift}.) Items of $A_c$ cost at most
$\theta_-\le2\theta_0=g_v/(64k_0)\le\varrho/(2\nu)$ since $\nu\le k_0$, and
$|I_{\rm res}|\le k_0$ gives the reserve bound.

Suppose outcome~2 is selected and $\Gamma=\Lambda_I$. For the anchor
$\vv$, assign every trial-removed item its preferred endpoint $e_i(\vpi)$.
The residual has normalized target $\bar\vt_v$ and allowance $g_v$, is
compatible, and fits the table window since all prefix deviations are
zero. Lemma~\ref{lem:depth}, applied to
$L_{\vpi}(\theta_-)$ with $\theta=\theta_0$, $\xi=1/128$, and $\zeta'=0$,
then deleting the at most $\kappa J+k_0$ items of that pool outside $I_\circ$, gives
$\bar\vt_v\in\cT_{I_\circ}(\varrho/2)$. Indeed, the remaining radius is
$K/128-2k_0-(\kappa J+k_0)>0$: \ref{S:lift} gives
$K/128-2k_0\ge2E$, and \ref{S:drift} gives $\kappa J+k_0\le E/2$.
The entry therefore survives all filters, and dominance can only increase
its allowance $g_v>g_v/4$, a contradiction. Thus $\Gamma$ is proper in
outcome~2. This uses only the anchor; no lower bound on the gaps of the
other assigned queries is needed.

The restored table must be recomputed on $X$ alone: a trial assignment may
use items of $R\setminus A_c$, which are now in the child pool.

Finally, the trial table is complete, which justifies outcome 0. Let $\vx$
solve a query $\vq$ assigned to $\vv$, and let $\vY$ be its restriction to the
removed set. Normalized costs are nonnegative, so the normalized cost of
$\vY$ is less than $a_q$, and by Lemma~\ref{lem:drift} every prefix of $\vY$ lies
in the table window. By Lemma~\ref{lem:transfer} the table keeps, at the
weight of $\vY$, an assignment of no larger cost, which $\vx|_{I_\circ}$ still
completes. That residual is compatible (its target is a sum of columns of
$I_\circ$, hence lies in $\Gamma$), fractionally feasible, and has positive
gap (a binary completion costs less than the allowance). So it survives
the filters. If no entry survives, no assigned query has a solution.
\end{proof}

Lemma~\ref{lem:high-low} and Proposition~\ref{prop:arithmetic} charge all
arithmetic children without a factor for the number of anchors.

%======================================================================
\section{Deep anchors II: an affordable fibre yields a solution or a thin coordinate}
\label{sec:lift}

\subsection{A point-or-thin primitive}

\begin{lemma}[Point or thin direction]\label{lem:point-thin}
Let $\mathcal K\subseteq\R^m$ be a nonempty bounded rational polytope,
accessed through an optimization oracle that returns, for a rational
linear objective, the optimal value and an optimal point, and a
membership test. Let $\lambda\ge\log_2(m+2)$ be such that every returned
point is a rational vector with a common denominator, whose numerators
and denominator have at most $\lambda$ bits. With $O_m(\lambda)$ oracle
calls and membership tests and $O_m(\lambda)$ further arithmetic
operations on numbers of $O_m(\lambda)$ bits one can find either a point
of $\mathcal K\cap\Z^m$, or a primitive $\vpsi\in\Z^m\setminus\{\vzero\}$
with $\max_{\mathcal K}\vpsi-\min_{\mathcal K}\vpsi\le W_m$, where $W_m$
depends only on $m$. (The second outcome may occur even if
$\mathcal K\cap\Z^m\ne\emptyset$.)
\end{lemma}

\begin{proof}
First find affinely spanning points of $\mathcal K$ by optimizing along
normals of the current affine hull ($O(m)$ calls); if
$\dim\mathcal K<m$, a primitive normal has width $0$. Otherwise start
from a simplex with vertices $\vv_0,\dots,\vv_m\in\mathcal K$, all
returned by the oracle, and barycentric functions $\lambda_i$. While some
returned optimizer $\vx$ of $\pm\lambda_i$ has $|\lambda_i(\vx)|>2$,
replace $\vv_i$ by $\vx$; the volume more than doubles. The volume of a
simplex with returned vertices is $|\det[\vv_i;1]_{i\le m}|/m!$; clearing
the $m+1$ denominators shows that it is either $0$ or at least
$2^{-(m+1)\lambda}/m!$, and it is at most $(m+1)!\,2^{(m+1)\lambda}$.
Hence there are at most $2(m+1)\lambda+O(m\log m)=O_m(\lambda)$
replacements, each found with $2(m+1)$ oracle calls. Afterwards
$|\lambda_i|\le2$ on $\mathcal K$.

With $V=[\vv_1-\vv_0,\dots,\vv_m-\vv_0]$, centroid $\vo$, $\varsigma=1/(m(m+1))$, and
$T=\varsigma^{-1}V^{-1}$, we get $B_2\subseteq T(\mathcal K-\vo)\subseteq
R_mB_2$ with $R_m=3m^2(m+1)$ (the standard simplex contains the ball of
radius $\varsigma$ about its centroid). LLL-reduce $T\Z^m$ to a basis
$\vb_1,\dots,\vb_m$ with $\prod\|\vb_i\|\le D_m|\det|$, $D_m=2^{m(m-1)/4}$
\cite{LLL82}. Round the coordinates of $T\vo$ with respect to $\vb_1,\dots,\vb_m$ to get $\vw\in T\Z^m$ with
$\|\vw-T\vo\|\le\frac m2\max\|\vb_i\|$. If $T^{-1}\vw\in\mathcal K$ (an LP test),
output it. Otherwise $\|\vw-T\vo\|>1$, so the longest basis vector has length
$B>2/m$, and its dual vector satisfies $\|\vb_j^*\|\le D_m/B<mD_m/2$. Then
$\vpsi=T^{\T}\vb_j^*\in\Z^m$ has width at most $2R_m\|\vb_j^*\|\le W_m:=R_mmD_m$
on $\mathcal K$; dividing by the gcd keeps this bound.

\emph{Operation count.} After clearing denominators, $V$ is an integer
matrix with entries of $O_m(\lambda)$ bits, and since LLL reduction
commutes with scaling, it suffices to reduce the integer basis
$\operatorname{adj}(V)$ of $(\det V)\,V^{-1}\Z^m$, whose entries have
$O_m(\lambda)$ bits. LLL performs $O(m^2\log B)$ exchange steps, where
$\log B=O_m(\lambda)$ bounds the logarithm of the squared norms of the
input basis, each with $O(m^2)$ operations on numbers of $O_m(\lambda)$
bits~\cite{LLL82}. Rounding $T\vo$ and the final membership test add
$O_m(1)$ operations and one test.
\end{proof}

\subsection{Set-up and depth}

We are in outcome 1 of Lemma~\ref{lem:high-low}. Use the working columns
$\bar \va_i$ and costs $\bar c_i\ge0$, with $\phi \vu_i=\bar \va_i$, and set
\[
 U=4g_v,\qquad k_{\rm res}=|I_{\rm res}|\le k_0 .
\]
Let $\vy_0$ be the normalized target of the centre (allowance at least
$g_v/4$). For $b\ge0$ define
\[
 \mathcal J_b=\vz_0+\Bigl\{\textstyle\sum_{i\in I_{\rm rem}}\vu_ix_i:
 \vzero\le \vx\le\mathbf 1,\ \sum_{i\in I_{\rm rem}}\bar c_ix_i\le b\Bigr\},\qquad
 \mathcal K_\varrho=\mathcal J_\varrho\cap\phi^{-1}(\vy_0).
\]
Here $r=\rank\phi=\rank I$. Let $S=Q\cap H_jP$ and
$\Pref=S/K\supseteq P$.

\paragraph{Depth.} By Lemma~\ref{lem:drift} the centre satisfies
$\vy_0\in\bar \vt_v+\gamma Q+\frac E2P$. Lemma~\ref{lem:depth} for
$L_{\vpi}(\theta_-)$ with $\xi=1/128$, $\zeta'=\gamma$, eroded by the at most
$\kappa J+k_0\le E/2$ items outside $I_\circ$, gives
\begin{equation}\label{eq:lift-depth}
 \vy_0+\mu\Pref\subseteq\cT_{I_\circ}(\varrho/2),\qquad
 \mu=(1/128-\gamma)K-2k_0-2E .
\end{equation}

\paragraph{Endpoint fixing.} Let a set $H'\subseteq I_\circ$ of $k'\le\mu$
items be assigned prescribed endpoints $e_i\in\{0,1\}$. Since costs are
nonnegative, deleting the coefficients of $H'$ from an affordable
representation gives
$\cT_{I_\circ}(\varrho/2)\subseteq\cT_{I_\circ\setminus H'}(\varrho/2)+
\sum_{i\in H'}[\vzero,\bar\va_i]$, and since $\pm\bar \va_i\in P\subseteq\Pref$,
$\sum_{i\in H'}[\vzero,\bar\va_i]\subseteq\sum_{i\in H'}e_i\bar \va_i+k'\Pref$.
Combining with \eqref{eq:lift-depth} and comparing support functions (the
summand $k'\Pref$ cancels, as all bodies are compact and convex) shows
\begin{equation}\label{eq:endpoint}
 \vy_0+(\mu-k')\Pref\subseteq\textstyle\sum_{i\in H'}e_i\bar \va_i
 +\cT_{I_\circ\setminus H'}(\varrho/2).
\end{equation}
Fixing a reserve witness for $\vz_0$ yields
\begin{equation}\label{eq:fixed-reserve}
 \vy_0+(\mu-k_0)\Pref\subseteq\phi\mathcal J_{\varrho/2}.
\end{equation}
In particular $\mathcal K_\varrho\ne\emptyset$. No dimension claim is
needed: Lemma~\ref{lem:point-thin} also handles lower-dimensional bodies.

\paragraph{Patch radius.} Every completion relevant to this anchor (a
buffered point of any child created below, completed as in
Lemma~\ref{lem:lift-afford}) has, after replacing its reserve part by the
fixed witness of $\vz_0$, an old target $\vy\in\phi\mathcal J_U$ with
\begin{equation}\label{eq:patch-radius}
 \vy-\vy_0\in R_{\rm p}\Pref,\qquad R_{\rm p}=(1+2\gamma)K+E+2k_0 .
\end{equation}
Indeed, the two queries differ from $\vt_v$ by elements of $\gamma S$, the
buffer displacement lies in $S$ (embedding and clipping, since
$\rho_*C_\gamma<1$), the two removed assignments each lie within
$\frac E2P$ of the same endpoint reference (Lemma~\ref{lem:drift}), and the
reserve replacement moves at most $k_{\rm res}\le k_0$ items. Recall
$S=K\Pref$ and $P\subseteq\Pref$.

The same holds for every binary solution $\vx$ of a query $\vq$ assigned
to $\vv$, restricted to $I_\circ$ after the removed assignment retained
by the trial table (Lemma~\ref{lem:high-low}); we include these among
the relevant completions. The retained entry need not come from $\vq$
itself: by dominance it may be an entry from another source $\vq'$
assigned to $\vv$, with an assignment $\vY'$ of the removed set of no
larger cost at the same weight. That entry survived the positive-gap
filter, so the normalized cost of $\vY'$ is less than $a_{q'}$, and
Lemma~\ref{lem:drift} applies to $\vY'$ with the allowance $a_{q'}$;
hence $\vY'$ lies within $\frac E2P$ of the endpoint reference. The
buffer displacement is zero, and the remainder $\vx|_{I_\circ}$, which
completes the retained entry, has normalized cost less than the
entry's residual allowance, hence less than
$a_{q'}\le(1+2\gamma)g_v<U$ by \eqref{eq:allowance}; so its old target
satisfies \eqref{eq:patch-radius} and lies in $\phi\mathcal J_U$. In
particular the strip bound below applies to the restriction of every
binary solution, which is what ``no completion is lost'' requires.

\subsection{Success: rounding an integer point of the fibre}

Compatibility gives $\vy_0\in\phi(\Z^k)$. Smith normal form gives
$\vu^0\in\Z^k$ with $\phi \vu^0=\vy_0$ and a basis $V_0$ of the saturated lattice
$\ker\phi\cap\Z^k$; the integer points of the fibre are $\vu^0+V_0\Z^{k-r}$.
If $k=r$ the unique fibre point is integral. Otherwise apply
Lemma~\ref{lem:point-thin} to $\mathcal K_\varrho$ in the coordinates of
$V_0$. A thin functional on $\ker\phi\cap\Z^k$ extends, since this lattice
is saturated, to an integral $\vpsi$ on $\Z^k$ that is nonzero on
$\ker\phi$ and has width at most $W$ on $\mathcal K_\varrho$.

If an integer point $\vz\in\mathcal K_\varrho$ is found, minimize cost subject
to $\sum_{i\in I_{\rm rem}}\vu_ix_i=\vz-\vz_0$, $\vzero\le \vx\le\mathbf 1$. A basic optimum has
at most $k$ fractional coordinates and cost at most $\varrho$; rounding
them down leaves an integral error in $k\Bx\subseteq c_\eta h\Bx$, which
the reserve represents exactly. The total normalized cost is at most
$\varrho+g_v/64=3g_v/64<g_v/4$, so this is an actual completion of the
centre, and witness pointers reconstruct a solution of the root query.

\subsection{Failure: the cost dual forces a thin slab}

Let $\vpsi$ have width at most $W$ on $\mathcal K_\varrho$. Minimize $\vpsi$
over $\mathcal K_\varrho$; let $\vbeta$ be the multiplier of the weight
equations and $\vartheta\ge0$ that of the cost constraint. Put
\[
 \omega_i=\vpsi(\vu_i)-\vbeta\bar \va_i+\vartheta\bar c_i,\qquad
 \epsilon_i=[\,\omega_i<0\,],\qquad
 \alpha_0=\vpsi(\vz_0)-\vbeta\phi \vz_0+\textstyle\sum_{i\in I_{\rm rem}}\min(0,\omega_i).
\]
By LP duality the minimum is $m_{\vpsi}=\vbeta \vy_0-\vartheta\varrho+\alpha_0$,
and for every representation $\vz=\vz_0+\sum \vu_ix_i$ in the centre fibre
\begin{equation}\label{eq:deficit}
 \vpsi(\vz)-m_{\vpsi}=\sum_{i\in I_{\rm rem}}|\omega_i||x_i-\epsilon_i|
 +\vartheta\Bigl(\varrho-\sum_{i\in I_{\rm rem}}\bar c_ix_i\Bigr).
\end{equation}
The identity holds for every $\vx\in[0,1]^{I_{\rm rem}}$ in the fibre; both
terms on the right are nonnegative when the cost is at most $\varrho$, i.e.\
when $\vz\in\mathcal K_\varrho$, and then the left side is at most $W$.
Evaluating at a half-budget point (which exists by \eqref{eq:fixed-reserve})
gives $\vartheta\le2W/\varrho$. By \eqref{eq:endpoint} with \ref{S:lift}
($\mu\ge2k_0\ge k_{\rm res}+\nu$), any prescribed endpoints on $\nu$ items
of $I_{\rm rem}$ can be completed in the centre fibre at cost at most
$\varrho/2+\nu\cdot\varrho/(2\nu)=\varrho$. Forcing the opposite endpoints
on the $\nu$ items with largest $|\omega_i|$ gives a point of
$\mathcal K_\varrho$ whose first term in \eqref{eq:deficit} contains all
their $|\omega_i|$, so their sum is at most $W$. These $\nu$ \emph{exceptional} items are transferred to the
removed-item table. Every remaining item has $|\omega_i|\le W/\nu$ and
$\vartheta\bar c_i\le W/\nu$, so
\begin{equation}\label{eq:thickness}
 |\vpsi(\vu_i)-\vbeta\bar \va_i|\le2W/\nu .
\end{equation}
Since $I_{\rm rem}$ minus $\nu\le h$ items still spans $\R^k$ and
$\vpsi|_{\ker\phi}\ne\vzero$, adjoining $\vpsi$ raises the rank by exactly one.

\subsection{One strip retains the whole table}

For $\vx\in[0,1]^{I_{\rm rem}}$ put
\[
 D_\omega(\vx)=\sum_{i\in I_{\rm rem}}|\omega_i||x_i-\epsilon_i|.
\]
This is affine and nonnegative on the box, since
$|x_i-\epsilon_i|=\epsilon_i+(1-2\epsilon_i)x_i$. Let $\vx$ be the
remainder of a relevant completion after the reserve replacement in
\eqref{eq:patch-radius}, with target $\vy$ and cost at most $U=128\varrho$.
Put $\vy'=\vy_0-(\vy-\vy_0)/254$. By \ref{S:lift} and
\eqref{eq:fixed-reserve}, there is a remainder $\vx'$ with target $\vy'$
and cost at most $\varrho/2$. Both use the fixed reserve shift $\vz_0$,
so $\widehat{\vx}=(\vx+254\vx')/255$ has target $\vy_0$ and cost at most
$(U+254\varrho/2)/255=\varrho$. Thus \eqref{eq:deficit} gives
$D_\omega(\widehat{\vx})\le W$, and affineness and nonnegativity give
\begin{equation}\label{eq:global-deficit}
 \sum_{i\in I_{\rm rem}}|\omega_i||x_i-\epsilon_i|
 =D_\omega(\vx)\le255D_\omega(\widehat{\vx})\le255W
\end{equation}
for every relevant fractional completion $\vx$, and hence for its
restriction to any subpool, since every summand is nonnegative.

Let $I'$ be the final pool ($I_{\rm rem}$ minus the exceptional items). For
a surviving assignment put $\vw=\sum_{I'}\bar \va_i(x_i-\epsilon_i)$,
$z=\sum_{I'}\vpsi(\vu_i)(x_i-\epsilon_i)$, and
$C_\epsilon=\sum_{I'}\bar c_i\epsilon_i$. Then
$z-\vbeta \vw=\sum_{I'}|\omega_i||x_i-\epsilon_i|-
\vartheta(\sum_{I'}\bar c_ix_i-C_\epsilon)$, and with $\vartheta U\le256W$,
\begin{equation}\label{eq:strip}
 \vartheta C_\epsilon-256W\le z-\vbeta \vw\le\vartheta C_\epsilon+255W .
\end{equation}
Above each integer old target there are therefore at most $511W+1$ integer
values of the new coordinate. Reserve and exceptional items are transferred
by an exact table in old weights, with window $2k_0P$ around the
endpoint reference: there are at most $k_{\rm res}+\nu\le2k_0$ such
items, so every prefix deviation lies in $2k_0P\subseteq S$, and the
restriction loses nothing. For each retained residual we enumerate
all integer new coordinates allowed by \eqref{eq:strip}. No completion is
lost. Before storing, orientations are undone: the new coordinate of the
permanent column $\va_i$ is $\sigma_i\vpsi(\vu_i)$, and complementation about
$\epsilon_i$ is undone by a fixed integral translation of the target
only. The allowance needs no translation: tables store original costs
(Section~\ref{sec:queries}), and complementation changes only the
working coordinates, not the labelled assignment or its original cost.

\begin{lemma}[Types survive lifting]\label{lem:types-lift}
All surviving items of one type receive the same new column, including
items whose working orientations differ.
\end{lemma}

\begin{proof}
By Theorem~\ref{thm:gap}\textup{(iii)}, $\phi$ is injective on the
symmetric set $c_\eta h\Bx\cap\Z^k$, which contains the lifts of all
working columns; hence the lift $\vu(\bar\va)$ depends only on the
working column, and $\vu(-\bar\va)=-\vu(\bar\va)$ whenever both occur. For
an item with sign $\sigma\in\{-1,1\}$ and column $\va$, the stored new
coordinate is $\sigma\vpsi(\vu(\sigma\va))=\vpsi(\vu(\va))$. Costs are
unchanged, so the cost order of the type is preserved.
\end{proof}

Consequently the children are again described by types with
cost-ordered item lists, and $T_{\rm child}$ obeys
Lemma~\ref{lem:types}.

\subsection{Mass and volume of lifted children}

\begin{lemma}[Affordability of buffered points]\label{lem:lift-afford}
Let a child query $(\vt',s')$ arise from a query $(\vt,s)$ assigned to $\vv$ via
a removed assignment $\vY$, and let $\vp\in \vt'+\rho Q_{\rm child}(\vt',s')$ with
$\rho\le\rho_*$. Then $\vp$ is the weight (in all coordinates of the child) of
a fractional completion whose total normalized cost, including $\vY$, is at
most $(1+\rho)a_q<2g_v<U$.
\end{lemma}

\begin{proof}
The paired witnesses for $\vp-\vt'$ have nonnegative normalized costs summing
to at most $2s'$; mixing the relevant one with weight $\rho$ with a centre
representation of cost at most $s'$ gives cost at most $(1+\rho)s'$. With
$s'=a_q-\bar c(\vY)$ and $\bar c(\vY)\ge0$ the total is at most
$(1+\rho)a_q\le(1+\rho_*)(1+2\gamma)g_v<2g_v$ by \eqref{eq:allowance}. Only
the chosen centre needs allowance at least $g_v/4$; the other assigned
queries may have arbitrarily small gaps.
\end{proof}

\begin{proposition}[Lift]\label{prop:lift}
For an anchor in outcome~1, the procedure either returns an actual
completion, or creates children, inheriting depth $j$ and lift count
$\ell+1$, that retain every completion. Summed over anchors,
\begin{equation}\label{eq:lift-mass}
 \sum w_{\rm child}\le\Alift w,\qquad\Alift\le C_{\rm dim}\alpha^{-C_{\rm dim}},
\end{equation}
and each child has rank $r+1$ and
\begin{equation}\label{eq:lift-volume}
 \Psi_{\rm child}\le\frac{2(r+1)W}{\nu}\Psi(I).
\end{equation}
\end{proposition}

\begin{proof}
Completeness was shown above. For mass, all removed assignments of one
anchor are merged into one table with one pool $I'$. Translate by the
integral endpoint reference. By Lemma~\ref{lem:lift-afford} and
\eqref{eq:global-deficit}, every buffered child point has new coordinate
in the strip \eqref{eq:strip}. Its old projection lies in $\vt_v+4S$: the
query offset is in $\gamma S$, the buffer displacement in $S$, the trial
removal in $\frac E2P\subseteq S$ (as $E\le K$ and $KP\subseteq S$ by deepness), and reserve plus
exceptional items in $2k_0P\subseteq S$. Lemma~\ref{lem:folding} with
$s=4$, the factor $511W+1$, and \eqref{eq:packing-count} give
\eqref{eq:lift-mass}; the same bound covers intermediate table states.

For volume, the analysis shear $(\vw,z)\mapsto(\vw,z-\vbeta \vw)$ and
\eqref{eq:thickness} place the new symmetric hull in
$P\times[-2W/\nu,2W/\nu]$. The new lattice projects into $\Lambda_I$ and has
a vertical kernel of integral period at least one, so its covolume is at
least $\covol\Lambda_I$. Hence
$\Psi_{\rm child}/\Psi(I)\le\frac{(r+1)!}{2^{r+1}}\frac{2^r}{r!}
\frac{4W}{\nu}=\frac{2(r+1)W}{\nu}$.
\end{proof}

%======================================================================
\section{Implementation: types, transfers, and oracles}\label{sec:implementation}

The reductions above are stated for labelled pools. This section
implements them in time proportional to the number $T_I$ of types rather
than the number of items, and with removed-item tables whose size is
paid by buffered mass. Each type keeps its items in a persistent ordered
structure sorted by cost, supporting rank, select, threshold, prefix-sum,
interval restriction, and labelled deletion in logarithmic time.
Exceptions are deleted individually; whole-type deletion or restoration
uses shared references. For a fixed price $\vpi$, the orientation
\eqref{eq:normalize} splits a type into at most two cost intervals, on
either side of the threshold $c_i=\vpi\va$.

\subsection{Grouped and multisource exact transfers}

The tables below are \emph{clipped}: their entries minimize over paths
whose intermediate states stay in prescribed finite windows.
Completeness is supplied by the window bounds proved for each reduction.

\begin{lemma}[Convex group update]\label{lem:group-update}
Let $V_{\rm old},V_{\rm new}$ be enumerated finite subsets of the current
lattice of total size at most $M_{\rm w}$. Let a group of independently
selectable items have common nonzero column $\va$ and sorted costs
$c_1\le\dots\le c_{\bar k}$, and put $C(k)=\sum_{i\le k}c_i$. The update
\begin{equation}\label{eq:group-update}
 F_{\rm new}(\vz)=\min_{\substack{0\le k\le\bar k\\ \vz-k\va\in V_{\rm old}}}
 \bigl\{F_{\rm old}(\vz-k\va)+C(k)\bigr\},\qquad \vz\in V_{\rm new},
\end{equation}
with attaining counts, can be computed in time
$\widetilde O_{d,\delta}(M_{\rm w})$, for an arbitrary old table (an empty
minimum is $+\infty$).
\end{lemma}

\begin{proof}
The cheapest $k$ items of the group cost $C(k)$, and $C$ is convex.
Partition both state sets into cosets of $\Z\va$ by fixed-rank integer
coordinates and sorting; distinct cosets on one geometric line are kept
apart when $\va$ is not primitive. In one coset use the integer step
coordinate, delete old columns of value $+\infty$, and consider the
matrix $\mathsf M(i,j)=F_{\rm old}(j)+C(i-j)$ for $0\le i-j\le\bar k$. For each row the
valid columns form an interval of the sorted old list, with endpoints
monotone in the row. Convexity of $C$ gives the Monge inequality whenever
all four entries are valid; if two finite row minimizers crossed, the
band inequalities would make the other two entries valid, and the Monge
inequality would uncross them. So the leftmost finite row minimizers are
nondecreasing, and divide-and-conquer monotone search computes all row
minima in a logarithmic factor times the number of rows and finite
columns; rows with no valid finite column receive $+\infty$. Each prefix
cost is a logarithmic-time query in the persistent profile.
\end{proof}

\begin{lemma}[Clipped multisource transfer]\label{lem:multisource}
Fix a removed pool $X$, a remaining pool $R$, and a batch of queries
$\vq=(\vt_q,s_q)$, all sharing $X$ and $R$. Fix a binary reference
$\ve$ on $X$ and put $\sigma_i=1-2e_i$. Partition the items of $X$ into
groups $G_1,\dots,G_T$, processed in a fixed order, each with a common
vector $\vb_g=-\sigma_i\va_i$, cost list $(\sigma_ic_i)_{i\in G_g}$, and
sorted prefix-cost function $C_g$ on $0\le k\le\bar k_g=|G_g|$. Let
$V_0,\dots,V_T$ be enumerated state sets. One grouped dynamic program
computes
\begin{equation}\label{eq:multisource}
 F_m(\vz)=\min\Bigl\{\vc_X^{\T}\ve-s_q+\sum_{g\le m}C_g(k_g)\Bigr\},
\end{equation}
the minimum over sources $\vq$ and counts $0\le k_g\le\bar k_g$ such that
$\vz_0=\vt_q-A_X\ve\in V_0$, $\vz_h=\vz_0+\sum_{g\le h}k_g\vb_g\in V_h$
for $h\le m$, and $\vz_m=\vz$; it stores an attaining source and an
actual witness, implicitly, as back-pointers to the source and to the
per-group counts $k_g$ (an explicit assignment is written out only when
a witness is reconstructed). With a common window of $M_{\rm w}$ points at every stage, the work is
$\widetilde O_{d,\delta}(TM_{\rm w}+|\cW|)$. Every reported entry is an actual
assignment $\vY$ of $X$ with residual $(\vt_q-A_X\vY,\,s_q-\vc_X^{\T}\vY)$,
and every assignment whose prefix states stay in the windows is
dominated by a retained entry.
\end{lemma}

\begin{proof}
Initialize $F_0(\vt_q-A_X\ve)=\vc_X^{\T}\ve-s_q$, keeping the smallest
value and an attaining source at duplicate targets, and process the
groups by Lemma~\ref{lem:group-update}; induction gives
\eqref{eq:multisource}. A count $k_g$ denotes the $k_g$ cheapest items of
$G_g$, so $\vY$ with $Y_i=e_i+\sigma_iy_i$ is an actual binary assignment,
and at the end $F_T(\vz)=\vc_X^{\T}\vY-s_q$ with $\vz=\vt_q-A_X\vY$. At a
fixed state, keeping the smallest value keeps the largest residual
allowance; by Lemma~\ref{lem:transfer} every completion of a discarded
entry completes the kept one, also at intermediate stages, because the
later groups, windows, and remaining pool are common to all sources.
Replacing chosen items within a group by its cheapest items preserves
all states and does not increase cost, so every assignment whose prefixes
stay in the windows is represented.
\end{proof}

Equal vectors in different groups are allowed. When $X$ consists of
whole cost intervals of types, a type contributes at most two oriented
groups, so $T\le2T_X\le2T_I$. Trial and lift tables also transfer
label-selected subsets (the outliers $I_{\rm out}$, the reserve, and the
exceptional items), at most $2k_0$ items in total; each such item forms
a singleton group, so in general $T\le2T_I+2k_0=O(T_I+k_0)$, and the
extra $O(k_0)$ groups cost $\widetilde O_{d,\delta}(k_0M_{\rm w})$, which
is within the $\nu$-factor of Lemma~\ref{lem:local-work} since $k_0=\widetilde O(\nu)$. A common window of $M_{\rm w}$
points has transfer work $\widetilde O_{d,\delta}(M_{\rm w}\Psi(I))$ by
Lemma~\ref{lem:types}. This removes both a factor for the number of
copies and a product of the number of queries with the table size.
Lattice compatibility, fractional feasibility, positive gap, and the
high-allowance test of Lemma~\ref{lem:high-low} are monotone in the
residual allowance, so they are applied after merging.

\paragraph{References and windows.}
The reference and the processing order are those for which the window
bounds were proved. In a thin reduction, $\ve$ consists of the anchor's
LP endpoints on the heavy items, after a branch $\vz_{\rm fr}$ on its
fractional coordinates; the initial state is
$\vt_q-A_{I_{\rm fr}}\vz_{\rm fr}-A_{\rm heavy}\vx_{0,\rm heavy}$, and by
Lemma~\ref{lem:thin} every needed prefix lies in $3rc_rS_v$. Equal-cost
LP ties within a type can be arranged in the persistent order, since such
a permutation of items with identical column and cost preserves
vertexhood. In a deep trial, a restoration, or a lift table, $\ve$ is
the canonical $\ve_X(\vpi)$ of \eqref{eq:normalize}, and by
Lemma~\ref{lem:drift} every needed prefix lies in
$\frac E2P\subseteq\frac{E}{2K}S_v=\frac\alpha2S_v$; a lift table also
transfers the reserve and the $\nu$ exceptional items, so its prefixes lie
in $(\frac E2+2k_0)P\subseteq\alpha S_v$, using $2k_0\le E/2$ and
$KP_I\subseteq S_v$. In either case put
$\vb_v=\vt_v-A\ve$ (with the fixed branch in the thin case). All initial
states lie in $\vb_v+\gamma S_v$, and every state of every needed path
lies in the single convex window
\[
 V=\Lambda_I\cap\bigl(\vb_v+(\gamma+c)S_v\bigr),\qquad
 c=3rc_r\ \text{(thin)},\quad c=\alpha\ \text{(deep)},
\]
which is enumerated once and used at every stage. By
Lemma~\ref{lem:folding}, $|V|$ is at most a factor polynomial in
$\alpha^{-1}$ (for fixed rank) times the lattice count of the inner body
$\vt_v+\frac\alpha2S_v$; by \eqref{eq:packing-count} these windows have
total size $O_{\rm dim}(\alpha^{-C_{\rm dim}}w)$ over all anchors.

\subsection{Typed continuous and enumeration oracles}

\begin{lemma}[Typed fixed-row LP oracle]\label{lem:typed-lp}
Consider the LPs used by the reductions: a fixed number of copies of the
pool with box variables (one per item, or one per type with an integer
upper bound), a fixed number of free variables, a fixed number of global
coupling rows whose coefficients depend only on the type and the copy,
and costs (original or normalized, $c_i-\vpi\va_i$) entering linearly
through the objective or through upper-cost constraints. Their exact
value, a supporting dual solution, and an optimal primal witness can be
computed, or infeasibility or unboundedness detected, with
$T_I\bigl(\log(n_0+2)\bigr)^{C_{\rm dim}}$ arithmetic operations on
numbers of $O_{\rm dim}(\lambda_0)$ bits, provided the input numbers of
the LP (columns, costs, prefix sums, and the coefficients of the
objective and of the right-hand side) are rationals with a common
denominator $D$, and $\lambda_0\ge\log(n_0+2)$ bounds the bit length of
$D$ and of every input number. (In the reductions, columns and original
costs are integers and the normalized costs $c_i-\vpi\va_i$ of one LP
share the denominator of the price $\vpi$, so $D$ is that denominator, or
a product of a fixed number of them.) When
$\lambda_0=O_{\rm dim}(\log(n_0N(c_{\max}+2)))$ this is time
$\widetilde O_{d,\delta}(T_I)$. If the input LP has no free variables,
the returned primal witness is basic. The number of operations does not
depend on $c_{\max}$.
\end{lemma}

\begin{proof}[Proof sketch]
After clearing the common denominator, bound the fixed number of free
variables by a sufficiently large Cramer bound (which detects
unboundedness and turns them into singleton box variables), and add
artificial slacks with a still larger fixed-row Cramer bound as
penalty. Retain the fixed objective offset $C_{\rm off}$ introduced by
the free-variable shifts. The dual is a concave sum
$\Phi(\vy)=C_{\rm off}+\vy\vb+\sum_i u_i\min\{0,\nu_G(\vy)+c_i\mu_G(\vy)\}$
over a bounded fixed-dimensional domain. Within each group its negative
coefficients form a cost prefix or suffix, so rank searches and weighted
prefix sums evaluate the required coefficients and slopes in
$O(T_I\log(n_0+2))$ operations.

Multidimensional search locates the dual maximizer by hyperplane side
tests: maximize on the hyperplane, then solve the two lower-dimensional
directional-derivative problems on the feasible-direction cone, clipped
at $\bar M_1=2^{C\lambda_0}$. Empty local domains are discarded; Cramer and
recession bounds justify the clipping. The signs of the two local maxima
determine the side by concavity. Weighted median pruning of the group
cost intervals gives $T_I(\log(n_0+2))^{C_{\rm dim}}$ operations.
Independent group perturbations and a generic right-hand side make the
maximizer unique. Only sparse coefficient comparisons and directional
derivatives are used, so perturbed objective values need never be
compared.

At most a fixed number of groups are tied. Complementary slackness,
capacity-weighted prefix profiles, and a fixed-size system recover a
basic primal witness for the boxed program. Taking standard parts and
undoing the shifts returns an optimal witness for the input LP, basic
when no free variables were boxed. An input optimum strictly inside the
artificial free-variable bounds shows that the reported dual also
supports the input LP; artificial slacks detect infeasibility.
Aggregating integer prefix sums before substitution and
using fixed-size Cramer bounds keeps every number at
$O_{\rm dim}(\lambda_0)$ bits. Thus the operation count introduces no
dependence on the magnitude of the costs. The full side test,
perturbation, and arithmetic details appear in Appendix~\ref{app:typed-lp}.
\end{proof}

\paragraph{Where basic witnesses are used.}
The anchor LP defining $f_I$ and the cost-minimization LP at a fixed
integer fibre point in Section~\ref{sec:lift} have only bounded item
variables. They therefore return the basic witnesses needed for the
thin reduction and reserve rounding. The other geometric oracle calls
require optimal points and duals with bounded rational encodings, which
the lemma provides also when free variables are present.
For a fixed bounded geometric body, the point bound is uniform over
support objectives: the denominators of a boxed basic witness are
controlled by minors of the fixed constraint matrix after clearing its
own common denominator. The added integral bounds change right-hand
sides, not those minors. The projected point lies in the fixed bounded
body, so its coordinate bounds also bound its numerators independently
of the queried objective.

\begin{lemma}[Output-sensitive enumeration]\label{lem:enum}
The $M$ lattice points of a bounded rational convex window of fixed
rank described by the oracle of Lemma~\ref{lem:typed-lp}, possibly
intersected with a fixed number of rational halfspaces or lifted by a
fixed number of coordinates, can be enumerated with witnesses in time
$\widetilde O_{d,\delta}\bigl((M+1)T_I\bigr)\cdot L_c$.
\end{lemma}

\begin{proof}
Map the current lattice to $\Z^r$ by a basis. Integer feasibility of a
window in a box is decided by the point-or-thin primitive
(Lemma~\ref{lem:point-thin}): either an integer point is found, or a
functional of width $O_r(1)$ is found, and its $O_r(1)$ integer levels
are treated recursively in lower rank. Bound the integer coordinate
ranges by support queries, bisect boxes, and discard boxes with no
lattice point; a box is stored by its $2r$ bounds, so the descriptions
keep a fixed number of rows. Every surviving box has an output point
below it. The windows lie in the region of reboxed coordinates of
absolute value at most $C_{\rm dim}n_0\Psi_0$ (see the paragraph on bit
lengths below). All bisection is done on boxes in these reboxed
coordinates of the window itself; the lower-rank recursion on the
$O_r(1)$ levels of a thin functional is used only to decide whether a
box contains a lattice point, never to bisect inside a slice (whose
basis could have $O(\lambda)$-bit entries). Hence the bisection depth is
$O_{\rm dim}(\log(N+2))$ and
$\widetilde O_{d,\delta}(M+1)$ boxes are visited; this count does not
involve the costs. The optimizers returned by
Lemma~\ref{lem:typed-lp} on a box-restricted window have
$\lambda=O_{\rm dim}(\log(n_0N(c_{\max}+2)))=O_{\rm dim}(\log(N+2)\,L_c)$
bits by fixed-row Cramer bounds, so each point-or-thin call uses
$O_{\rm dim}(\log(N+2)\,L_c)$ typed LP calls, each in time
$\widetilde O_{d,\delta}(T_I)$, and $O_{\rm dim}(\log(N+2)\,L_c)$ further
operations on numbers of $O_{\rm dim}(1)$ words.
\end{proof}

\begin{lemma}[Product-volume boxing]\label{lem:boxing}
Let $\va_1,\dots,\va_n$ generate a rank-$r$ lattice $\Lambda$ and
$P=\conv\{\pm \va_i\}$. One can choose coordinates identifying $\Lambda$
with $\Z^r$ in which $|(\va_i)_j|\le q_j$ with integers $q_j\ge1$ and
$\prod_jq_j\le C_r\vol(P)/\covol\Lambda\le C_r\Psi$, in time
$\widetilde O_{d,\delta}(T_I)\bigl(1+\lambda/\log(N+2)\bigr)$, where
$\lambda$ bounds the bit length of the given columns; this is
$\widetilde O_{d,\delta}(T_I)$ at the root and for every child other
than a lift child, and $\widetilde O_{d,\delta}(T_I)L_c$ for a lift
child. In particular, for the box
$\mathsf Q=\prod_j[-q_j,q_j]$, $|\tau\mathsf Q\cap\Z^r|\le C_r\tau^r\Psi$
for every $\tau\ge1$.
\end{lemma}

\begin{proof}
The simplex construction in the proof of Lemma~\ref{lem:point-thin}
gives a rational linear map $T$ and a point $\vo$ with
$\vo+B_2\subseteq TP\subseteq \vo+R_rB_2$. By symmetry also
$-\vo+B_2\subseteq TP$, so $B_2\subseteq TP$ by convexity, and
$TP\subseteq(\vo+R_rB_2)\cap(-\vo+R_rB_2)\subseteq R_rB_2$. Replace $P$ by $TP$
and $\Lambda$ by $T\Lambda$ (this preserves $\vol/\covol$), and
LLL-reduce the dual of the transformed lattice to a basis $\vb_j^*$ with
$\prod\|\vb_j^*\|\le C_r/\covol$. The widths $q_j=h_{P}(\vb_j^*)$ are
positive integers, their product is at most
$C_r/\covol\le C_r\vol(P)/\covol$, and the dual basis gives the
coordinates. Each support query over $P$ scans the types and returns a
maximizing column $\pm\va_i$, so its output has the bit length of a
current column. By
Lemma~\ref{lem:point-thin} the work is $\widetilde O_{d,\delta}(T_I)$
times $1+\lambda/\log(N+2)$, where $\lambda$ bounds the bit length of the
current columns. At the root and for every child that is not a lift
child, the columns are reboxed integers with $O_{\rm dim}(\log(N+2))$
bits, and the work is $\widetilde O_{d,\delta}(T_I)$. The new coordinate
$\sigma_i\vpsi(\vu_i)$ of a lift child comes from a functional
$\vpsi$ produced by point-or-thin on a cost-dependent body, and may have
$O_{\rm dim}(\log(N+2)\,L_c)$ bits; reboxing a lift child therefore takes
$\widetilde O_{d,\delta}(T_I)L_c$, after which its columns again have
$O_{\rm dim}(\log(N+2))$ bits.
\end{proof}

A per-coordinate bound $q_j\le C\Psi$ would give $\Psi^r$ and a quadratic
dependence on $d$ in the exponent; the product bound is essential.

\paragraph{Bit lengths.} After every lift the pool is reboxed by
Lemma~\ref{lem:boxing}, so stored columns have entries at most
$C_{\rm dim}\Psi_0$ and feasible targets at most $C_{\rm dim}n_0\Psi_0$.
Costs remain the original integers, and allowances have
$O(\log(n_0(c_{\max}+1)))$ bits. LP prices ($\vpi$, $\vbeta$,
$\vartheta$), thin functionals, and rounding matrices are temporary:
they are used inside one node and never stored, and only the integer
values $\sigma_i\vpsi(\vu_i)$, which are immediately reboxed, are passed
to children. For witness reconstruction the successive coordinate maps
are stored separately rather than composed. Hence bit lengths do not
compound along paths, and all numbers have
$O_{\rm dim}(\log(n_0N(c_{\max}+2)))$ bits, i.e.\ $O_{\rm dim}(1)$ words.
The bit lengths influence the running time only through the iteration
counts of the point-or-thin primitive (Lemma~\ref{lem:point-thin}), of
the rounded bodies $D_v$, and of the reboxing of lift children
(Lemma~\ref{lem:boxing}), which are $O_{\rm dim}(\log(N+2)\,L_c)$.

\paragraph{Bit lengths of the structural data.} The output of
Theorem~\ref{thm:gap} at a deep anchor consists of integers of
$O_{\rm dim}(\log(N+2))$ bits. Indeed, by (iv) and (iii), $\phi$ maps
$\vz_0+(c_\eta h\Bx\cap\Z^k)$ injectively into the set of subset sums of
the at most $k_0$ reserve columns, which lies in a box of side
$O(k_0C_{\rm dim}\Psi_0)$ of the reboxed lattice. Hence
$\prod_j(2c_\eta hb_j+1)\le(C_{\rm dim}k_0\Psi_0)^r$, so every $b_j$,
every coordinate of $\vu_i\in\Bx$ and of $\vz_0$ (a reserve sum), and
every column $\phi\ve_j$ (a difference of two reserve sums, as
$\ve_j\in c_\eta h\Bx$) is bounded by a polynomial in $N$. A basis
$V_0$ of $\ker\phi\cap\Z^k$ then has $O_{\rm dim}(\log(N+2))$ bits by
Smith normal form, and the fibre $\mathcal K_\varrho$ of
Section~\ref{sec:lift} satisfies the hypothesis of
Lemma~\ref{lem:point-thin} with $\lambda=O_{\rm dim}(\log(N+2)\,L_c)$.
The objectives passed to the LP oracle by point-or-thin and by
enumeration have $O_{\rm dim}(\lambda)$ bits; since point-or-thin
recursion in Lemma~\ref{lem:enum} has depth at most $r$, all numbers
remain $O_{\rm dim}(1)$ words.

\subsection{Local work}

\begin{lemma}[Local work]\label{lem:local-work}
A node of mass $w$ and normalized volume $\Psi$ is processed, including
the construction and filtering of all its children, in time
$\widetilde O_{d,\delta}(\nu^2w\Psi)\cdot L_c$.
\end{lemma}

\begin{proof}
Discarding queries, boundary reduction, values, gaps, supporting prices,
and LP vertices take $\widetilde O(T_I)$ per query, and $|\cW|\le w$.
Sorting by gap and center capture enumerate the capture windows, which
have $O_r(w)$ points in total (Section~\ref{sec:buffers}); with
Lemma~\ref{lem:enum} this costs $\widetilde O(wT_I)L_c$. The rounded bodies
cost $\widetilde O(T_I)L_c$ per anchor (Section~\ref{sec:boundary-thin}).
Every other step below uses only typed LPs, table updates, and
cost-free lattice computations, so it carries no factor $L_c$. All thin, trial, restoration, and
lift tables use the common windows of Section~\ref{sec:implementation},
of total size $O_{\rm dim}(\alpha^{-C_{\rm dim}}w)$, and
Lemmas~\ref{lem:multisource} and~\ref{lem:enum} process them, together
with the filters and the strip enumeration over at most $511W+1$ new
values, in $\widetilde O(wT_I)L_c$, the factor $L_c$ coming only from
Lemma~\ref{lem:enum}; recall $\alpha^{-1}=O_{d,\delta}(\log N)$. A fixed
number of trial and restored tables is used per anchor, with no product
with its number of assigned queries. By Lemma~\ref{lem:types},
$T_I\le(2^r+r)\Psi$, so all of this is $\widetilde O(w\Psi)L_c$.

It remains to pay for the canonical-core searches and the structural
call at a deep anchor $\vv$. Its packed inner body contains
$\vt_v+\frac{\alpha K}2P_I$, since $S_v\supseteq KP_I$; let $w_v$ be its
lattice count. By Lemma~\ref{lem:types}(3) and $\alpha K=E$,
\begin{equation}\label{eq:deep-mass}
 w_v\ge\frac{(E/2)^r}{r!}\Psi .
\end{equation}
The plateau search computes $J+2$ canonical cores, each in time
$\widetilde O((\kappa+1)(J+1)T_I\Psi)$ by Lemma~\ref{lem:core-alg}, and
$\kappa J\le E/2$, so this is
$\widetilde O(J\cdot\frac E2T_I\Psi)\le\widetilde O(J\,w_vT_I)$ by
\eqref{eq:deep-mass}. The extra factor $J$ (which matters only for
$r=1$) is polylogarithmic in $N$ and absorbed in $\widetilde O$, so this
is $\widetilde O(w_vT_I)=\widetilde O(w_v\Psi)$.
The structural call of Theorem~\ref{thm:gap} costs
$\widetilde O(hW_{\rm loc})$ with $h\le\tau$ and, by
Lemma~\ref{lem:boxing}, $W_{\rm loc}\le C_r\tau^r\Psi$; since
$E\ge k_0\ge\tau$ (by the choice of constants), \eqref{eq:deep-mass}
gives $\tau^r\Psi\le C_rw_v$, so this costs $\widetilde O(\nu w_v)$.
Only $O(k_0)$ actual exception and reserve items are selected per call,
the point-or-thin computation on the fibre takes
$\widetilde O(T_I)L_c$ (Lemma~\ref{lem:point-thin} with
$\lambda=O_{\rm dim}(\log(N+2)L_c)$), and the selection of the
$\nu$ exceptional items (by largest $|\omega_i|$, an affine function of
the cost on each oriented type, so $|\omega_i|$ is monotone on at most two
pieces of each type's cost order) takes $\widetilde O(T_I+\nu)$ work by
selection by rank across these $O(T_I)$ sorted lists, with
$O(\log(n_0+2))$ weighted-median rounds; no search over values of
$\omega$ is used. Lift children are reboxed in
$\widetilde O(T_I)L_c$ each (Lemma~\ref{lem:boxing}), which is paid in
the same way as the other local work. The inner bodies are disjoint, so $\sum_vw_v\le w$ by
\eqref{eq:packing-count}. Summing all charges proves the lemma, with
room to spare in the power of $\nu$.
\end{proof}

%======================================================================
\section{The recursion, the potential, and the running time}\label{sec:complexity}

\subsection{One node, its children, and merging}\label{sec:merge}

A node with state $(I,\cW,j,\ell)$ is processed as follows.
\begin{enumerate}
\item Discard queries with $\vt\notin\Lambda_I$, $\vt\notin Z(I)$, or
nonpositive gap (for a child, the parent has already applied these
filters before emitting it, Section~\ref{sec:buffers}); settle rank-zero pools directly; reduce boundary
targets (Lemma~\ref{lem:boundary}).
\item Compute gaps, choose anchors by center capture
(Lemma~\ref{lem:packing}), and compute an LP vertex and a supporting
price at each anchor.
\item Construct the rounded body $D_v$. At a thin anchor, branch on the
fractional coordinates, transfer the heavy items from the LP endpoints
by the clipped multisource table, and create the light children
(Lemma~\ref{lem:thin}).
\item At a deep anchor, find the threshold plateau and the canonical
core, apply Theorem~\ref{thm:gap}, and keep its actual reserve and
reconstruction data.
\item Run the trial table from the canonical reference, merge by residual
target and allowance within the fixed remaining pool, and filter
(Lemma~\ref{lem:high-low}).
\item If a high residual survives, use the affordable fibre: return a
completed actual witness if the point-or-thin call finds an integer
point; otherwise transfer the exceptional items and create the lifted
children (Proposition~\ref{prop:lift}).
\item If every surviving residual is low, restore $R=I\cap\Gamma$,
recompute the transfer of $I\setminus R$, and create arithmetic children
(Proposition~\ref{prop:arithmetic}).
\end{enumerate}
Children are of four kinds:
\begin{center}\small
\begin{tabular}{@{}llll@{}}
\toprule
kind & pool & depth & lift count\\
\midrule
face / rank-losing & subpool of $I$, smaller rank & $j$ & $\ell$\\
thin, same rank & subpool of $I$ & $j$ & $\ell$\\
arithmetic & restored core $R$ & $j+1$ & $\ell$\\
lift & $I'$ with one new coordinate, rank $r+1$ & $j$ & $\ell+1$\\
\bottomrule
\end{tabular}
\end{center}
Queries are merged only within one parent node and only when pool, depth,
lift count, and coordinate system agree; equal targets keep the largest
allowance and an attaining source. Since depth never decreases,
$\rho_{j'}\le\rho_j$ and $H_{j'}\le H_j$ for every descendant: no scale
grows. The scales are numbers; each child measures $Q$, $P$, and its
lattice in its own coordinates. Every bound above on the mass of a child
was proved by mapping its buffered points into the parent buffer region,
so it is independent of the child's coordinates.

\begin{lemma}[Depth bounds]\label{lem:depth-bounds}
On every root-to-node path: there are at most $\T_{\rm lift}$ lifts; the
stored rank is at most $D$; $\Psi$ never increases, so $\Psi\le\Psi_0$; and
there are at most $\lfloor\log_2\Psi_0\rfloor$ arithmetic children, so
$j\le\lfloor\log_2\Psi_0\rfloor<J$ at every node.
\end{lemma}

\begin{proof}
All reductions other than lifts do not increase $\Psi$
(Lemmas~\ref{lem:volume}, \ref{lem:boundary}, \ref{lem:thin}; for
same-rank thin children $c_{\rm thin}rK/H_j\le c_{\rm thin}DK/H_{\min}\le1$
by \eqref{eq:schedule}), an arithmetic child at least halves it, and a
lift multiplies it by at most $2(r+1)W/\nu\le N^{-\delta/2}$ as long as
$r\le D$, by \ref{S:large}. Consider a path up to a hypothetical
$(\T_{\rm lift}+1)$-st lift. Before it the rank is at most
$d+\T_{\rm lift}=D$, so all estimates apply, and after it
$\Psi\le N^{2d-(\T_{\rm lift}+1)\delta/2}<1$, contradicting
Lemma~\ref{lem:volume}. As $\Psi\ge1$, arithmetic children occur at most
$\log_2\Psi_0<J$ times.
\end{proof}

This is a prospective argument: before a first violation all constants
belong to the fixed ranks, and the accumulated volume reductions would
force $\Psi<1$.

\subsection{A linear potential}

Set
\begin{equation}\label{eq:potential-weights}
 \chi=12J(\Athin+\Alift+1).
\end{equation}
For a node $v$ of rank $r_v$, lift count $\ell_v$, depth $j_v$, mass $w_v$,
and normalized volume $\Psi_v$ put
\begin{equation}\label{eq:potential}
 B_v=\chi^{\,r_v+2(\T_{\rm lift}-\ell_v)}w_v\Psi_v,\qquad
 \Phi_v=(J-j_v)B_v .
\end{equation}
For fixed $d$ and $\delta$, the exponent is bounded and $\chi$ and $J$
are polylogarithmic in $N$.

\begin{proposition}[Linear work amortization]\label{prop:contraction}
At every node that creates children,
\begin{equation}\label{eq:potential-drop}
 \sum_{u\ \text{child of}\ v}\Phi_u\le\Phi_v-\tfrac34B_v .
\end{equation}
Consequently the recursion for one root query terminates, and
$\sum_vw_v\Psi_v\le\sum_vB_v\le\frac43\Phi_{\rm root}$.
\end{proposition}

\begin{proof}
Arithmetic children keep the exponent of $\chi$ and have depth $j_v+1$.
By \eqref{eq:arith-contraction} their $B$-weights sum to at most $B_v$,
so their $\Phi$-weights sum to at most $\Phi_v-B_v$.

Rank-losing and lift children decrease the exponent of $\chi$ by at
least one: a lift changes it by $1-2=-1$. Their normalized volumes
are at most $\Psi_v$, and their total mass is at most
$(\Athin+\Alift+1)w_v$
(Lemmas~\ref{lem:boundary}, \ref{lem:thin}, and
Proposition~\ref{prop:lift}). Thus their $B$-weights sum to at most
$B_v/(12J)$, and their $\Phi$-weights to at most $B_v/12$.
Same-rank thin children keep the exponent. Lemma~\ref{lem:thin} and
\eqref{eq:schedule} give total $B$-weight at most
$(\Athin DK/H_{\min})B_v\le B_v/(12J)$, and hence total
$\Phi$-weight at most $B_v/12$.
The three families therefore satisfy
\[
 \sum_{u\ \text{child of}\ v}\Phi_u
 \le\Phi_v-B_v+\tfrac16B_v
 \le\Phi_v-\tfrac34B_v.
\]

Every node has $w_v\ge1$, since only nonempty batches of filtered
queries are emitted and every such target lies in its own buffer
(Section~\ref{sec:buffers}; children whose queries are all discarded
are never created, and the filtering is charged to the parent), and
$\Psi_v\ge1$ by Lemma~\ref{lem:volume}. Hence
$\Phi_v\ge B_v\ge w_v\Psi_v\ge1$. Sum \eqref{eq:potential-drop} over any
finite truncation, using $\Phi_v\ge B_v$ at its leaves and frontier,
to obtain $\sum_vB_v\le\frac43\Phi_{\rm root}$. Since $B_v\ge1$, this
also proves termination.
\end{proof}

Thus $J-j$ pays for arithmetic steps, and each rank loss or lift
spends one power of $\chi$. The factors $J$ in
\eqref{eq:potential-weights} and \eqref{eq:schedule} allow the other
children to keep the multiplier $J-j$.

\subsection{Correctness}

\begin{proposition}[Exact strict decision]\label{prop:decision}
The recursion accepts a strict query on the kernel exactly when it has a
binary solution, and then returns an attaining labelled witness. It
correctly rejects fractionally feasible but integer infeasible queries.
\end{proposition}

\begin{proof}
\emph{Soundness.} Every finite table entry has an actual labelled
witness, its exact original cost, and its source; every success combines
disjoint actual pools and stays below its exact allowance; every
coordinate change retains all old equations.

\emph{Completeness.} By induction on the finite recursion tree
(Proposition~\ref{prop:contraction}): discarded queries have no strictly
feasible completion; forced face endpoints hold for all completions; the
thin window is safe by Lemmas~\ref{lem:deviation} and~\ref{lem:thin},
using the LP endpoints even at zero reduced cost; the deep windows are
safe by Lemma~\ref{lem:drift}; tables with dominance within a fixed
remaining pool preserve all continuations
(Lemmas~\ref{lem:transfer} and~\ref{lem:multisource});
Lemma~\ref{lem:high-low} either certifies that the anchor's queries have
no solution, finds a centre, or restores the full core and recomputes
the outside table; and Proposition~\ref{prop:lift} retains every integer
value of the new coordinate of every completion. A failed rounding
attempt is never used to reject. Hence every true query has a true child
unless a success has been returned. Predecessor, source, and reserve
records reconstruct the attaining items.
\end{proof}

\paragraph{From strict decisions to optimization.}
The kernel has $n_0$ items with costs in $[-c_{\max},c_{\max}]$. The query
$A_I\vy=\vb_0$, $\vc_I^{\T}\vy<n_0c_{\max}+1$ is equivalent to integer
feasibility, and it has positive fractional gap whenever the kernel LP
is feasible. If it is rejected, \eqref{eq:main} is infeasible by
Lemma~\ref{lem:kernel}. Otherwise binary search over
$[-n_0c_{\max},n_0c_{\max}]$ uses $O(1+\log(n_0(c_{\max}+1)))$ strict
decisions, testing $\vc^{\T}\vy\le q$ by the allowance $q+1$. An attaining
witness at the optimum is decoded to original variables: add upward
counts and subtract downward counts from $\bar\vx$, reinsert fixed and
zero-column variables, and drop slack variables. If the original LP
optimum is already integral it is returned at once.

\subsection{Proof of the main theorem}

\begin{proof}[Proof of Theorem~\ref{thm:main}]
Proposition~\ref{prop:ranged} and the preprocessing of
Section~\ref{sec:kernel} cost $O_d(n)$, the LP costs $T_{\rm LP}$, and
the kernel, with its sorted type profiles, costs
$O_d(n)+O_d(N^{2d}\log(N+2))$ further operations
(Lemma~\ref{lem:kernel}). Exactness follows from
Lemma~\ref{lem:kernel} and Proposition~\ref{prop:decision};
Lemma~\ref{lem:kernel-vertex} is used only for the running time (it
supplies the kernel LP vertex and price without a second LP).

Consider one root query. Its pool has rank at most $d$ and
$\Psi_0=C_dN^d$ (Lemma~\ref{lem:volume}); its buffer lies in a translate
of $H_0P_I$, so by Lemma~\ref{lem:types}(2) and
Lemma~\ref{lem:scales},
\begin{equation}\label{eq:root}
 w_0\le C_dH_0^d\Psi_0,\qquad
 w_0\Psi_0=\widetilde O_{d,\delta}\bigl(N^{2d+d\delta}\bigr).
\end{equation}
The root potential differs from $w_0\Psi_0$ by the factors $J$ and
$\chi^{r_0+2\T_{\rm lift}}$, where $r_0\le d$ is the root rank; these
factors are polylogarithmic.
Lemma~\ref{lem:local-work} and Proposition~\ref{prop:contraction}
bound the work for one strict decision by
\begin{equation}\label{eq:decision-time}
 \widetilde O_{d,\delta}\Bigl(\nu^2\sum_vw_v\Psi_v\Bigr)L_c
 =\widetilde O_{d,\delta}(\nu^2\Phi_{\rm root})L_c
 =\widetilde O_{d,\delta}\bigl(N^{2d+(d+2)\delta}\bigr)L_c .
\end{equation}
Recall that $\widetilde O_{d,\delta}$ hides only a factor
$C_{\rm dim}(\log(N+2))^{C_{\rm dim}}$.
Temporary stored and latent ranks only affect constants and logarithmic
degrees; their geometric cost is paid locally by deep packed mass, and no
global $\nu^D\Psi_0$ is charged at every node. By \eqref{eq:delta},
$(d+2)\delta\le\eps/2$, and
$C_{\rm dim}(\log(N+2))^{C_{\rm dim}+1}\le C'_{d,\eps}N^{\eps/2}$, so one
decision takes $O_{d,\eps}(N^{2d+\eps}L_c/\log(N+2))$ time. The number
of strict decisions is
$O(1+\log(n_0(c_{\max}+1)))\le C_d\log(N+2)\,L_c$, using
$n_0\le C_dN^{2d}$, and witness recovery fits in the stored computation.
All arguments requiring $N\ge N_{\rm thr}(d,\delta)$ were settled for
smaller $N$ after Lemma~\ref{lem:scales}, where the kernel has
$C_{\rm dim}$ items and is solved by enumeration in $C_{\rm dim}$
operations. Multiplying
gives $O_{d,\eps}(N^{2d+\eps}L_c^2)$ for the decisions, and the kernel
term $O_d(N^{2d}\log(N+2))$ is smaller, which proves
\eqref{eq:main-time}. Proposition~\ref{prop:ranged} covers inequality and
two-sided rows without increasing $d$.
\end{proof}

\begin{center}\small
\begin{tabular}{@{}ll@{}}
\toprule
Operation & Charge\\
\midrule
Signed kernel and type profiles & $T_{\rm LP}+O_d(n)+\widetilde O_d(N^{2d})$, once\\
Grouped transfer on an $M_{\rm w}$-point window & $\widetilde O(T_IM_{\rm w})\le\widetilde O_{d,\delta}(M_{\rm w}\Psi)$\\
Batch transfer & one multisource table per removed/remaining pair\\
Typed LP (value, price, witness) & $\widetilde O_{d,\delta}(T_I)$, no $c_{\max}$ dependence\\
Point-or-thin, $D_v$, reboxing a lift child & $\widetilde O_{d,\delta}(T_I)\,L_c$\\
Enumeration of an $M_{\rm w}$-point window & $\widetilde O_{d,\delta}((M_{\rm w}+1)T_I)\,L_c$\\
Anchor selection and classification & $\widetilde O_{d,\delta}(wT_I)\,L_c$\\
Canonical cores & $\widetilde O_{d,\delta}(\kappa J^2T_I\Psi)$, paid by deep packed mass\\
Structural call and lifting & $\widetilde O_{d,\delta}(\nu w_v)$, paid by deep packed mass\\
All recursive nodes & paid by $(J-j)\chi^{r+2(\T_{\rm lift}-\ell)}w\Psi$\\
One root decision & $\nu^2H_0^d\Psi_0^2L_c=\widetilde O_{d,\delta}(N^{2d+(d+2)\delta})\,L_c$\\
All decisions & $O(\log(N+2)\,L_c)$ decisions, total $O_{d,\eps}(N^{2d+\eps}L_c^2)$\\
\bottomrule
\end{tabular}
\end{center}

\subsection{Encoding lengths}\label{sec:bits}

The running time in Theorem~\ref{thm:main} separates the large original
data from the reduced combinatorial instance. An original bound or
right-hand side may require many bits; the algorithm must read them,
solve the initial LP, compute $\bar\vx$ and the offset $\vc^{\T}\bar\vx$,
and write the large original counts. Let $B_{\rm in}$ bound
$1+\log(n+d+2)+\log(N+2)+\log(c_{\max}+2)+\log(1+\|\vb\|_\infty)
+\log(1+\|\vell\|_\infty+\|\vu\|_\infty)$, including row endpoints for
ranged inputs. Box activities and offsets have $O_d(B_{\rm in})$ bits,
and so do the rounded LP coordinates, by Cramer bounds. With fast
integer arithmetic the preliminary and final arithmetic take
$\widetilde O_d(nB_{\rm in})$ bit operations, in addition to the LP.
After compression the explicit data are $n_0=O_d(N^{2d})$ items with
$\|\va_i\|_\infty\le N$, $|c_i|\le c_{\max}$, and $\|\vb_0\|_\infty\le dN$;
the inherited fractional parts have denominators at most $d!N^d$. All
later numbers have $b=O_{\rm dim}(\log(n_0N(c_{\max}+2)))$ bits, and the
large offset is never used as a search range. Fractions need not be
reduced, since their bit bounds come from fixed-size Cramer bounds, so
each operation on them is a fixed number of integer multiplications,
divisions, and comparisons, in $O(b\log b)$ bit operations with fast
integer arithmetic~\cite{HarveyvdH21}. As
$b\log b\le C_{\rm dim}(\log(N+2))^2L_c\log_2(L_c+1)$, a bit-model
reading of the theorem is
\[
 T^{\rm bit}_{\rm LP}+\widetilde O_d(nB_{\rm in})
 +O_{d,\eps}\bigl(N^{2d+\eps}L_c^3\log(L_c+1)\bigr),\qquad
 L_c=\log_2(c_{\max}+2).
\]

\begin{remark}[The logarithmic factor]\label{rem:log}
Only two parts of the algorithm depend on the magnitude of the costs.
The first is the binary search over the optimal value, with
$O(\log(n_0(c_{\max}+1)))$ strict decisions. The second is the number of
iterations of the rounding steps applied to bodies whose description
involves costs: the determinant-doubling loops in
Lemma~\ref{lem:point-thin} and in the construction of $D_v$, and the LLL
steps of Lemma~\ref{lem:point-thin} and of the reboxing of lift
children (Lemma~\ref{lem:boxing}); each performs
$O_{\rm dim}(\log(n_0N(c_{\max}+2)))$ iterations, which is the
logarithm of the ratio of the largest and smallest possible volumes. All
LPs are solved by the multidimensional search of
Lemma~\ref{lem:typed-lp}, whose operation count is independent of
$c_{\max}$; with the ellipsoid method~\cite{GLS88} each LP would instead
contribute a further power of the bit length depending on the
dimension, which is the source of the exponent $O_{d,\eps}(1)$ in a
cruder analysis. All searches in sorted cost lists (thresholds,
plateau sets $L_{\vpi}(\theta)$, exceptional items) are searches by rank,
not by value. The structural theorem, the canonical cores, the
removed-item tables, and all scales are independent of the costs.
Removing the second factor would require rounding steps whose number of
iterations does not grow with the bit length of the costs, for instance
by rounding normalized costs relative to the current gap; removing the
first would require a parametric version of the recursion. We do not
pursue either.
\end{remark}

\begin{remark}[The factor $N^\eps$]\label{rem:eps}
The reserve scale $\nu=\lceil N^\delta\rceil$ meets the positive-power
growth condition $\tau\ge W_{\rm loc}^{2\eta}$ of Theorem~\ref{thm:gap}
and yields fixed bounds on the stored rank, the latent GAP rank, and the
number of lifts. Replacing it by a polylogarithmic scale would let the
structural exponent and these dimensions grow with $N$, together with
their constants. Normalized volume alone cannot pay for unbounded rank:
a unimodular cross-polytope has $\Psi=1$ in every rank. Removing the
factor $N^\eps$ would require a uniform version of the reserve and lift
steps, or an analysis paying their full dimension-dependent cost.
\end{remark}

\appendix
%======================================================================
\section{Verification of the scale inequalities}\label{app:scales}

\begin{proof}[Proof of Lemma~\ref{lem:scales}]
We use the choices in Section~\ref{sec:params}, in particular
$c_{\rm thin}=1$ and $H_{\min}=12J(\Athin+1)DK$.
Since $D\le J$, $\Athin\ge1$, and $K/E=1/\alpha>J$, this choice
satisfies $H_{\min}\ge K$ and $H_{\min}\ge JE$.

\ref{S:large}: $\Psi_0=C_dN^d$ and $W,D$ depend only on $d,\delta$,
so the first two inequalities hold for sufficiently large $N$.
By Lemma~\ref{lem:boxing}, every application of
Theorem~\ref{thm:gap} under consideration has
$W_{\rm loc}\le C_{\rm dim}\tau^D\Psi_0$. Hence
\[
 W_{\rm loc}^{2\eta}
 \le C_{\rm dim}N^{2\eta(\delta D+d)}\log^C(N+2)
 <\nu\le\tau
\]
for large $N$, by \eqref{eq:eta}. Also
$\log W_{\rm loc}=O_{\rm dim}(\log N)$, so a sufficiently large
constant in \eqref{eq:tau} gives $h\ge\nu$. The permitted GAP ranks
are bounded in terms of $d,\delta$, and therefore $c_\eta h\ge k$
for large $N$. The inequality $\nu\le k_0$ is part of its definition.

\ref{S:sched}: $E/K=\alpha$ by definition. Since
$C_\gamma-1=2\gamma/(1-\gamma)\le7\alpha$,
$\Delta_\rho<5\alpha$ and
\[
 J\Delta_\rho
 <\frac{5J\rho_*}{1000(J+1)}<\rho_*/4.
\]
Thus $\rho_*/4\le\rho_j<\rho_*/2$ and
$\alpha/2=E/(2K)<\rho_*/4$. Finally
$H_{\min}\ge K\ge E$.

\ref{S:own}: The identity
$\rho_j-\rho_{j+1}=\Delta_\rho$, together with
$(C_\gamma-1)\rho_{j+1}\le(C_\gamma-1)\rho_*$,
gives the first inequality. For the second, use
$H_j\le H_{\min}+J\Delta_H$ and $\gamma J<1/2$ to obtain
\[
 \gamma H_j+E
 \le\gamma H_{\min}+\tfrac12\Delta_H+E
 =\Delta_H=H_j-H_{j+1}.
\]

For the remaining inequalities, recall
\[
 K=1000\cdot2^{10}(J+1)E,\qquad
 E=(260+6J)k_0,\qquad \kappa=3k_0.
\]
\ref{S:core}: $K/128=8000(J+1)E$, while
$\kappa(J+D)\le6(J+1)k_0$; since $E\ge260k_0$, the claimed strict
inequality follows.

\ref{S:collar}: We have
$\zeta\le3\alpha+1.01\rho_*<1/512$. Its left-hand side is therefore
at least $K/128-4k_0$, whereas its right-hand side is at most
$2(J+1)E$.

\ref{S:lift}: Since $1/128-\gamma>1/129$,
$\mu\ge K/129-3E\ge2k_0$. Also
\[
 \frac1{128}-\gamma-\frac{1+2\gamma}{254}>0.0038,
\]
so
\[
 \mu-k_0-\frac{R_{\rm p}+k_0}{254}
 \ge0.0038K-4E>0.
\]

\ref{S:drift}: Since $256\gamma<1$, the left-hand side is at most
$(130+3J)k_0=E/2$.

The growth bounds follow from $\nu\le2N^\delta$,
$J=O_{d,\delta}(\log N)$, and
$\Athin,\Alift\le C_{\rm dim}\alpha^{-C_{\rm dim}}$.
In particular,
$k_0=O_{\rm dim}(N^\delta\log^{c_0}N)$,
$\alpha^{-1}=O_{d,\delta}(\log N)$, and
$K=O_{\rm dim}(N^\delta\log^{c_0+2}N)$.
The exact choice $H_{\min}=12J(\Athin+1)DK$ gives
$H_{\min}=\widetilde O_{d,\delta}(\nu)$, and
\[
 H_0=H_{\min}+2JE+2\gamma JH_{\min}\le4H_{\min}.
\]
This proves $H_0=\widetilde O_{d,\delta}(\nu)$ and completes the
verification.
\end{proof}

\section{Typed fixed-row LP oracle}\label{app:typed-lp}

\begin{proof}[Proof of Lemma~\ref{lem:typed-lp}]
\emph{Dual and ordered profiles.}
Multiply the objective and every row by the common denominator $D$.
This scales the value by $D$, leaves the multipliers unchanged, and
makes the data integers of at most $2\lambda_0$ bits.

Free variables are boxed. Let $\bar M_{\rm f}=2^{C_{\rm f}\lambda_0}$ and
replace each free variable $z$ by $z=-\bar M_{\rm f}+z'$ with
$0\le z'\le2\bar M_{\rm f}$; this is a box variable not attached to items,
with integer capacity, so it forms a singleton group below, and its
value is recovered from the witness like any other variable. Update
every row's right-hand side under this substitution and retain the
objective constant
\[
 C_{\rm off}=-\bar M_{\rm f}\sum_{z\ {\rm free}}c_z,
\]
where $c_z$ is its coefficient in the scaled objective. Call the
resulting program, including this constant, LP$'$. Every polyhedron
defined by a fixed number of
rows of $O(\lambda_0)$-bit integer data and the variable bounds has, if
nonempty, a point whose free coordinates are given by a fixed-size
Cramer system (a point of a minimal face, with the coordinates
outside a maximal independent set of columns put at a bound or at
zero), hence of norm $2^{O_{\rm dim}(\lambda_0)}$. Consequently, if the
original LP is feasible and bounded, it has an optimal solution with
free coordinates of norm $2^{O_{\rm dim}(\lambda_0)}$, and its value
$v^*$ satisfies $|v^*|\le V_0=2^{C_0\lambda_0}$; LP$'$ has the same
value when $C_{\rm f}$ is large. If the original LP is feasible but
unbounded, it has a feasible point $\vx_0$ of norm
$2^{O(\lambda_0)}$ and an integral recession direction $\mathbf r$,
supported on free variables, with entries of $O(\lambda_0)$ bits and
$\vc\mathbf r\le-1$; then $\vx_0+t\mathbf r$ is feasible for LP$'$ for
$t=\bar M_{\rm f}2^{-O(\lambda_0)}$, so the value of LP$'$ is below
$-V_0$ when $C_{\rm f}$ is large. Hence the original LP is unbounded
exactly when LP$'$ is feasible with its full value, including
$C_{\rm off}$, less than $-V_0$; otherwise LP$'$ supplies its optimal
value and an optimal witness. All new numbers, including $C_{\rm off}$,
have $O_{\rm dim}(\lambda_0)$ bits. Choose $C_{\rm f}$ large enough that,
for a feasible bounded input, an optimal original representative has
every free coordinate strictly between $-\bar M_{\rm f}$ and
$\bar M_{\rm f}$. Below we solve the bounded program LP$'$; write
$\vb$ and $\vc$ for its updated right-hand side and variable objective
coefficients. Its fixed objective term is $C_{\rm off}$, taken as zero
when no free variables were boxed. The basic-witness guarantee for the
input LP will be used only in this latter case.
A \emph{group} consists of a type, a copy, and, when needed, one of the
at most two cost intervals defined by \eqref{eq:normalize}; a box
variable not attached to items forms a singleton group. There are
$O(T_I)$ groups, each in cost order, and its objective and row column
has the form $\mathbf g_G+c_i\mathbf h_G$.

Add artificial slack variables of both signs to every row, with cost
$\bar M=2^{C\lambda_0}$ for a sufficiently large $C=C_{\rm dim}$.
The augmented LP is feasible. If LP$'$ is feasible, its set of optimal
dual solutions is nonempty; it may have a
lineality space, for instance when the coupling rows are linearly
dependent, so basic dual solutions need not exist. Consider instead the
arrangement $\mathcal A$ of the hyperplanes $\{\ell_i=0\}$ and the sign
facets $\{y_j=0\}$ of the inequality rows; its hyperplanes have
$O(\lambda_0)$-bit integer data (columns and costs only). The dual
function is affine on the closure of every cell of $\mathcal A$, so the
set of optimal duals meets such a closure, if at all, in a face of it
(the maximizers of an affine function over a polyhedron); hence it
contains a minimal face of a cell closure of $\mathcal A$. Such a face is
an affine flat cut out by at most $q$ linearly independent equations
$\ell_i=0$ or $y_j=0$; adding coordinate equations $y_j=0$ for a
complementary set of coordinates gives a square nonsingular system
whose solution $\vy^\circ$ lies on the flat. By Cramer's rule
$\vy^\circ$ has $O_{\rm dim}(\lambda_0)$ bits, so its norm is
$2^{O_{\rm dim}(\lambda_0)}<\bar M$ for large $C$. This argument uses only
the hyperplane data, not the optimal value, and is unaffected by the
exponential number of linear pieces of the dual function. This bounded optimal dual is used below; an optimal
augmented solution therefore uses an artificial slack exactly when
LP$'$ is infeasible. Its dual maximizes
\[
 \Phi(\vy)=C_{\rm off}+\vy\vb+\sum_i u_i\min\{0,\ell_i(\vy)\},\qquad
 \ell_i=\nu_G+c_i\mu_G\quad(i\in G),
\]
over the polytope $Y$ of all $\vy$ with $\|\vy\|_\infty\le \bar M$ and
with the signs prescribed by the inequality rows; thus the facets of $Y$
are the box facets $|y_j|=\bar M$ and the sign facets $y_j=0$, and both
kinds are meant whenever facets of $Y$ are used below. Here $0\le x_i\le u_i$, with $u_i=1$
for an item, and $\nu_G,\mu_G$ are affine. At any fixed $\vy$, the
negative coefficients in a group form a prefix, a suffix, or the whole
group, according to the sign of $\mu_G(\vy)$. Searches by rank and
capacity and capacity-weighted cost prefix sums evaluate the function
and a subgradient in $O(T_I\log(n_0+2))$ operations.

\emph{Perturbation and standard parts.}
Perturb the objective coefficient in group $G$ by
$\epsilon'_G+c_i\epsilon''_G$, with independent infinitesimals and
$\epsilon''_G>0$, and perturb the right-hand side, at a lower order, by
$(\epsilon,\ldots,\epsilon^q)$
\cite{EdelsbrunnerMucke90}. The infinitesimals $\epsilon''_G$ of
distinct groups are independent, so this is not a single secondary
objective of total working cost across groups; within each group,
however, for a geometric objective ($\mu^0_G=0$), $\epsilon''_G>0$
makes the selected items a cheapest prefix of the cost order, and only this per-group property is used. Order the group infinitesimals by group
index and put $\epsilon$ last; comparisons are lexicographic.
Below, $\nu_G$ and $\mu_G$ include their additive perturbations;
$\nu_G^0,\mu_G^0$ denote their unperturbed versions.

Let $\Phi_\epsilon$ be the perturbed dual function. A basic perturbed
primal solution has the form
$\vx^\epsilon=B^{-1}(\vb_\epsilon-N\vx_N)$, with a standard basis
matrix and nonbasic values at standard bounds. Its standard part is
feasible for the unperturbed augmented LP$'$; a perturbed dual optimum lies
in the bounded standard set $Y$, so its standard part is dual feasible.
Taking standard parts in strong duality gives
\[
 C_{\rm off}+\vc\,\operatorname{st}\vx^\epsilon
 =\operatorname{st}\Phi_\epsilon(\vy^\epsilon)
 =\Phi(\operatorname{st}\vy^\epsilon).
\]
Weak duality then proves optimality of both standard parts.
We report these standard parts, not the perturbed solutions.

\emph{Recursive problems and genericity.}
Megiddo's multidimensional search
\cite{Megiddo84,Dyer86,Zemel84} will solve problems in dimension
$k\le q$ of the form
\[
 \max_{\vz\in P}\Upsilon(\vz),\qquad
 \Upsilon(\vz)=\beta+\vz\vb'+\sum_i u_i\min\{0,\nu_G(\vz)+c_i\mu_G(\vz)\}.
\]
The polytope $P$ has $O_q(1)$ inequalities with standard integer
normals; every group is an interval of its cost order, and all
gradients are standard. The $\epsilon$-part of $\vb'$ is
$(\epsilon,\ldots,\epsilon^q)\mathcal E$, where the standard integer matrix
$\mathcal E\in\Z^{q\times k}$ has rank $k$. Initially $P=Y$,
$\mathcal E=I$, and $\beta=C_{\rm off}$.
For every standard nonzero direction $\mathbf d$, a directional
derivative is a standard number plus
$(\epsilon,\ldots,\epsilon^q)\mathcal E\mathbf d$, and hence is nonzero.
Arrangement edges have standard directions, so the objective is
strictly monotone along them and has a unique maximizer $\vz^*$.

\emph{The hyperplane side test.}
Let $H=\{\vz:\mathbf n\vz=\eta\}$ have a standard integer normal.
If $H\cap P=\emptyset$, its side is read off at a vertex of $P$.
Otherwise write $\vz=\vz_0+Q\vw$ on $H$, where $Q$ is a standard
integer matrix of rank $k-1$ with $\mathbf nQ=0$.
Restriction gives a problem in dimension $k-1$, with coefficients
$\ell_i(\vz_0)+\nabla\ell_iQ\vw$ and full-rank matrix $\mathcal EQ$.
Solve it recursively and call its maximizer $\vz_H$.

The one-sided derivative at $\vz_H$ is
\[
 D\Upsilon(\vz_H;\mathbf d)=\mathbf s\mathbf d+
 \sum_{\ell_i(\vz_H)=0}u_i\min\{0,\nabla\ell_i\mathbf d\},\qquad
 \mathbf s=\vb'+\sum_{\ell_i(\vz_H)<0}u_i\nabla\ell_i.
\]
The negative items form a prefix or suffix in each group; its tied
items form one equal-cost run, all of the group, or none of it.
Thus all these sets and $\mathbf s$ are obtained by rank searches and
weighted prefix sums. Let $\mathcal K$ be the feasible-direction cone
at $\vz_H$, defined by the inequalities of $P$ tight there. For
$\sigma\in\{-1,1\}$ maximize
\[
 \psi_\sigma(\vw)=D\Upsilon(\vz_H;\sigma\mathbf n^{\T}+Q\vw)
 \quad\text{on}\quad
 P_\sigma=\{\vw:\sigma\mathbf n^{\T}+Q\vw\in\mathcal K,
                      \ \|\vw\|_\infty\le \bar M_1\},
 \qquad \bar M_1=2^{C\lambda_0}.
\]
This is another problem in dimension $k-1$: its groups are the tied
intervals, with the same capacities and coefficients
\[
 \nabla\nu_G(\sigma\mathbf n^{\T}+Q\vw)
 +c_i\nabla\mu_G(\sigma\mathbf n^{\T}+Q\vw).
\]
Its linear part is $\mathbf sQ$, with $\epsilon$-matrix $\mathcal EQ$, and its
constant is $\sigma\mathbf s\mathbf n^{\T}$. The domain has
$O_q(1)$ standard inequalities. Test its emptiness by enumerating the
$O_q(1)$ candidate vertices and checking their inequalities; this
uses only standard data. If it is empty put $m_\sigma=-\infty$.
Otherwise solve recursively and let $m_\sigma$ and $\vw_\sigma$ be
the maximum and its maximizer.

The box loses neither feasibility nor an optimal value of a nonempty
unboxed domain when $C$ is sufficiently large. A nonempty standard
polyhedron here has a feasible point of $O_{\rm dim}(\lambda_0)$
bits, by a fixed-size subsystem. Moreover its recession function is
\[
 \vw\longmapsto D\Upsilon(\vz_H;Q\vw)
 \quad\text{on}\quad\{\vw:Q\vw\in\mathcal K\}.
\]
It is nonpositive by maximality on $H$ and strictly negative in every
nonzero direction by genericity. The domain cut by its tied
breakpoint hyperplanes therefore has no lineality, and the supremum
is attained at an arrangement vertex. Its data and its coordinates
have $O_{\rm dim}(\lambda_0)$ bits, so the vertex lies in the box.
In particular $\vw_\sigma$ is standard; with
$\mathbf d_\sigma=\sigma\mathbf n^{\T}+Q\vw_\sigma$ we have
$\mathbf n\mathbf d_\sigma=\sigma\|\mathbf n\|^2$, and genericity
gives $m_\sigma=D\Upsilon(\vz_H;\mathbf d_\sigma)\ne0$.

If $m_+>0$ the maximizer lies on the positive side of $H$; if
$m_->0$ it lies on the negative side; otherwise it is $\vz_H$.
Indeed, a positive derivative gives a feasible point on its side
with value greater than $\Upsilon(\vz_H)$. Were the maximizer on the
opposite closed side, the segment joining these two points would
meet $H$ with a larger value, by concavity, a contradiction.
If both maxima are nonpositive, directions tangent to $H$ have
nonpositive derivative by maximality on $H$; every other feasible
direction scales to one of the two unboxed domains. Empty domains
correspond to sides with no such directions. Thus every feasible
direction has nonpositive derivative, proving global maximality of
$\vz_H$. This also covers the case in which both domains are empty.

\emph{Sparse symbolic comparisons.}
The algorithm never compares perturbed objective values. All normals
are standard because infinitesimals occur only as additive constants.
Every point at which a coefficient is tested is defined by a system
of fixed size whose constants involve $O_q(1)$ group infinitesimals:
pairing combines $O_q(1)$ hyperplanes, and at most $q$ flats are
restricted. Cramer's rule makes these points and base points affine
in those group infinitesimals and polynomial in $\epsilon$.
Evaluating $\ell_i$ or $\mu_G$ adds only the two symbols of its group.
Restriction preserves this invariant.

In a local derivative problem differentiation removes all group
infinitesimals. Its coefficients, domain, and arrangement vertices
are standard; only its linear part and constant contain $\epsilon$.
The vector $\mathbf s$ depends on the restriction maximizer only
through the signs of its coefficients. The sign of a nonempty
local maximum is evaluated from the displayed derivative formula at
the standard $\mathbf d_\sigma$; it therefore involves $\epsilon$
alone. Consequently every comparison uses $O_q(1)$ arithmetic
operations on polynomials with $O_q(1)$ monomials. The reported value
is evaluated only after taking standard parts, using the decided
signs and prefix sums.

\emph{Pruning and operation count.}
The weighted pairing and median step of multidimensional search
locates the maximizer relative to hyperplanes carrying a fixed
fraction of their total weight, using $O_q(1)$ side tests and
$O_q(1)$ work per offered hyperplane. Resolve first the $O_q(1)$
facets of the current domain. For each group keep an undecided
interval and the status of $\mu_G=0$. Give it weight twice its
interval length while the sign of $\mu_G$ is unknown, and its
interval length afterwards. Offer $\mu_G=0$ in the former case and
the median item's hyperplane in the latter. Resolving the first
fixes the sign of $\mu_G$; resolving the second fixes the signs of
at least half the interval by monotonicity. Equal-cost items share
their hyperplane. In either case the group weight at least halves.
If an offered hyperplane contains the maximizer, restrict to it;
this lowers the dimension at most $q$ times. A hyperplane containing
the current flat is decided immediately; if $\mu_G$ vanishes there,
all its items have the same coefficient and form one hyperplane.
Thus $O_q(\log(n_0+2))$ rounds suffice. Once all signs are decided
the objective is affine on the remaining region, whose unique
maximizer is located by its active hyperplanes.

Each side test has $O(T_I\log(n_0+2))$ setup work and at most three
calls in one lower dimension. Hence, with
$L=\log(n_0+2)$, the operation count satisfies
\[
 \mathcal T(q)\le C_qL\bigl(T_IL+\mathcal T(q-1)\bigr),\qquad
 \mathcal T(q)\le T_I(C_qL)^{2q}.
\]
This is within the stated $T_IL^{C_{\rm dim}}$ bound and is
independent of the cost magnitudes.

\emph{Bit bounds.}
A compared point or direction solves a fixed-size system whose
entries are combinations of $O_q(1)$ input numbers with
$O_q(\lambda_0)$-bit coefficients. Cramer's rule gives a common
denominator $\Delta$ and numerators of
$O_{\rm dim}(\lambda_0)$ bits, separately for each symbolic
monomial. Substituting into an affine coefficient preserves these
bounds. For sums over groups, aggregate in the original integer
coordinates before substituting a base point or direction:
\[
 \mathbf s=\vb+\sum_G(X_G\nabla\nu_G+C_G\nabla\mu_G),\qquad
 \Phi(\vy)=C_{\rm off}+\vy\mathbf s+
       \sum_G(X_G\nu_G(\mathbf0)+C_G\mu_G(\mathbf0)),
\]
where $X_G$ is total capacity and $C_G$ is the capacity-weighted
cost sum of the relevant prefix or suffix; tied terms are chosen
according to the region or derivative being evaluated. Standard
coefficients in these sums are integers bounded by
$O(n_0)2^{O_{\rm dim}(\lambda_0)}$, including partial sums: ordinary
groups contain $O(n_0)$ labelled items across the fixed number of
copies, while the $O_q(1)$ added singleton groups have capacities and
costs of $O_{\rm dim}(\lambda_0)$ bits. Since
$\lambda_0\ge\log(n_0+2)$, the aggregates have
$O_{\rm dim}(\lambda_0)$ bits. Restrictions are applied after this
aggregation, using their fixed-size maps and common denominators.
A directional derivative is an
aggregate of this kind dotted with a direction of denominator
$\Delta$, plus its $\epsilon$-part, so it has the same bit bound.
The offset $C_{\rm off}$ is also an $O_{\rm dim}(\lambda_0)$-bit
integer. Thus $\Phi(\operatorname{st}\vy^*)$ is an integer plus an
integer vector dotted with a point of denominator $\Delta$.
Undoing the initial scaling adds only the denominator $D$.
Thus all arithmetic has the bit bound claimed in the lemma.

\emph{Few ties and primal recovery.}
For a group with $\mu_G(\vy^*)\ne0$, tied items have one common
cost and identical columns. Otherwise the entire group is tied
exactly when $\nu_G(\vy^*)=0$. These cases give, respectively,
one equation
\[
 \nu_G^0(\vy)+\epsilon'_G
       +c(\mu_G^0(\vy)+\epsilon''_G)=0
\]
or the pair $\mu_G^0(\vy)+\epsilon''_G=0$,
$\nu_G^0(\vy)+\epsilon'_G=0$, together with active facet
equations of $Y$ (box facets $|y_j|=\bar M$ and sign facets $y_j=0$). Distinct groups have independent infinitesimals,
and facet equations involve $\vy$ only. More than $q$ such
equations at one point would impose a nontrivial rational linear
relation on the infinitesimals. Thus at most $q$ groups contain
tied variables.

Complementary slackness fixes $x_i=u_i$ when $\ell_i(\vy^*)<0$
and $x_i=0$ when $\ell_i(\vy^*)>0$. Tied variables remain in
$[0,u_i]$; ordinary slacks occur only at zero multipliers and
artificial slacks only at facets $|y_j|=\bar M$. In a group of the
first kind the tied variables enter through
$X_G\in[0,\sum_{\rm tied}u_i]$, with one column. In a group of
the second kind they enter through
$X_G\mathbf g_G+C_G\mathbf h_G$, where
\[
 (X_G,C_G)\in Z_G=
 \left\{\left(\sum_{i\in G}x_i,\sum_{i\in G}c_ix_i\right):
                   \vx\in\prod_{i\in G}[0,u_i]\right\}.
\]
The lower and upper boundaries of $Z_G$ are the convex cheapest
and concave most expensive capacity-weighted prefix profiles.
The resulting feasible system has $O(q)$ unknowns, $q$ equations,
and two ordered piecewise-linear constraints per second-kind group;
the same search over breakpoints solves it in
$(C_q\log(n_0+2))^{O(q)}$ operations. In a first-kind group fill
variables in order to realize its total. In a second-kind group
slide a window of mass $X_G$ along the cost order, measured in
capacities. Its cost increases continuously from the cheapest to
the most expensive value, so a rank search realizes $C_G$ with at
most two variables strictly between their bounds. Moving along
kernel directions of the resulting fixed-size system, stopping
at $0$ or $u_i$, preserves feasibility and optimality and leaves
at most $q$ interior variables with independent columns.
This yields a basic optimum $\vx^\epsilon$.

Its standard part is a basic optimal solution of the unperturbed
augmented LP$'$ with the same basis. If LP$'$ is feasible, its optimal
dual $\vy^\circ$ of norm less than $\bar M$ is also optimal for the
augmented LP$'$ and gives every artificial slack positive reduced cost.
The witness therefore has no artificial slack and is basic for LP$'$.
If LP$'$ is infeasible every augmented solution uses an artificial slack.

For a feasible bounded input, undoing the free-variable shifts returns
an optimal witness for the original LP. If there were no free
variables, LP$'$ is the input LP and this witness is basic; otherwise
no basicness claim is needed after the artificial bounds are removed.
The reported dual also supports the original LP: pair it with the
optimal original representative chosen strictly inside all artificial
free-variable bounds. This representative is optimal for LP$'$, so
complementary slackness forces zero reduced coefficient on every
originally free variable. The dual therefore satisfies the original
free-variable equalities; the right-hand-side shifts and $C_{\rm off}$
cancel in its value. Finally, divide the full value
$\Phi(\operatorname{st}\vy^*)$, including $C_{\rm off}$, by $D$ to undo
the initial scaling. This supplies the stated primal and dual witnesses,
value, and feasibility and unboundedness tests with the claimed bounds.
\end{proof}

\section{Why actual pool identity matters}\label{app:pool}

Cost dominance at a fixed weight is valid only when the remaining
labelled pool is the same. Consider labelled scalar weights
$1,2,3,4,5$ with costs $2,20,40,50,69$. Two outside choices of weight $5$
are $\{1,4\}$ of cost $52$ and $\{2,3\}$ of cost $60$. If the first leaves
$\{2,3,5\}$ and the second leaves $\{1,4,5\}$, they cannot be merged by
keeping cost $52$: residual weight $4$ is available only in the second
remaining pool, although both remaining pools generate the same lattice.
Every transfer in this paper fixes its removed and remaining pools
before merging sources (Lemma~\ref{lem:multisource}), and arithmetic
restoration uses the full pool $R=I\cap\Gamma$, including expensive
eligible items, and recomputes its outside table
(Lemma~\ref{lem:high-low}).

\end{document}